\documentclass[11pt, letterpaper]{article}
\usepackage[left=1in, right=1in, top=1in, bottom=1in]{geometry}
\usepackage{graphicx}
\usepackage{amsmath}
\usepackage{amssymb}
\usepackage{amsthm}
\usepackage{enumitem}
\usepackage{calc}
\usepackage{setspace}
\usepackage[d]{esvect}
\usepackage{url}
\usepackage{hyperref}
\usepackage{tikz}
\usetikzlibrary{calc, positioning, fit, shapes.misc, shapes.geometric, decorations.pathmorphing}
\usepackage[vcentermath]{youngtab}
\usepackage[english]{babel}
\usepackage[utf8]{inputenc}
\usepackage[T1]{fontenc}
\usepackage{ytableau}
\usepackage{mathrsfs}
\usepackage{caption}
\usepackage{float}
\restylefloat{figure}
\usepackage{verbatim}

\newcommand{\A}{\ensuremath{\overline{1}}}
\newcommand{\B}{\ensuremath{\overline{2}}}

\newcommand{\VPs}{\ensuremath{\textup{\textsf{VP}}_{\textsf{s}}}}
\newcommand{\VNP}{\textup{\textsf{VNP}}}
\newcommand{\NP}{\textup{\textsf{NP}}}
\renewcommand{\P}{\textup{\textsf{P}}}
\newcommand{\dc}{\textup{\textsf{dc}}}
\newcommand{\dcbar}{\underline{\textup{\textsf{dc}}}}
\newcommand{\pc}{\textup{\textsf{pc}}}
\newcommand{\sh}{\textup{\textsf{sh}}}

\newcommand{\size}{\textup{\textsf{size}}}
\newcommand{\col}{\textup{\textsf{col}}}
\newcommand{\leftpart}{\textup{\textsf{leftpart}}}
\newcommand{\rightpart}{\textup{\textsf{rightpart}}}
\newcommand{\barred}{\textup{\textsf{barred}}}
\newcommand{\kbar}{\textup{\textsf{k-bar}}}
\newcommand{\la}{\lambda}
\newcommand{\HWV}{\textup{HWV}}
\newcommand{\sV}{\mathscr{V}}
\newcommand{\sW}{\mathscr{W}}
\newcommand{\Sym}{\textup{Poly}}
\newcommand{\End}{\textup{End}}

\newcommand{\sgn}{\textup{sgn}}
\newcommand{\per}{\textup{per}}
\renewcommand{\det}{\textup{det}}
\newcommand{\mult}{\textup{mult}}
\newcommand{\stab}{\text{stab\,}}
\newcommand{\aS}{\mathfrak{S}}
\newcommand{\GL}{\textup{GL}}
\newcommand{\SL}{\textup{SL}}
\newcommand{\IC}{\mathbb{C}}
\newcommand{\IN}{\mathbb{N}}
\newcommand{\IZ}{\mathbb{Z}}
\newcommand{\IS}{\mathbb{S}}

\newcommand{\tensor}{\textstyle\bigotimes}

\swapnumbers
\numberwithin{equation}{section}
\newtheorem{proposition}[equation]{Proposition}
\newtheorem{claim}[equation]{Claim}
\newtheorem{lemma}[equation]{Lemma}

\newtheorem{corollary}[equation]{Corollary}
\newtheorem{theorem}[equation]{Theorem}
\theoremstyle{definition}
\newtheorem{definition}[equation]{Definition}

\title{Implementing geometric complexity theory: On the separation of orbit closures via symmetries}
\author{Christian Ikenmeyer\thanks{University of Liverpool, christian.ikenmeyer$@$liverpool.ac.uk} \ and Umangathan Kandasamy\thanks{Universit\"at des Saarlandes, umangathan.kandasamy@outlook.de}}
\date{November 10, 2019}

\begin{document}
\raggedbottom
\maketitle
\thispagestyle{empty}

\begin{abstract}
Understanding the difference between group orbits and their closures is a key difficulty in geometric complexity theory (GCT):
While the GCT program is set up to separate certain orbit closures,
many beautiful mathematical properties are only known for the group orbits, in particular close relations with symmetry groups and invariant spaces, while the orbit closures seem much more difficult to understand.
However, in order to prove lower bounds in algebraic complexity theory, considering group orbits is not enough.

In this paper we tighten the relationship between the orbit of the power sum polynomial and its closure,
so that we can separate this orbit closure from the orbit closure of the product of variables
by just considering the symmetry groups of both polynomials and their representation theoretic decomposition coefficients.
In a natural way our construction yields a multiplicity obstruction that is neither an occurrence obstruction, nor a so-called vanishing ideal occurrence obstruction. All multiplicity obstructions so far have been of one of these two types.

Our paper is the first implementation of the ambitious approach that was originally suggested in the first papers on geometric complexity theory by Mulmuley and Sohoni (SIAM J Comput 2001, 2008):
Before our paper, all existence proofs of obstructions only took into account the symmetry group of one of the two polynomials (or tensors) that were to be separated.
In our paper the multiplicity obstruction is obtained by comparing the representation theoretic decomposition coefficients of both symmetry groups.

Our proof uses a semi-explicit description of the coordinate ring of the orbit closure of the power sum polynomial in terms of Young tableaux, which enables its comparison to the coordinate ring of the orbit.
\end{abstract}

\bigskip

{
\footnotesize
\begin{spacing}{0.9}
\noindent\textbf{Acknowledgements: }
This paper was written partially when CI was at the Max Planck Institute for Informatics, at the Max Planck Institute for Software Systems, and at the University of Liverpool.
This paper contains results that are present in UK's master's thesis at the Universit\"at des Saarlandes.
\end{spacing}
}

\medskip

{
\footnotesize
\begin{spacing}{0.9}
\noindent\textbf{2012 ACM CCS:} Theory of computation $\rightarrow$ Algebraic complexity theory
\end{spacing}
}

\medskip

{
\footnotesize
\begin{spacing}{0.9}
\end{spacing}
}



\section{Motivation: Geometric complexity theory}
\paragraph{Symmetries}
The idea of using the symmetries of the determinant $\det_n := \sum_{\pi \in \aS_n}\sgn(\pi)\prod_{i=1}^n x_{i,\pi(i)}$ and the permanent
$\per_m := \sum_{\pi \in \aS_m} \prod_{i=1}^m x_{i,\pi(i)}$
to separate algebraic complexity classes was pioneered by Mulmuley and Sohoni in 2001 \cite{MS:01}.
This approach is based on the observation that $\det_n$ and $\per_m$ are both characterized by their respective symmetry groups.
For example, consider homogeneous degree $n$ polynomials in $n^2$ variables $x_{1,1},\ldots,x_{n,n}$.
Let $X$ denote the $n \times n$ matrix whose entry in row $i$ and column $j$ is $x_{i,j}$. Then $\det(X)=\det_n$.
Now, for matrices $A,B \in \SL_n(\IC)$ the entries of the matrix $AXB$ are homogeneous linear polynomials in the $n^2$ variables.
The crucial fact is that every homogeneous degree $n$ polynomial $q$ in $n^2$ variables that satisfies $q(AXB)=q(X)$ equals $\alpha \cdot \det_n$ for some scalar $\alpha \in \IC$.
This means that $\det_n$ is \emph{characterized by its symmetries}. For the permanent polynomial, an analogous statement holds, and also for many other structurally simpler polynomials, for example for the power sum polynomial $x_1^D+\cdots+x_m^D$ and for the product of variables $x_1 x_2 \cdots x_D$, see \cite{Ike:19}.

\paragraph{Algebraic complexity theory}
An \emph{affine projection} of a polynomial is its evaluation at a point whose coordinates are given by affine linear polynomials, e.g., $(x_1+x_2+1)^2 = x_1^2 + 2 x_1 x_2 + x_2^2 + 2 x_1 + 2 x_2 + 1$ is an affine projection of $x_1^2$.
Kayal proved that it is $\NP$-hard to decide whether a polynomial is an affine projection of another polynomial \cite{Kay:12}.
Valiant proved \cite{Val:79b} that every polynomial $p$ is an affine projection of some $\det_n$ for $n$ large enough. The smallest $n$ for which this is possible is called the \emph{determinantal complexity} $\dc(p)$.
The class of sequences of polynomials $(p_m)$ whose sequence of natural numbers $\dc(p_m)$ is polynomially bounded is called $\VPs$.
For the permanent we can define the \emph{permanental complexity} $\pc(p)$ in a completely analogous manner: $\pc(p)$ is the smallest $n$ such that $p$ is an affine projection of $\per_n$.
The class of sequences of polynomials $(p_m)$ whose $\pc(p_m)$ is polynomially bounded is called $\VNP$.
Since $\pc(\det_n)$ is polynomially bounded, $\VPs \subseteq \VNP$.
Valiant's flagship conjecture in algebraic complexity theory, which is also known as the \emph{determinant versus permament} conjecture can be succinctly phrased as $\VPs \neq \VNP$.
This is equivalent to conjecturing that $\dc(\per_m)$ grows superpolynomially fast.

\paragraph{Homogeneous projections and endomorphism orbits}
It will be beneficial to phrase Valiant's conjecture in a homogeneous setting: A \emph{homogeneous projection} of a homogeneous polynomial (i.e., all monomials have the same total degree) is its evaluation at a point whose coordinates are given by homogeneous linear polynomials. The set of all homogeneous projections of $\det_n$ to polynomials in the variables $x_1,\ldots,x_N$ can then be written as $\{\det_n(\ell_1,\ldots,\ell_{n^2}) \mid \text{$\ell_i$ is a homogeneous linear polynomial in $x_1,\ldots,x_N$}\}$. Note that we put the $n \times n = n^2$ inputs of the determinant in a linear order.
The polynomial function $(x_{1,1},\ldots,x_{n,n}) \mapsto \det_n(\ell_1,\ldots,\ell_{n^2})$ equals the composition $\det_n \circ A$, where $A$ is the linear map $(x_{1,1},\ldots,x_{n,n}) \mapsto (\ell_1,\ldots,\ell_{n^2})$.
As it is common in representation theory, we write $A \cdot \det_n$ or just $A \det_n$ for $\det_n \circ A$.
The \emph{endomorphism orbit} $\End_{n^2}\det_n$ is defined as $\{A \det_n \mid A \in \IC^{n^2 \times n^2}\}$,
which is the set of all homogeneous projections of $\det_n$ to polynomials in at most $n^2$ variables.
Since all polynomials in $A \det_n$ are homogeneous of degree $n$, we have $\per_m \notin \End_{n^2} \det_n$ for any $m \neq n$.
This slight technicality is treated by a procedure called \emph{padding}:
For fixed $m$, $n$ with $m < n$, define the \emph{padded permanent} $\per_{m,n} := (x_{n,n})^{n-m}\cdot \per_m$.
Let $\dc'(\per_{m,n})$ denote the smallest $n$ such that $\per_{m,n} \in \End_{n^2} \det_n$.
A short calculation shows that Valiant's conjecture is equivalent to the conjecture that $\dc'(\per_m)$ grows superpolynomially.

\paragraph{Group orbits}
It turns out that if we restrict $\End_{n^2} \det_n$ to only the points $A\det_n$ for which $A$ is invertible, we get the much simpler \emph{group orbit} $\GL_{n^2}\det_n := \{g \det_n \mid g \in \GL_{n^2}\} \subseteq \End_{n^2} \det_n$.
The group orbit of the determinant consists of ``determinants in disguise'', i.e., determinants after a base change.
The question whether a polynomial $p$ lies in $\GL_{n^2}\det_n$ can be answered in randomized polynomial time \cite{Kay:12}.
Finding $g \in \GL_{n^2}$ such that $p = g \det_n$ is called the \emph{reconstruction} problem, which is the focus of several recent papers, where the determinant is replaced with other algebraic computational models, see e.g.\
\cite{GKQ:14}, \cite{KS:19}, \cite{KNS:19}.

\paragraph{Representation theory}
Let $H_{\det_n}\subseteq \GL_{n^2}$ denote the symmetry group of $\det_n$.
From the viewpoint of algebraic geometry, the set $\GL_{n^2}\det_n$ is an affine variety (\cite[Sec 4.2]{BLMW:11}, \cite[Cor., p.~206]{Mat:60}) and a homogeneous space that is isomorphic to the quotient $\GL_{n^2}/H_{\det_n}$.
It is crucial to note here that the group $H_{\det_n}$ does not carry any information about the fact that we study a space of polynomials!
In this way, we study the orbit $\GL_{n^2}\det_n$ independently of its embedding into the space of polynomials.
This gives a particularly beautiful description of its coordinate ring via invariant theory:
\[
\IC[\GL_{n^2}\det_n] = \IC[\GL_{n^2}]^{H_{\det_n}},
\]
where $\IC[\GL_{n^2}]^{H_{det_n}}$ is the set of regular $H_{\det_n}$-invariant functions on the variety $\GL_{n^2}$ \cite[(5.2.6)]{BLMW:11}.
The coordinate ring of $\GL_{n^2}$ has classically been studied: It is a $\GL_{n^2} \times \GL_{n^2}$-representation whose representation theoretic decomposition is multiplicity free:
\begin{equation}\label{eq:algpeterweyl}
\IC[\GL_{n^2}] = \bigoplus_{\la}\{\la\} \otimes \{\la^*\},
\end{equation}
where $\la$ runs over all nondecreasing lists of $n^2$ integers and $\{\la\}$ denotes the irreducible rational representation of $\GL_{n^2}$ corresponding to $\la$ (see e.g.\ Section~\ref{sec:prelrepr} or the textbooks \cite{FH:91, MaD:95, fult:97} for more information).
Eq.~\eqref{eq:algpeterweyl} is known as the algebraic Peter-Weyl theorem.
It implies that the multiplicity of $\la^*$ in $\IC[\GL_{n^2}\det_n]$ equals the dimension of the $H_{\det_n}$-invariant space $\{\la\}^{H_{\det_n}}$.
This coefficient $\dim \{\la\}^{H_{\det_n}}$ is known as the symmetric rectangular Kronecker coefficient and has been the object of several papers \cite{BOR:09, Man:11, IP:17}.
Even though Kronecker coefficients are mysterious, many properties are known (see e.g.~\cite{CDW:12, Bla:12, SS:16, IMW:17, Liu:17, BCMW:17} for some recent advances), and character theory is available to compute their values in many cases, see e.g.\ \cite[Appendix]{ike:12b}.
A better understanding of Kronecker coefficients could lead to a better understanding of $\GL_{n^2}\det_n$, which could eventually help to separate points from $\End_{n^2}\det_n$.

If we replace $\det_n$ by other polynomials, we get an analogous theory that is often equally beautiful. The power sum polynomial and the product $x_1 \cdots x_D$ will be of particular interest in this paper. The corresponding coefficients are not Kronecker coefficients, but plethysm coefficients and related coefficients that appear in algebraic combinatorics, see Proposition~\ref{pro:gctandsymmetries} for the power sum and \eqref{eq:reverseplethysm} for the product of variables.

\paragraph{Closures}
As we just have seen, group orbits have several desirable properties and can be understood directly via symmetry groups and algebraic combinatorics.
But endomorphism orbits do not behave that nicely.
In general, $\End_{n^2}\det_n$ is not a variety.
In order to enable the study of $\End_{n^2}\det_n$ with methods from algebraic geometry, we go to the closure (Euclidean closure and Zariski closure coincide here by general principles, see \cite[\S2.C]{Mum:95}):
$\overline{\End_{n^2}\det_n}$, which coincides with the group orbit closure
$\overline{\GL_{n^2}\det_n}$, see e.g.~\cite[Sec.~3.5]{BI:18}.
Hence we have a chain of inclusions $\GL_{n^2}\det_n \subseteq \End_{n^2}\det_n \subseteq \overline{\GL_{n^2}\det_n}$.
The border determinantal complexity $\dcbar(\per_m)$ is defined as the smallest $n$ such that $\per_{m,n} \in \overline{\GL_{n^2}\det_n}$.
Mulmuley and Sohoni's conjecture (closely related to B\"urgisser's conjecture \cite[hypothesis~(7)]{Bue:01})
is that $\dcbar(\per_m)$ grows superpolynomially.
Since $\dcbar(\per_m) \leq \dc'(\per_m)$, this would imply Valiant's conjecture.
Mulmuley and Sohoni's conjecture can be attacked by representation theoretic multiplicities (see Section~\ref{sec:preliminaries}) as follows.
If we assume for the sake of contradiction that $\overline{\GL_{n^2}\per_{m,n}} \subseteq \overline{\GL_{n^2}\det_n}$,
then by Schur's lemma (see \cite[Lemma 1.7]{FH:91}) the multiplicities satisfy $\mult_\la\IC[\overline{\GL_{n^2}\per_{m,n}}] \leq \mult_\la\IC[\overline{\GL_{n^2}\det_n}]$.
Thus, if there exists $(\la,d)$ that satisfies
\begin{equation}\label{eq:ineq}
\mult_\la\IC[\overline{\GL_{n^2}\per_{m,n}}]_d > \mult_\la\IC[\overline{\GL_{n^2}\det_n}]_d,
\end{equation}
then $\dcbar(m)>n$.
Such a pair $(\la,d)$ is called a \emph{multiplicity obstruction}.

\paragraph{Orbits vs orbit closures}
The algebraic geometry of $\overline{\GL_{n^2}\det_n}$ and the representation theory of its coordinate ring are rather difficult to understand, see e.g.~\cite{Kum:15, BHI:17}.
But the close relationship between orbit and orbit closure gives hope that results can be transferred from the orbit to the closure.
Indeed, $\IC[\overline{\GL_{n^2}\det_n}] \subseteq \IC[\GL_{n^2}\det_n]$ is a subalgebra, and hence we have $\mult_\la\IC[\GL_{n^2} \det_n] \geq \mult_\la\IC[\overline{\GL_{n^2} \det_n}]$.
Getting lower bounds on multiplicities in $\IC[\overline{\GL_{n^2} \det_n}]$ seems much harder. For example, the result in \cite{Kum:15} holds for those $n$ for which the Alon-Tarsi property holds, in particular if $n$ is an odd prime number $\pm 1$, see Section~\ref{sec:notoccobsorvanidoccobs}.
The occurrence results in \cite{BIP:19} use explicit constructions using Young symmetrizers, which is a laborious process.
But as a first step towards lower bounds on $\mult_\la\IC[\overline{\GL_{n^2} \det_n}]$,
\cite{BI:17} proved that $\GL_{n^2} \det_n$ is open in its closure and that the ring $\IC[\GL_{n^2} \det_n]$ is a localization of $\IC[\overline{\GL_{n^2} \det_n}]$.

\section{Our contribution}
In this paper we tighten the results from \cite{BI:17} in the case of the power sum polynomial.
For $m \geq D$ let $p := x_1^D + x_2^D + \cdots + x_m^D$ and let $q := x_1 x_2 \cdots x_D$. Let $G := \GL_{m}$.
For $m=D$ we separate the two families of orbit closures $\overline {Gp} \not\subseteq \overline {Gq}$ of polynomials $p$ and $q$ using multiplicity obstructions~$\la$, i.e., $\mult_\la\IC[\overline{Gp}]>\mult_\la\IC[\overline{Gq}]$.
Our key contribution is a proof method that for the first time implements closely the strategy in \cite{MS:01,MS:08}:
Both the lower bound on $\mult_\la\IC[\overline{Gp}]$ and the upper bound on $\mult_\la\IC[\overline{Gq}]$ are obtained directly from the symmetry groups of $p$ and~$q$ and the dimension of the spaces of $H_p$- and $H_q$-invariants in irreducible $\GL_m$-representations.
This is the result of our tightening of the relationship between 
$\mult_\la\IC[Gp]$ and $\mult_\la\IC[\overline{Gp}]$, see Theorem~\ref{thm:proofofconcept}.

Before our paper, all existence proofs of multiplicity obstructions $\overline{Gp} \not\subseteq \overline{Gq}$ for any $p$ and $q$ required to explicitly construct (with multilinear algebra) copies of irreducible representations in $\mult_\la\IC[\overline{Gp}]$. These papers only took into account the symmetry group of $q$ instead of both symmetry groups, see \cite{BI:11, BI:13, GIP:17, DIP:19}.

In particular, we prove $\overline{Gp}\not\subseteq\overline{Gq}$ by explicitly constructing a multiplicity obstruction $\la$ in Theorem~\ref{thm:proofofconcept} using the symmetry groups of $p$ and $q$ and their representation theoretic decomposition coefficients,
but we do \emph{not} construct an explicit function that separates the two orbit closures! Since our obstruction is neither an occurrence obstruction, nor a vanishing ideal occurrence obstruction (see next paragraph), the separating function is quite nontrivial to recover.
This is a step in the right direction,
since the explicit construction of separating functions for Valiant's conjecture could turn out to be problematic because of the algebraic natural proofs barrier \cite{FSV:17,GKSS:17,BIJL:18}.

\subsection*{Occurrence obstructions and vanishing ideal occurrence obstructions}
The classical approach of Mulmuley and Sohoni \cite{MS:01, MS:08} conjectures the existence of so-called \emph{occurrence obstructions},
which are types $\la$ for which the stronger property $\mult_\la\IC[\overline{G\per_{m,n}}] > 0 = \mult_\la\IC[\overline{G\det_n}]$ holds.
These obstructions are not enough to prove strong complexity lower bounds (at least not in the classical setting of $\det_n$ vs $\per_{m,n}$), see \cite{IP:17, BIP:19}.
Recently it was shown that there are settings in which multiplicity obstructions are provably stronger than just occurrence obstructions \cite{DIP:19}.
The types $\la$ that are used in \cite{DIP:19} occur in the vanishing ideal of one orbit closure, but not in the vanishing ideal of the other,
hence we call them \emph{vanishing ideal occurrence obstructions}.
How useful vanishing ideal occurrence obstructions are for separating orbit closures is an open question, but it seems unlikely that strong complexity lower bounds can be proved by using only vanishing ideal occurrence obstructions.

The techniques that we develop in this paper study multiplicities and go beyond just occurrence obstructions and vanishing ideal occurrence obstructions.
Theorem~\ref{thm:proofofconcept} gives the first family of multiplicity obstructions that are neither occurrence obstructions nor vanishing ideal occurrence obstructions, see Proposition~\ref{pro:alontarsi}.
To prove this fact, we make use of Drisko's and Glynn's progress on the Alon-Tarsi conjecture about the difference between the number of even and odd Latin squares of a given size.

\subsection*{The toy setting as a starting point}
Our separation $\overline{Gp}\not\subseteq\overline{Gq}$
with $p=x_1^D+\cdots+x_m^D$ and $q=x_1\cdots x_D$
is clearly a toy problem,
but even though its complexity theoretic relevance is quite limited (multivariate factorization of a power sum) it shares all (as far as we know) crucial geometric and representation theoretic features with the determinant versus permanent problem:
Both problems are problems about orbit closures of polynomials, and the group action is the same canonical action. The only difference between the two setups are the specific polynomials $p$ and $q$.
Even though $p$ and $q$ do not share the complexity theoretic properties of the determinant and the permanent, $p$ and $q$ are characterized by their symmetry groups and are stable points (see \cite{BI:17}).
Therefore this setup can be seen a starting point from which $p$ and $q$ could now be gradually adjusted until some orbit closure separations can be obtained
that give lower bounds in algebraic complexity theory.

\subsection*{Structure of the paper}
In Section~\ref{sec:preliminaries} we start with preliminaries that are necessary to state our results precisely in Section~\ref{sec:resultdetails}.
The main connection between representation theory and tableau combinatorics is discussed in Sections~\ref{sec:prelrepr}--\ref{sec:hworbitclosure}.
Section~\ref{sec:threetableaux} proves a technical result about plethysm coefficients that was postponed from Section~\ref{sec:resultdetails}.
Section~\ref{sec:proofmain} proves the main technical theorem under the assumption that the so-called \emph{Tableau Lifting Theorem} \ref{thm:prolongation} is true. The rest of the paper (Sections~\ref{EVENsec:hypergraphsevenD} to~\ref{sec:psipropertiesDodd}) is then used to prove the Tableau Lifting Theorem using elementary but subtle Young tableau combinatorics.
Even and odd degrees $D$ are treated mostly independently, where the odd degree case is much more involved.

\section{Preliminaries}\label{sec:preliminaries}
A \emph{partition} $\la$ is a nonincreasing finite sequence of natural numbers $(\la_1,\la_2,\ldots)$.
We identify a partition with its \emph{Young diagram}, which is a top-left aligned array of boxes with $\la_i$ boxes in row~$i$.
For example, the Young diagram for $\la=(4,3,1)$ is $\tiny\yng(4,3,1)$.
The \emph{length} $\ell(\la)$ is the number of rows in the Young diagram of $\la$, formally $\ell(\la)=\max\{i \mid \la_i>0\}$.
The \emph{number of boxes} of $\la$ is defined as $|\la|:=\sum_{i\in\IN} \la_i$. We also define $|\varrho|:=\sum_{i\in\IN} \varrho_i$ in the case where $\varrho$ is not a partition.
The \emph{transpose} $\la^t$ of a partition $\la$ is the partition corresponding to the Young diagram that is the reflection of $\la$ at the main diagonal,
e.g., $(4,3,1)^t = (3,2,2,1)$. The entries of $\la^t$ are the column lengths of $\la$.
For a partition of length $\leq m$ and $\delta$ many boxes, we write $\la \vdash_m \delta$.
For two partitions $\la$ and $\mu$ we define their sum $\la+\mu$ in a row-wise fashion: $(\la+\mu)_i:=\la_i+\mu_i$.
For natural numbers $a,b$ let $a \times b$ denote the partition that corresponds to the rectangular Young diagram with $a$ rows and $b$ columns, i.e., $a \times b := (b,b,\ldots,b)$.

Fix $m \in \IN$.
For a partition $\la \vdash_m \delta$ we write $\{\la\}$ to denote the irreducible $\GL_m$-representation of type $\la$.
These $\{\la\}$ form a pairwise non-isomorphic list of irreducible polynomial $\GL_m$-representations, see \cite{fult:97}.
The dual ($=$contragredient) representation of $\{\la\}$ is denoted by $\{\la^*\}$.
Since $G$ is a reductive group,
every finite dimensional $G$-representation $\sV$ can be decomposed into a direct sum of irreducible $G$-representations, and we write
$\mult_\la \sV$ for the multiplicity of $\{\la\}$ in such a decomposition (although the decomposition might not be unique, the multiplicity is the same in any decomposition).

A \textit{tableau of shape} $\lambda$ over some finite \textit{alphabet} $\mathcal{A}$ is a mapping $\lambda \rightarrow \mathcal{A}$ of boxes to elements in $\mathcal{A}$. For example a tableau of shape $(4,3,1)$ over the alphabet $\mathbb{N}$ is given by $\Yvcentermath1\tiny\young(1322,232,1)$. For a tableau $T$ let $\sh(T):= \lambda$ denote its \textit{shape}, i.e., its vector of row lengths. For a tableau $T$ of shape $\lambda \vdash_m \delta$ over the alphabet $\{1,\dots,m\}$ we define its \textit{content} to be the vector $\varrho \in (\mathbb{N}_{\geq 0})^m$ where $\varrho_i$ counts the number of occurrences of $i$ in~$T$. For example, the tableau $\Yvcentermath1\tiny\young(1322,232,1)$ has content $(2,4,2)$.
The \emph{sorted content} of a tableau is the partition obtained by sorting the entries in the content in a decreasing order, e.g.~$(4,2,2)$ in the preceding example.
For two tableaux $T$ and $S$ we define their \textit{concatenation} $T+S$ by concatenating rows. The resulting tableau $T+S$ has shape $\sh(T)+\sh(S)$. For example, $\Yvcentermath1\tiny\young(1111,2222,3333) + \Yvcentermath1\tiny\young(444,55,6) = \Yvcentermath1\tiny\young(1111444,222255,33336)$. A tableau with entries from $\mathbb{N}$ is called \textit{semistandard} if the entries of each row are nondecreasing from left to right and the entries of each column are strictly increasing from top to bottom.
A semistandard tableau of shape $\la$ is called \emph{standard} if every number $1,\ldots,|\la|$ appears exactly once.
A column of a tableau $T$ is called \textit{regular} if it does not have a repeated entry. A tableau $T$ is called \textit{regular} if each of its columns is regular.
A tableau $L$ is called \emph{duplex} if each column in $L$ appears an even number of times.
For example, $L={\Yvcentermath1\tiny\young(11113344,2222,3344)}$ is duplex, while $L={\Yvcentermath1\tiny\young(11112344,2223,3344)}$ is not duplex. Duplex tableaux are the main idea behind the construction in \cite{bci:10}, see also \cite[Sec.~6.2]{ike:12b}.

Let $\Sym^D \IC^m$ denote the vector space of homogeneous degree $D$ polynomials in $m$ variables.
Let $G := \GL_m$.
The group $G$ acts on $\Sym^D \IC^m$ via $(gp)(x) := p(g^t x)$ for all $g \in G$, $p \in \Sym^D \IC^m$, $x \in \IC^m$,
where the transpose makes this a left action.
We consider the power sum polynomial $p := x_1^D + \cdots + x_m^D \in \Sym^D \IC^m$.
Clearly $p$ is fixed by permuting variables and by rescaling any variable by $D$th roots of unity.
For $D\geq 3$ these group elements generate the whole stabilizer of $p$, see e.g.~\cite{CKW:10}.
Let $H := \stab p = \IZ_D^m \rtimes \aS_m \subseteq G$ denote the stabilizer of $p$.
Let $q := x_1 \cdots x_D$. Clearly $q$ is fixed under permuting the variables and under rescaling each variable by a scalar $\alpha_i$ such that their product $\prod_{i=1}^D \alpha_i$ equals 1. These elements generate the stabilizer of $q$, see e.g.~\cite[Prop.~3.1]{Ike:19}.

Considering $\Sym^D \IC^m$ as a vector space with a $G$-action, let
$\IC[\Sym^D \IC^m]_d$ be the vector space of homogeneous degree $d$ polynomials on $\Sym^D \IC^m$.
In a natural way, $\IC[\Sym^D \IC^m]_d$ is a $G$-representation: $(gf)(q) := f(g^{-1} q)$ for all $g \in G$, $f \in \IC[\Sym^D \IC^m]_d$, $q \in \Sym^D \IC^m$.
The multiplicity $\mult_{\nu^*} \IC[\Sym^D \IC^m]_d$
is called the \emph{plethysm coefficient}, denoted by $a_\nu(d,D)$.
Note that it does not depend on $m$ for $m \geq \ell(\la)$, see e.g.~\cite[Sec.~4.3]{ike:12b}.
For the empty partition $(0)$ we define $a_{(0)}(0,i)=1$.
Whether or not the plethysm coefficient has a nice combinatorial description is an open research question in algebraic combinatorics, see Problem~9 in \cite{sta:00}. Among computer scientists, this question is commonly phrased as whether or not the map $(\nu,d,D)\mapsto a_\nu(d,D)$ is in the complexity class $\#\P$.

Given two irreducible $G$-representations $\{\mu\}$ and $\{\nu\}$, their tensor product $\{\mu\}\otimes\{\nu\}$ is a $G \times G$-representation. Embedding $G \hookrightarrow G \times G$ diagonally via $g \mapsto (g,g)$ the tensor product $\{\mu\}\otimes\{\nu\}$ becomes a $G$-representation that decomposes into irreducibles. The multiplicity $\mult_\la \{\mu\}\otimes\{\nu\}$ is called the \emph{Littlewood-Richardson coefficient} $c^\la_{\mu,\nu}$.
This quantity has numerous beautiful combinatorial interpretations, see e.g.\ \cite{bz:92}, \cite[Sec.~5]{fult:97}, \cite{knta:99}, \cite{buc:00}, \cite[Sec.~10]{ike:12b} and many more.
In particular, the map $(\la,\mu,\nu)\mapsto c^\la_{\mu,\nu}$ is in the complexity class $\#\P$. Even though the exact computation of $c^\la_{\mu,\nu}$ is $\NP$-hard \cite{nara:06}, deciding its positivity is possible in polynomial time \cite{DLM:06}, \cite{MNS:12}, \cite{BI:13LRC}.
Completely analogous properties hold when we take tensor products of polynomially many irreducible $G$-representations $\{\mu^1\}\otimes \cdots \otimes \{\mu^d\}$ and embed $G \hookrightarrow G \times G \times \cdots \times G$. The corresponding coefficient is called the \emph{multi-Littlewood-Richardson coefficient}
$c_{\mu^1,\mu^2,\ldots,\mu^d}^\la$.

For a partition $\varrho \vdash_m d$ the \emph{frequency notation} $\hat\varrho \in \IN^m$ is defined via
$\hat\varrho_i := |\{j \mid \varrho_j = i\}|$.
For example, the frequency notation of $\varrho=(3,3,2,0)$ is $\hat\varrho=(0,1,2,0)$.
We observe that $|\varrho|=\sum_i i \hat\varrho_i$.

\section{Result details}\label{sec:resultdetails}
The following Proposition~\ref{pro:gctandsymmetries} writes the multiplicity
$\mult_{\la^*} \IC[G p]$ as a nonnegative sum of products of multi-Littlewood-Richardson coefficients and plethysm coefficients.

\begin{proposition}\label{pro:gctandsymmetries}
Let $\la \vdash_m dD$.
Define 
\[
b(\la,\varrho,D,d) := \sum_{\mu^1,\mu^2,\ldots,\mu^d \atop \mu^i \vdash D i \hat\varrho_i} c_{\mu^1,\mu^2,\ldots,\mu^d}^\la \prod_{i=1}^d a_{\mu^i}(\hat\varrho_i,iD).
\]
Then
\[
\mult_{\la^*} \IC[G p] = \sum_{\varrho\vdash_m d} b(\la,\varrho,D,d).
\]
\end{proposition}
The proof technique is based on the technique in \cite{BI:11}. The proof is postponed to Section~\ref{sec:gctandsymmetries}.

We remark that if Problem~9 in \cite{sta:00} is resolved positively, then Proposition~\ref{pro:gctandsymmetries} implies that the multiplicity $\mult_{\la^*}\IC[Gp]$ has a combinatorial description, i.e., the map $(\la,m,d,D)\mapsto \mult_{\la^*}\IC[Gp]$ is in $\#\P$. The same holds also for its summands $b(\lambda,\varrho,D,d)$.
It is known that $\mult_{\la^*}\IC[Gq] = a_\la(D,d)$ (see e.g.~\cite[Sec.~9.2.3]{Lan:17}), so the same holds for $\mult_{\la^*}\IC[Gq]$.

Our main technical theorem that enables us to find obstructions is the following.
\begin{theorem}[Main technical theorem]\label{thm:main}
Let $m, d, D \in \IN$. If $D$ is odd, we assume that ${{2(D-1)}\choose{D-1}} \geq 2(m-1)$.
Let $\la \vdash_m dD$.
For each $\varrho\vdash_m d$ define the number $e_\varrho$ as follows:
        \begin{itemize}
                 \item if $D$ is even, let $e_\varrho := \sum_{i=1}^{m} \lceil \frac{\varrho_i}{D-2} \rceil$.
                 \item if $D$ is odd, let $e_\varrho := \sum_{i=1}^{m} 2\lceil \frac{\varrho_i}{2(D-2)} \rceil$.
        \end{itemize}
Let $\Xi$ be a subset of the set of all partitions $\varrho\vdash_m d$.
Let $e_\Xi := \max\{e_\varrho \mid \varrho \in \Xi\}$.
Then
$$\mult_{(\la+(m \times e_\Xi D))^*} \IC[\overline{Gp}] \geq \sum_{\varrho \in \Xi} b(\la,\varrho,D,d).$$
\end{theorem}

The proof of Theorem~\ref{thm:main} is postponed to Section~\ref{sec:proofmain}.

\subsection*{Explicit obstructions via symmetries}

We use Theorem~\ref{thm:main} as follows to construct obstructions.

\begin{theorem}\label{thm:proofofconcept}
Let $m=D\geq 4$ be even.
Let $d=2$.
Let $\la = (2m)$ and let $\nu = (2m)+(m\times 2m)$.
Let $p:=x_1^m+\cdots+x_m^m$ and $q := x_1 x_2 \cdots x_m$.
Let $\Xi = \{(2),(1,1)\}$.
In the notation of Theorem~\ref{thm:main} we have $e_{(2)}=1$ and $e_{(1,1)}=2$, thus $e_\Xi = 2$.
Note that $\nu = \la + (m \times e_\Xi D)$.
We have
\begin{itemize}
 \item $b((2m),(2),m,2)=1$ and $b((2m),(1,1),m,2)=1$ and hence with Theorem~\ref{thm:main} we have $\mult_{\nu^*} \IC[\overline{Gp}] \geq 2$.
 \item $1 \geq \mult_{\nu^*} \IC[G q] \geq \mult_{\nu^*} \IC[\overline{G q}]$.
\end{itemize}
In particular $\mult_{\nu^*} \IC[\overline{Gp}] \geq 2 > 1 \geq \mult_{\nu^*} \IC[\overline{Gq}]$ and hence
$\nu$ is a multiplicity obstruction that proves the separation $\overline{Gp}\not\subseteq\overline{Gq}$.
\end{theorem}
\begin{proof}
We use the following well-known properties about Littlewood-Richardson coefficients (see e.g.~\cite[Ch.~5]{fult:97}).\\
\noindent $\bullet$ The Littlewood-Richardson coefficient is symmetric in the subscript parameters:
$c_{\mu,\nu}^\la = c_{\nu,\mu}^\la$.\\
\noindent $\bullet$ $c_{\mu,(0)}^\la$ equals 1 iff $\mu=\la$, otherwise $c_{\mu,(0)}^\la=0$.

Recall that $\hat\varrho$ is the frequency notation of $\varrho$.
We calculate
\begin{eqnarray*}
b(\la,(2),m,2) &=& \sum_{\mu^1,\mu^2 \atop \mu^i \vdash m i \widehat{(2)}_i} c_{\mu^1,\mu^2}^\la \prod_{i=1}^2 a_{\mu^i}(\widehat{(2)}_i,im) \\
&=&
\sum_{\mu^2 \vdash 2D} c_{(0),\mu^2}^\la \cdot
a_{(0)}(\widehat{(2)}_1,m) \cdot a_{\mu^2}(\widehat{(2)}_2,2m)
\underbrace{=}_{c_{(0),\mu}^\la = \delta_{\mu,\la}}
a_{(0)}(\widehat{(2)}_1,m) \cdot a_{\la}(\widehat{(2)}_2,2m) \\
&=&
a_{(0)}(0,m) \cdot a_{(2m)}(1,2m) = 1 \cdot 1 = 1.
\end{eqnarray*}
Here we used the trivial fact that $a_{(\delta)}(1,\delta)=1$ for all $\delta$.
\begin{eqnarray*}
b(\la,(1,1),m,2) &=& \sum_{\mu^1,\mu^2 \atop \mu^i \vdash m i \widehat{(1,1)}_i} c_{\mu^1,\mu^2}^\la \prod_{i=1}^2 a_{\mu^i}(\widehat{(1,1)}_i,im) \\
&=&
\sum_{\mu^1 \vdash m} c_{\mu^1,(0)}^\la \cdot
a_{\mu^1}(\widehat{(1,1)}_1,m) \cdot a_{(0)}(\widehat{(1,1)}_2,2m)
\underbrace{=}_{c_{\mu,(0)}^\la = \delta_{\mu,\la}}
a_{\la}(\widehat{(1,1)}_1,m) \cdot a_{(0)}(\widehat{(1,1)}_2,2m) \\
&=&
a_{(2m)}(2,m) \cdot a_{(0)}(0,2m) = 1 \cdot 1 = 1.
\end{eqnarray*}
Here we used the classical fact that $a_{\la}(2,m)=1$ if $\la$ has $2m$ boxes and at most 2 rows and both rows have even length 
(see the formula for $h_2 \circ h_n$ in \cite[I.8, Exa.~9(a), p.~140]{MaD:95}).
This proves the first part of the claim.

$\IC[\overline{Gq}] \subseteq \IC[Gq]$ is a subalgebra, so
$\mult_{\nu^*} \IC[\overline{Gq}] \leq \mult_{\nu^*} \IC[Gq]$.
To show $\mult_{\nu^*} \IC[Gq] \leq 1$ we use that
\begin{equation}\label{eq:reverseplethysm}
\mult_{\nu^*} \IC[Gq] = a_\nu(m,2m+2), 
\end{equation}
see e.g.~\cite[Sec.~9.2.3]{Lan:17}.
We apply the upper bound given by the Kostka numbers $a_\nu(m,2m+2) \leq K(\nu,m \times (2m+2))$, which is the number of semistandard tableaux of shape $\nu$ and content $m \times (2m+2)$.
This classical upper bound can be deduced for example from \cite{Gay:76}, see also \cite[Sec.~4.3(A)]{ike:12b}.
It is clear from the special shape of $\nu$ that this Kostka number is 1.
As an illustration, we give an example of this tableau in the case where $m=4$:
\begin{center}{\tiny
\young(1111111111223344,22222222,33333333,44444444)
}\end{center}
\qedhere
\end{proof}

In Section~\ref{sec:notoccobsorvanidoccobs} we prove that the obstructions in Theorem~\ref{thm:proofofconcept} are neither occurrence obstructions nor vanishing ideal occurrence obstructions (in infinitely many cases. This holds in \emph{all} cases if the Alon-Tarsi conjecture is true).

\paragraph{A remark on plethysm coefficients}
As far as we know, Theorem~\ref{thm:main} is the first result of its type in the literature so far.
Even the following direct corollary about plethysm coefficient positivity is new.
\begin{corollary}\label{cor:plethpos}
Let $D\geq 3$ be odd and let $m$ be arbitrary with
$\binom{2(D-1)}{D-1}\geq 2(m-1)$.
Let $d$ be arbitrary.
Let $e=2\lceil\frac d {2(D-2)}\rceil$. Then $a_{(dD)+m\times eD}(d+me,D)\geq 1$.
\end{corollary}
\begin{proof}
Let $\lambda=(dD)$ and $\varrho=(d)$, $\Xi = \{(d)\}$.
Theorem~\ref{thm:main} implies $\mult_{((dD)+(m \times e D))^*} \IC[\overline{Gp}] \geq b((dD),(d),D,d) = 1$.
Since $\overline{Gp}$ is a subvariety of $\Sym^D\IC^m$, it follows $a_{(dD)+m\times eD}(d+me,D)\geq 1$.
\end{proof}
Many additional direct corollaries of this type can be drawn from Theorem~\ref{thm:main}.

\section{Neither occurrence obstructions nor vanishing ideal occurrence obstructions}
\label{sec:notoccobsorvanidoccobs}
The main new property of our obstructions in Theorem~\ref{thm:proofofconcept} is that they use both the symmetry group of $p$ and the symmetry group of~$q$. This is a fundamentally new way of constructing obstructions.
To highlight the novelty of the approach, in this section we prove that the obstructions in Theorem~\ref{thm:proofofconcept} are not vanishing ideal occurrence obstructions. We also prove that they are not occurrence obstructions, provided a property of Latin squares is true (which we know is true for an infinite number of cases).
The novelty of our method gives hope that more and stronger results can be proved in a similar way.

Recall the following proposition from \cite{BI:17}:
\begin{lemma}[\cite{BI:17}]\label{lem:shift}
If $D$ is even, then $\mult_{\la^*} \IC[G p] = \mult_{(\la+(m \times D))^*} \IC[G p]$.
If $D$ is odd, then $\mult_{\la^*} \IC[G p] = \mult_{(\la+(m \times 2D))^*} \IC[G p]$.
\end{lemma}
\begin{lemma}\label{lem:cormain}
Let $\la \vdash_m dD$.
If $D$ is even, let
\[
e := \begin{cases}
     d & \text{ if } d\leq m \\
     m + \lfloor \frac{d-m}{D-2} \rfloor & \text{ if } d\geq m
     \end{cases}.
\]
If $D$ is odd and ${{2(D-1)}\choose{D-1}} \geq 2(m-1)$, let
\[
e := \begin{cases}
     2d & \text{ if } d\leq m \\
     2m + 2\lfloor \frac{d-m}{2(D-2)} \rfloor & \text{ if } d\geq m
     \end{cases}.
\]
In both cases we have $\mult_{(\la+(m \times e D))^*} \IC[\overline{Gp}] = \mult_{\la^*} \IC[G p] = \mult_{(\la+(m \times e D))^*} \IC[Gp]$.
\end{lemma}
\begin{proof}
The second equality follows from Lemma~\ref{lem:shift}.

Let $\Xi$ denote the set of \emph{all} partitions $\varrho\vdash_m d$.
Using Theorem~\ref{thm:main}, then according to Proposition~\ref{pro:gctandsymmetries} we have
$$\mult_{(\la+(m \times e_\Xi D))^*} \IC[\overline{Gp}] \geq \mult_{\la^*} \IC[G p].$$
Using Lemma~\ref{lem:shift} (and the fact that $e_\Xi$ is even if $D$ is odd) we see that
$$\mult_{(\la+(m \times e_\Xi D))^*} \IC[\overline{Gp}] = \mult_{\la^*} \IC[G p].$$
It remains to show that $e_\Xi = e$. Recall that for natural numbers $a,b$ we have $\lceil\frac{1+a}{b}\rceil-1=\lfloor\frac a b\rfloor$.

Let $D$ be even and $d \leq m$. Then the number of nonzero $\varrho_i$ is at most $d$. Hence $e_\varrho \leq d$.
On the other hand, $\varrho=(1,1,\ldots,1,0,0,\ldots,0) \vdash_m d$ provides $e_\varrho = d$, so the bound is tight.
The argument for $D$ odd and $d \leq m$ is completely analogous and yields $e = 2d$.

Let $D$ be even and $d \geq m$.
Then $\varrho = (1+d-m,1,1,\ldots,1)\vdash_m d$ provides $e_\varrho = m+(\lceil\frac{1+d-m}{D-2}\rceil-1) = m + \lfloor \frac{d-m}{D-2}\rfloor$.
The upper bound is provided via
\[
\textstyle e_\varrho = {\displaystyle\sum_{i=1}^m}\lceil\frac{\varrho_i}{D-2}\rceil = {\displaystyle\sum_{i=1}^m}\left(\lceil\frac{\varrho_i+1-1}{D-2}\rceil+1-1\right) = {\displaystyle\sum_{i=1}^m} \left( \lfloor\frac{\varrho_i-1}{D-2}\rfloor +1 \right) = m+\underbrace{\textstyle{\displaystyle{\displaystyle\sum_{i=1}^m}}\lfloor\frac{\varrho_i-1}{D-2}\rfloor}_{\leq \lfloor\frac{d-m}{D-2}\rfloor} \leq m+\lfloor\frac{d-m}{D-2}\rfloor.
\]

Let $D$ be odd and $d \geq m$.
Then $\varrho = (1+d-m,1,1,\ldots,1)\vdash_m d$ provides $e_\varrho = 2m+2(\lceil\frac{1+d-m}{2(D-2)}\rceil-1) = 2m + 2\lfloor \frac{d-m}{2(D-2)}\rfloor$.
The upper bound is provided via
\begin{eqnarray*}
\textstyle e_\varrho &=& {\displaystyle\sum_{i=1}^m}2\lceil\tfrac{\varrho_i}{2(D-2)}\rceil = {\displaystyle\sum_{i=1}^m}2\left(\lceil\tfrac{\varrho_i+1-1}{2(D-2)}\rceil+1-1\right) = {\displaystyle\sum_{i=1}^m} 2\left( \lfloor\tfrac{\varrho_i-1}{2(D-2)}\rfloor +1 \right) \\
&=& 2m+2\underbrace{\textstyle{\displaystyle{\displaystyle\sum_{i=1}^m}}\lfloor\tfrac{\varrho_i-1}{2(D-2)}\rfloor}_{\leq \lfloor\tfrac{d-m}{2(D-2)}\rfloor} 
\leq 2m+2\lfloor\tfrac{d-m}{2(D-2)}\rfloor.
\end{eqnarray*}

\vspace{-0.5cm}\end{proof}

A Latin square of dimension $m$ is an $m \times m$ matrix for which in each row and in each column each number $1,\ldots,m$ appears exactly once.
The \emph{sign} of a column is 1 if the permutation in the column is even, and $-1$ otherwise.
The \emph{sign} of a Latin square is defined as the product of all column signs.
A Latin square is called \emph{even} if its sign is 1, and \emph{odd} otherwise.
If $m$ is odd, then it is easy to construct an involution on the set of all $m \times m$ Latin squares that pairs each even Latin square with an odd one. Hence, if $m$ is odd, then the number of even $m \times m$ Latin squares equals the number of odd $m \times m$ Latin squares. If $m$ is even, then Alon and Tarsi \cite{AT:92} conjecture that the number of even $m \times m$ Latin squares \emph{differs from} the number of odd $m \times m$ Latin squares (see also \cite{HR:94} and \cite{KL:15} for equivalent formulations). This is proved for all $m=\tau+1$ \cite{Dri:97} for an odd prime number $\tau$ and for all $m=\tau-1$ \cite{Gly:10} for an odd prime number $\tau$, making $m=26$ the smallest open case.
If the number of even $m \times m$ Latin squares differs from the number of odd $m \times m$ Latin squares, then we say that $m$ satisfies the \emph{Alon-Tarsi condition}.
The Alon-Tarsi condition first appeared in connection with geometric complexity theory in \cite{Kum:15}.

\begin{proposition}\label{pro:alontarsi}
Let $\nu$, $m$, $D$, $d$ be as in Theorem~\ref{thm:proofofconcept}.
\begin{itemize}
\item $\mult_{\nu^*} \IC[\overline{Gp}]< a_\nu(2(m+2),m)$ and hence $\nu$ is not a vanishing ideal occurrence obstruction.
\item If $m$ satisfies the Alon-Tarsi condition, then $\mult_{\nu^*} \IC[\overline{Gq}] > 0$ and hence $\nu$ is not an occurrence obstruction. In particular this is true for $m=\tau\pm1$ for all odd primes $\tau$.
\end{itemize}
\end{proposition}
\begin{proof}
The first bullet point is treated as follows.
$\mult_{\nu^*} \IC[\overline{Gp}] = \mult_{\nu^*} \IC[{Gp}]$ by Lemma~\ref{lem:cormain}.
Moreover, $\mult_{\nu^*} \IC[{Gp}]=2$ (Proposition~\ref{pro:gctandsymmetries} and Theorem~\ref{thm:proofofconcept}).
It remains to prove that $a_\nu(2(m+1),m)\geq 3$, which is postponed to Proposition~\ref{pro:threetableaux}.

The second bullet point is treated as follows.
Kumar proved \cite{Kum:15} that $\mult_{(m)^*} \IC[\overline{G(x_1\cdots x_m)}] \geq 1$
and that $\mult_{(m \times m)^*} \IC[\overline{G(x_1\cdots x_m)}] \geq 1$,
provided that $m$ satisfies the Alon-Tarsi condition.
Since $\nu = (m \times m) + (m \times m) + (m) + (m)$,
the semigroup property for occurrences in $\IC[\overline{G(x_1\cdots x_m)}]$ (see e.g.~\cite{DIP:19}) yields that $\mult_{\nu^*} \IC[\overline{G(x_1\cdots x_m)}] \geq 1$.
\end{proof}

The rest of this paper (besides Section~\ref{sec:threetableaux}) is dedicated to the proof of Theorem~\ref{thm:main}.

\section{Preliminaries - Representation Theory} \label{sec:prelrepr}
In the remainder of this paper we write $G:=\textup{GL}_m(\mathbb{C})$ to denote the general linear group for some fixed natural number $m$. 
A \emph{representation} of $G$ is a finite dimensional complex vector space $\mathscr{V}$ together with a group homomorphism $\xi: G \rightarrow \textup{GL}(\mathscr{V})$.
If $\xi$ is given by a polynomial map, then we call $(\mathscr{V},\xi)$ a \emph{polynomial representation}.
We write $gf := \xi(g)(f)$ for all $g \in G$ and $f \in \mathscr{V}$. A linear subspace $\mathscr{W} \subseteq \mathscr{V}$ that is closed under the action of $G$ is called a \textit{subrepresentation}. A representation $\mathscr{V}$ is called \textit{irreducible} iff it has only two subrepresentations, namely the zero-dimensional subspace and $\mathscr{V}$ itself.

For a $G$-representation $\mathscr{V}$, a \textit{highest weight vector} $f \in \mathscr{V}$ of weight $\lambda \in \IZ^m$ is defined to be a vector that satisfies the following two properties:
\begin{itemize}
        \item for all upper triangular matrices $g$ with 1s on the main diagonal we have $gf=f$, and
        \item for all diagonal matrices $g := \textup{diag}(a_1,\dots, a_m)$ we have $gf = a_1^{\lambda_1}\dots a_m^{\lambda_m} f$.
\end{itemize}
We denote by $\HWV_\lambda(\mathscr{V})$ the vector space of highest weight vectors of weight $\lambda$ in $\mathscr{V}$. It turns out that for an irreducible polynomial $G$-representation $\mathscr{V}$ there is exactly one partition $\lambda \in \mathbb{N}^m$ such that $\HWV_\lambda(\mathscr{V})$ has dimension $=1$,
while for all $\mu\neq\la$ we have $\dim\HWV_\mu(\mathscr{V})=0$.
In this case we write $\mathscr{V} = \{\lambda\}$. Moreover, for each partition $\lambda \vdash_m$ there is an irreducible polynomial $G$-representation $\{\lambda\}$ and we call it the irreducible $G$-representation of \textit{isomorphism type} $\lambda$. Furthermore, these $\{\lambda\}$ are pairwise non-isomorphic. We can now define the \textit{multiplicity} $\textup{mult}_\lambda(\mathscr{V})$ of $\{\lambda\}$ in $\mathscr{V}$ as
\begin{align*}
\textup{mult}_\lambda(\mathscr{V}) = \dim \HWV_\lambda(\mathscr{V}),
\end{align*}
which is the same as the multiplicity with which $\{\la\}$ appears as a summand in a decomposition of $\mathscr{V}$ into a direct sum of irreducible $G$-representations.

\subsection*{The irreducible $G$-representations}
In the following exposition we closely follow \cite{fult:97}.

For a tablean $T$ let $T(r,c)$ denote the entry of $T$ in row~$r$ and column~$c$.
Let $(e_i)_i$ be the standard basis of $\IC^m$.
Let $\mu = \lambda^t$.

To each tableau $T : \la \to \{1,\ldots,m\}$ we assign a tensor
\[
e_{T(1,1)} \otimes e_{T(2,1)} \otimes \cdots \otimes e_{T(\mu_1,1)} \otimes e_{T(1,2)} \otimes e_{T(2,2)} \otimes \cdots \otimes e_{T(\mu_2,2)} \otimes \cdots \cdots \otimes e_{T(\mu_{\la_1},\la_1)}.
\]

These form a basis of $\tensor^{|\la|}\IC^m$.
In this way, the space $\tensor^{|\la|}\IC^m$ is isomorphic to the space of formal linear combinations of tableaux $T : \la \to \{1,\ldots,m\}$.
Using this isomorphism, the space of tableaux inherits the natural $G$-action on $\tensor^{|\la|}\IC^m$, which is given by
\begin{align*}
g(v_1 \otimes \dots \otimes v_{|\lambda|}) = g(v_1) \otimes \dots \otimes g(v_{|\lambda|}).
\end{align*}
The following vectors span the linear subspace $K(\la)$ of \textit{Grassmann-Plücker relations} (sometimes called \textit{shuffle relations}), which is invariant under the $G$-action:
\begin{itemize}
        \item $T+T'$, where $T'$ is a tableau that arises from $T$ by switching two numbers within one column.
        \item $T-\Sigma S$, where for two fixed columns $j,j'$ and a fixed number of entries $k$ the sum is over all tableaux $S$ that arise from $T$ by exchanging the top $k$ entries in column $j$ with any $k$ entries in column $j'$, preserving the internal vertical order.
\end{itemize}
It turns out that $\{\la\} \simeq (\tensor^{|\la|}\IC^m) / K(\la)$, see \cite[Ch.~8]{fult:97}.
A basis of $\{\lambda\}$ is given by the semistandard tableaux of shape $\lambda$ with entries from $\{1,\dots,m\}$.
The unique highest weight vector (up to scale) in $\{\lambda\}$ is the \textit{superstandard tableau} of shape $\lambda$, which is the semistandard tableau that has only entries $i$ in row $i$. It has weight $\lambda$.
\section{Two similar projections of tableaux} \label{sec:twoprojections}
In this section we present two symmetrizations that look quite similar on the basis of tableaux.
We crucially use this similarity in Section~\ref{sec:proofmain} in the proof of the Main Theorem~\ref{thm:main}.
Indeed, this peculiarity is the driving force behind our result.

We embed $\mathfrak{S}_m \subseteq G$ via permutation matrices. Given a tableau $S:\lambda \rightarrow \{1,\dots,m\}$ we define:
\begin{align*}
        P_m S := \sum_{\pi \in \mathfrak{S}_m} \pi S
\end{align*}
Interpreting a tableau $S:\lambda \rightarrow \{1,\dots,m\}$ as an element in $\{\lambda\}$ this can be seen as a map $P_m: \{\lambda\} \rightarrow \{\lambda\}^{\mathfrak{S}_m}$, where $\{\lambda\}^{\mathfrak{S}_m}$ consists of the elements in $\{\lambda\}$ that are invariant under the action of~$\mathfrak{S}_m$. For example, using the Grassmann-Plücker relations we get that:
\begin{eqnarray*}
        P_{3}\,{\Yvcentermath1\small\young(1122,33)} &=& {\Yvcentermath1\small\young(1122,33)}+\underbrace{{\Yvcentermath1\small\young(2211,33)}}_{=\,{\Yvcentermath1\tiny\young(1133,22)}\,-2\,{\Yvcentermath1\tiny\young(1123,23)}\,+\,{\Yvcentermath1\tiny\young(1122,33)}}+\underbrace{{\Yvcentermath1\small\young(3322,11)}}_{=\,{\Yvcentermath1\tiny\young(1122,33)}}\\
        &+&{\Yvcentermath1\small\young(1133,22)} + \underbrace{{\Yvcentermath1\small\young(2233,11)}}_{=\,{\Yvcentermath1\tiny\young(1133,22)}}+\underbrace{{\Yvcentermath1\small\young(3311,22)}}_{=\,{\Yvcentermath1\tiny\young(1133,22)}\,-2\,{\Yvcentermath1\tiny\young(1123,23)}\,+\,{\Yvcentermath1\tiny\young(1122,33)}} \\
        &=& 4\,{\Yvcentermath1\small\young(1122,33)}-4\,{\Yvcentermath1\small\young(1123,23)}+4\,{\Yvcentermath1\small\young(1133,22)}
\end{eqnarray*}

Now we consider tableaux that have $\delta$ many symbols.
Let $\varphi: \{1,\dots,\delta\} \rightarrow \{1,\dots,m\}$ be a map. For a tableau $T: \lambda \rightarrow \{1,\dots,\delta\}$, we define $\varphi T$ as the tableau that is the result of replacing the entries from $\{1,\dots,\delta\}$ in $T$ with entries from $\{1,\dots,m\}$.
Let
\[
        \mathcal{M}_{\delta,m} := \{\varphi \mid \varphi: \{1,\dots,\delta\} \rightarrow \{1,\dots,m\} \; \text{is a map}\}.
\]
We use this to define
\begin{align*}
M_{\delta,m}T:=\sum_{\varphi \in \mathcal M_{\delta,m}} \varphi T \in \{\la\}^{\aS_m}.
\end{align*}
\section{Tableau contraction} \label{sec:tableaucontraction}
Let $\IS \in \{\la\}$ be the superstandard tableau.
Let $\gamma \in \{\la\}^*$ denote the vector dual to $\IS$, i.e., the linear map that satisfies $\gamma(\IS)=1$ and $\gamma(T)=0$ for every semistandard tableau $T \neq \IS$.

In the following our goal is to understand $\gamma$ explicitly in terms of determinants, see eq.~\eqref{eq:gammadet} below.
For a matrix $g$ let $g_{1..j,i} \in \IC^j$ denote the vector that consists of the top $j$ elements in the $i$th column of $g$.
We interpret a list of $j$ vectors in $\IC^i$ as an $i \times j$ matrix of column vectors.
A determinant of a matrix with more rows than columns is defined as the determinant of the square matrix of its top rows.
For a tableau $T : \la \to \{1,\ldots,m\}$ let $T(r,c) \in \{1,\ldots,m\}$ denote the number in row $r$ and column $c$.
\begin{lemma}\label{lem:coldets}
For $g\in\IC^{m \times m}$ we have
\begin{equation}\label{eq:gammadet}
\gamma(gT) = \prod_{c=1}^{\la_1}\det(g_{1..\mu_c, T(1,c)},\ldots,g_{1..\mu_c, T(\mu_c,c)}),
\end{equation}
where $\mu=\la^t$.
\end{lemma}
\begin{proof}
\[
gT = \sum_{S:\la\to\{1,\ldots,m\}} \left(\prod_{(r,c)\in \la} g_{S(r,c),T(r,c)} \right) S = \sum_{S:\la\to\{1,\ldots,m\}\atop S \text{ regular}} \left(\prod_{(r,c)\in \la} g_{S(r,c),T(r,c)} \right) S
\]
The fact that $\IS$ is superstandard implies that $\gamma(S)=0$ for all $S$ that do not have a permutation of $\{1,\ldots,\mu_i\}$ in column $i$ for all $1 \leq i \leq \la_1$.
Therefore
\[
\gamma(gT) = \sum_{S:\la\to\{1,\ldots,m\}\atop S \text{ is a col-perm.\ of $\IS$}} \left(\prod_{(r,c)\in \la} g_{S(r,c),T(r,c)} \right) \gamma(S) =
\sum_{S:\la\to\{1,\ldots,m\}\atop S \text{ is a col-perm.\ of $\IS$}} \left(\prod_{(r,c)\in \la} g_{S(r,c),T(r,c)} \right) \prod_{\text{column }c} \sgn(S,c),
\]
where $\sgn(S,c)$ is the sign of the permutation of column $c$ of $S$.
We conclude
\[
\gamma(gT) = \prod_{c=1}^{\la_1} \sum_{\pi \in \aS_{\mu_c}} \sgn(\pi) \left(\prod_{r=1}^{\mu_c} g_{\pi(r),T(r,c)} \right),
\]
which finishes the proof.
\end{proof}
We draw three quick corollaries:
\begin{corollary}\label{cor:gammafactorization}
$\gamma(gT) = \prod_{c=1}^{\la_1}\gamma(g \col_c),$
where $\col_c$ is the $c$-th column of $T$.
\end{corollary}
\begin{proof}
 Obvious from Lemma~\ref{lem:coldets}.
\end{proof}
\begin{corollary}\label{cor:pm1}
For a tableau $T$ that consists of a single regular column of $m$ boxes, we have $\gamma(gT) \in \{-\det(g),\det(g)\}$.
In particular, if $g \in \SL_m$, we have $\gamma(gT) \in \{-1,1\}$.
\end{corollary}
\begin{proof}
Since $T$ is regular, the entries of $T$ are precisely all numbers $1,\ldots,m$, not necessarily sorted. Thus according to Lemma~\ref{lem:coldets}, $\gamma(gT)\in\{-\det(g),\det(g)\}$.
\end{proof}
\begin{corollary}\label{cor:gammazero}
If $T$ is not regular, then $\gamma(gT)=0$, independent of $g$.
\end{corollary}
\begin{proof}
$\gamma(gT)$ factorizes according to Cor.~\ref{cor:gammafactorization}.
Since $T$ is not regular, there is a column with a repeated entry.
Thus, according to Lemma~\ref{lem:coldets},
one of the factors of $\gamma(gT)$ equals the determinant of a matrix with a repeated column,
and hence is zero.
\end{proof}

\section{Highest weight functions on the orbit}\label{sec:hworbit}
In this section we prove the following theorem.
\begin{theorem}\label{thm:functionsonorbit}
The vector space $\HWV_{\la^*}(\IC[Gp]_d)$
decomposes into a direct sum of vector spaces
$\HWV_{\la^*}(\IC[Gp]_d) = \bigoplus_{\varrho\vdash_m d} \sW_\varrho$,
and each $\sW_\varrho$ is generated by the functions
\[
g \mapsto \gamma(g P_m S),
\]
where $S$ runs over all semistandard tableaux $S$ of shape $\la$
and content $\varrho D$.

Let $\sW_\varrho'$ denote the linear space spanned by
all semistandard tableaux $S$ of shape $\la$
whose sorted content is $\varrho D$. Then $\sW_\varrho$ is isomorphic to the $\aS_m$-invariant subspace of $\sW_\varrho'$.
\end{theorem}

\subsection*{Explicit algebraic Peter-Weyl theorem}
$G \subseteq \IC^{m \times m}$ is the nonvanishing set of a polynomial and thus $G$ is a variety.
It carries a canonical action of $G\times G$ via $(h',h)g := h' g h^{-1}$ for all $g,h,h'\in G$.
This action lifts to to the coordinate ring $\IC[G]$ via the canonical pullback: $((h',h)f)(g) := f({h'}^{-1} g h)$ for all $g,h,h'\in G$, $f \in \IC[G]$.
\begin{theorem}[{Explicit algebraic Peter-Weyl theorem (see\ e.g.\ \cite[Thm.~8.6.4.3]{Lan:17})}]\label{thm:explalgpeterweyl}
        As a $G \times G$-representation we have
        \[
        \IC[G] = \bigoplus_{\la} \{\la\}^* \otimes \{\la\}
        \]
        If we embed $G \hookrightarrow G \times G$ via $g \mapsto (g,\textup{id})$,
        then the $\la^*$-isotypic component of $\IC[G]$ is $\{\la\}^* \otimes \{\la\}$,
        which is spanned by functions
        \[
        g \mapsto l(gv),
        \]
        where $l\in\{\la\}^*$ and $v \in \{\la\}$.
\end{theorem}

\subsection*{$H$-invariants}
Recall the definition of the symmetry subgroup $H \leq G$ from Section~\ref{sec:preliminaries}.
Again, consider the action of $G\times G$ on $G$ via $(h',h)g := h' g h^{-1}$ for all $g,h,h'\in G$.
The action of $G\times G$ lifts to the coordinate ring $\IC[G]$ and we denote by $\IC[G]^{\vv H}$ the linear subspace of right $H$-invariants. $\IC[G]^{\vv H}$ carries a left $G$-action.
There is a bijection $gp \sim gH$ between points in the orbit $Gp$ and left cosets of $H$.
This leads to the following explicit $G$-equivariant algebra isomorphism \cite[25.4.6, Prop.]{TY:05}:
\begin{eqnarray}
\IC[G p] &\stackrel{\sim}{\longrightarrow}& \IC[G]^{\vv H} \label{eq:algiso}\\
F & \mapsto & [g \mapsto F(gp)] \nonumber\\
\mbox{~} [gp \mapsto f(g)] & \reflectbox{\ensuremath{\mapsto}} & f \nonumber
\end{eqnarray}

Taking right $H$-invariants in Theorem~\ref{thm:explalgpeterweyl} yields:
\[
\IC[G]^{\vv H} = \bigoplus_{\la} \{\la\}^* \otimes \{\la\}^H,
\]
and explicitly: the left $\la^*$-isotypic component of $\IC[G]^{\vv H}$ is spanned by functions
\[
g \mapsto l(gv),
\]
where $l\in\{\la\}^*$ and $v \in \{\la\}^H$.

Taking left highest weight functions, we see that $\HWV_{\la^*}(\IC[G]^{\vv H}) \simeq \{\la\}^H$
and that $\HWV_{\la^*}(\IC[G]^{\vv H})$
is spanned by functions
\begin{equation}\label{eq:ggammagv}
g \mapsto \gamma(gv),
\end{equation}
where $v \in \{\la\}^H$. In the next subsection we will use the special structure of $H$ to finish the proof of Theorem~\ref{thm:functionsonorbit}.

\subsection*{The coordinate ring of the orbit of the power sum}
Considering eq.~\eqref{eq:ggammagv}, in order to understand $\HWV_{\la^*}(\IC[G]^{\vv H})$
we need to understand the $H$-invariant space $\{\la\}^H$.
The group $H$ is the semidirect product $\{\la\}^H = (\{\la\}^{\IZ_D^m})^{\aS_m}$, which enables us to analyze $\{\la\}^H$ via a two-step process
by first analyzing $\{\la\}^{\IZ_D^m}$ and then taking $\aS_m$-invariants.
We consider the basis of $\{\la\}$ given by the semistandard tableaux $S:\la\to\{1,\ldots,m\}$.
Note that for $g \in \IZ_D^m$ we have that $S$ and $gS$ coincide up to rescaling with a $D$th root of unity.
Indeed, if any symbol in the tableau $S$ does not occur a multiple of $D$ times, then $S$ vanishes under the symmetrization
\[
S \mapsto \frac{1}{D^m} \sum_{g \in \IZ_D^m} \pi S.
\]
Moreover, each $S$ in which every number appears a multiple of $D$ many times is fixed under this symmetrization map.
Hence $\{\la\}^{\IZ_D^m}$ has a basis given by semistandard tableaux of shape $\la$ in which each number appears a multiple of $D$ many times.
For a partition $\varrho$ let $\sW_\varrho'$ denote the linear space spanned by all semistandard tableaux of shape $\la$ whose sorted content is $\varrho D$.
This gives a decomposition $\{\la\}^{\IZ_D^m} = \bigoplus_{\varrho}\sW_\varrho'$ as a direct sum.
The action of $\aS_m$ leaves $\sW_\varrho'$ fixed.
Hence $(\{\la\}^{\IZ_D^m})^{\aS_m} = \bigoplus_\varrho (\sW_\varrho')^{\aS_m}$.
A generating set for $\{\la\}^H$ is thus obtained by taking the $\sW_\varrho'$ and their tableau bases
and independently symmetrizing over $\mathfrak S_m$, which is done (up to scale) by applying $P_m$.
Using eq.~\eqref{eq:ggammagv} we thus obtain the following proposition.

\begin{proposition}
The vector space $\HWV_{\la^*}(\IC[G]^{\vv H}_d)$
decomposes into a direct sum of vector spaces
$\HWV_{\la^*}(\IC[G]^{\vv H}_d) = \bigoplus_{\varrho\vdash_m d} \sW_\varrho$,
and each $\sW_\varrho$ is generated by the functions
\[
g \mapsto \gamma(g P_m S),
\]
where $S$ runs over all semistandard tableaux $S$ of shape $\la$
and content $\varrho D$.

Let $\sW_\varrho'$ denote the linear space spanned by
all semistandard tableaux $S$ of shape $\la$
whose sorted content is $\varrho D$. Then $\sW_\varrho$ is isomorphic to the $\aS_m$-invariant subspace of $\sW_\varrho'$.
\end{proposition}

Theorem~\ref{thm:functionsonorbit} now follows immediately by applying the algebra isomorphism \eqref{eq:algiso}.

\section{A formula for the multiplicities in the coordinate ring of the orbit of the power sum (Proof of Prop.~\protect{\ref{pro:gctandsymmetries}})}\label{sec:gctandsymmetries}
Crucial parts of this section are the result of a collaboration with Greta Panova and appeared in the first author's unpublished lecture notes for a winter 2017/2018 course on geometric complexity theory \cite{Ike:19}.

In this section we prove Proposition~\ref{pro:gctandsymmetries}.
In fact, we prove a slightly stronger result as follows.
\begin{proposition}\label{pro:precisedecomposition}
Let $\{\la\}_{\varrho} \subseteq \{\la\}$ denote the linear subspace spanned by the tableaux whose sorted content is $\varrho D$.
Then
$\dim (\{\la\}_{\varrho})^{\aS_m} = b(\la,\varrho,D,d)$.
\end{proposition}
\begin{proof}
Let $\{\la\}^\varrho$ denote the $\varrho$-weight space, i.e., the linear space spanned by tableaux of shape $\la$ and content $\varrho$.
Then $\{\la\}_\varrho = \bigoplus_{\gamma \in \aS_m\varrho} \{\lambda\}^\gamma$.

For a partition $\varrho \vdash_m d$ let $\aS_m\varrho \subseteq \IN^m$ denote the orbit of $\varrho$. Note that $\varrho$ is the only partition in its orbit, while the other lists are not in the correct order.
Let $\stab \varrho \leq \aS_m$ denote the stabilizer of $\varrho$.

\begin{claim}
$\dim (\{\la\}_\varrho)^{\aS_m} = \dim\left(\{\lambda\}^\varrho\right)^{\stab \varrho}.$
\end{claim}
\begin{proof}
We construct an isomorphism of vector spaces.

Let $W := \{\la\}$.
Let $\pi_1,\ldots,\pi_r$ be a system of representatives of left cosets for $\stab \varrho \leq \aS_m$ with $\pi_1 = \text{id}$,
i.e., $\aS_m = \pi_1 \stab \varrho \ \dot\cup \ \cdots \ \dot\cup \ \pi_r\stab\varrho$ and we have $\aS_m \varrho = \{\pi_1 \varrho , \ldots, \pi_r \varrho\}$.
Therefore we have the decomposition
\[
W_\varrho = \bigoplus_{j=1}^r \pi_j W^\varrho.
\]
Let $\overline p : W_\varrho \twoheadrightarrow W^\varrho$ be the projection according to this decomposition.
We claim that the restriction
\[
p : \left(W_\varrho\right)^{\aS_m} \to \left(W^\varrho\right)^{\stab\varrho}
\]
is an isomorphism of vector spaces. This then finishes the proof.
We verify well-definedness, injectivity, and surjectivity of $p$.

Well-definedness: The spaces $\pi_1 W^\varrho,\ldots,\pi_r W^\varrho$ are permuted by $\aS_m$.
Every $\sigma \in \stab \varrho$ fixes $W^\varrho$, thus $\sigma v_1 = v_1$ if $v_1 \in W^\varrho$.
Thus the map $v = \sum_{j=1}^r v_j \stackrel{\overline p}{\mapsto} v_1$ maps $W_\varrho$ to $(W^\varrho)^{\stab\varrho}$.

Injectivity: If $v \in (W_\varrho)^{\aS_m}$, then $v = \pi v = \sum_j \pi v_j$. Therefore $v_j = \pi_j v_1$.
If $p(v)=0$, then $v_1=0$, thus all $v_j=0$, which proves injectivity.

Surjectivity: Let $v_1 \in (W^\varrho)^{\stab \varrho}$. Set $v_j := \pi_j v_1$ and put $v:=\sum_j v_j$.
Clearly $p(v)=v_1$. It remains to verify that $v$ is $\aS_m$-invariant.
\[
v = \sum_{j=1}^r \pi_j v_1 = \sum_{j=1}^r \tfrac{1}{|\stab \varrho|} \sum_{\tau \in \stab \varrho} \pi_j \tau v_1 = \tfrac{1}{|\stab \varrho|} \sum_{\pi \in \aS_m}\pi v_1,
\]
which is $\aS_m$-invariant.
\end{proof}

We are left with determining $\dim\left(\{\la\}^{\varrho}\right)^{\stab \varrho}$.
We use a detour via Specht modules:
the Specht module $[\la]$ is an irreducible $\aS_{|\la|}$-representation that can be constructed as the subrepresentation of $\{\la\}$ spanned by all standard tableaux.
Define the Young subgroup
$G_\varrho \subseteq \aS_{dD} := \aS_{\varrho_1 D} \times \cdots \times \aS_{\varrho_m D}$.
We use that $\{\la\}^{\varrho} \simeq [\la]^{G_\varrho}$, see e.g.~\cite[Sec.~4.3(A)]{ike:12b}.

Schur-Weyl duality implies that
\begin{eqnarray*}
\dim\left([\la]^{G_\varrho}\right)^{\stab \varrho} = \dim\HWV_\la(\{\la\}\otimes ([\la]^{G_\varrho})^{\stab \varrho} ) = \dim\HWV_\la((\otimes^{dD}V)^{G_\varrho\rtimes{\stab \varrho}})
\end{eqnarray*}
for $V$ having large enough dimension.

\begin{eqnarray*}
(\otimes^{dD}V)^{G_\varrho\rtimes{\stab \varrho}} &=& (\Sym^{D\varrho_1}V \otimes \cdots \otimes \Sym^{D\varrho_m}V)^{\stab \varrho}\\
&=& \left(\bigotimes^{\hat\varrho_1}\Sym^{D}V \otimes \bigotimes^{\hat\varrho_2}\Sym^{2D}V \otimes \cdots \otimes \bigotimes^{\hat\varrho_d}\Sym^{dD}V\right)^{\stab \varrho} \\
&=& \underbrace{\IC[\Sym^{D}V]^*_{\hat\varrho_1}}_{=\bigoplus_{\mu^1} \{\mu^1\}^{\oplus a_{\mu^1}(\hat\varrho_1,D)}} \otimes \IC[\Sym^{2D}V]^*_{\hat\varrho_2} \otimes \cdots \otimes \underbrace{\IC[\Sym^{dD}V]^*_{\hat\varrho_d}}_{{=\bigoplus_{\mu^d} \{\mu^d\}^{\oplus a_{\mu^d}(\hat\varrho_d,dD)}}} \quad\quad\quad\hfill(\dagger).
\end{eqnarray*}
The multiplicity of $\{\mu^i\}^*$ in $\IC[\Sym^{iD}V]_{\hat\varrho_i}$ is $a_{\mu^i}(\hat\varrho_i,iD)$.
Using distributivity we obtain that the multiplicity of $\{\la\}$ in the representation $(\dagger)$ equals
\[
\sum_{\mu^1,\mu^2,\ldots,\mu^d \atop \mu^i \vdash \hat\varrho_i D i} c_{\mu^1,\mu^2,\ldots,\mu^d}^\la \prod_{i=1}^d a_{\mu^i}(\hat\varrho_i,iD)
\]
\end{proof}
\begin{proof}[Proof of Proposition~\ref{pro:gctandsymmetries}]
$\mult_\la \IC[Gp]_d \stackrel{\text{Thm. }\ref{thm:functionsonorbit}}{=} \sum_{\varrho\vdash_m d}(\sW_\varrho)^{\aS_m} \stackrel{\text{Prop.~\ref{pro:precisedecomposition}}}{=} \sum_{\varrho\vdash_m d} b(\la,\varrho,D,d)$.
\end{proof}

\section{Highest weight functions on the orbit closure}\label{sec:hworbitclosure}

In this section we study the coordinate ring $\IC[\overline{Gp}]$. Much less is known about this ring compared to $\IC[Gp]$, in particular we do not have formulas for $\IC[\overline{Gp}]$ that are comparable to Prop.~\ref{pro:gctandsymmetries}.
Nevertheless, the following theorem provides a way to analyze $\IC[\overline{Gp}]$ and connect it to $\IC[Gp]$.

For an $m \times n$ matrix $A$ let $A p := \ell_1^D+\ell_2^D+\cdots+\ell_n^D$, where $\ell_i$ is the linear form given by the $i$-th column of $A$. Note that this is a generalization of $gp$ for $g \in \GL_m$.

\begin{theorem}\label{thm:functionsonorbitclosure}
If $|\la|$ is not divisible by $D$, then
$\HWV_{\la^*}(\IC[\overline{Gp}])=0$.
        
Let $\la \vdash \delta D$.
The vector space $\HWV_{\la^*}(\IC[\overline{Gp}])$ is generated by the functions
\[
g \mapsto \gamma(g M_{\delta,m} T),
\]
where $T$ runs over all semistandard tableaux of shape $\la$ in which each entry $1,\ldots,\delta$ appears exactly $D$ many times.

Additionally, for a semistandard tableau $T$ of shape $\la$ in which each entry $1,\ldots,\delta$ appears exactly $D$ many times, the function
\[
A \mapsto \gamma(A M_{\delta,m} T)
\]
is either zero or a HWV of weight $\la^*$ in $\IC[\Sym^D \IC^m]$.
\end{theorem}
This is a rephrasing of \cite{AIR:16}, which is a special case of \cite[Prop.~4.5 and Thm.~4.7]{BIP:19}.
Indeed, \cite{BIP:19} covers more general cases.
For the sake of completeness and to highlight that the proof technique is very different from the technique in Section~\ref{sec:hworbit}, we prove Theorem~\ref{thm:functionsonorbitclosure}.

\begin{proof}
The first observation follows from the fact that $\overline{Gp} \subseteq \Sym^D\IC^m$ is a closed subvariety that is closed under rescaling,
and hence $\IC[\overline{Gp}]$ is a graded subalgebra of $\IC[\Sym^D\IC^m]$.
We know that in each degree $\delta$ component $\IC[\Sym^D\IC^m]_\delta$ of $\IC[\Sym^D\IC^m]$ the only types $\la^*$ that occur satisfy $\la \vdash \delta D$, see e.g.~\cite[Lemma~4.3.3]{ike:12b}.

Schur-Weyl duality yields that
\begin{equation}\label{eq:schurweyl}
\tensor^{\delta D} \IC^{m*} = \bigoplus_{\la \vdash_m \delta D} \{\la^*\} \otimes [\la].
\end{equation}
A highest weight vector of weight $\la^*$ in $\tensor^{\delta D} \IC^{m*}$ is given for example by
\[
v_\la := x_1 \wedge x_2 \wedge \cdots \wedge x_{\mu_1} \otimes x_1 \wedge x_2 \wedge \cdots \wedge x_{\mu_2} \otimes \cdots \cdots \otimes x_1 \wedge x_2 \wedge \cdots \wedge x_{\mu_{\la_1}},
\]
where $\mu = \la^t$ and $\{x_i\}_i$ is the basis of $\IC^{m*}$.
Let $T_\la$ denote the column-standard tableau of shape $\la$, i.e., the tableau that is filled with the numbers $1,\ldots,|\la|$ in a columnwise fashion from left to right, top to bottom.
Since $[\la]$ is irreducible,
from \eqref{eq:schurweyl} we see that
$\HWV_{\la^*}(\tensor^{\delta D}\IC^{m*})$ is generated by the set $\{v_\la \pi \mid \pi \in \aS_{\delta D} \text{ s.t. $\pi T_\la$ is standard}\}$.

By the polarization principle (see e.g.~\cite[Claim~4.2.13]{ike:12b}),
all functions $f$ in $\IC[\Sym^D{\IC^{m*}}]_\delta$ can be obtained via some tensor $v_f \in \tensor^{\delta D}\IC^m$ and defining
\[
f(y) := \langle v_f , y^{\otimes \delta}\rangle.
\]
The resulting function $f$ is a highest weight function iff $v_f$ is a HWV.
Thus we see that $\HWV_{\la^*}(\IC[\Sym^D\IC^m]_\delta)$ is generated by the functions
\[
f(y) := \langle \pi v_\la, y^{\otimes \delta}\rangle.
\]
In the following we analyze how to restrict these functions to $\overline{Gp}$.
When $y = gp = \ell_1^D+\cdots+\ell_m^D$ is in the orbit of the power sum, then clearly
\[
y^{\otimes \delta} = \sum_{\varphi:\{1,\ldots,\delta\}\to\{1,\ldots,m\}} \ell_{\varphi(1)}^D \otimes \cdots \otimes \ell_{\varphi(\delta)}^D.
\]
The linear forms $\ell_i$ correspond to the vectors $g_{1..m,i}$.
We evaluate
\[
f(y) := \langle \pi v_\la, y^{\otimes \delta}\rangle = \langle v_\la, \pi(y^{\otimes \delta})\rangle
\]
\[
= \sum_{\beta:\{1,\ldots,\delta D\}\to\{1,\ldots,m\}\atop\text{respecting $\pi T_\la$}} \prod_{c=1}^{\la_1}\det(g_{1..\mu_c, \beta(\pi T_\la(1,c))},\ldots,g_{1..\mu_c, \beta(\pi T_\la(\mu_c,c))}),
\]
where $\beta$ \emph{respects} a tableau $S$ if all numbers $1,\ldots,D$ are mapped to the same value, and all numbers $D+1,\ldots,2D$ are mapped to the same value, and so on.

Therefore the vector space $\HWV_{\la^*}(\IC[\overline{Gp}])$ is generated by the functions
\begin{equation}\label{eq:respectS}
g \mapsto \sum_{\beta:\{1,\ldots,Dd\}\to\{1,\ldots,m\}\atop\text{respecting $S$}} \prod_{c=1}^{\la_1}\det(g_{1..\mu_c, \beta(S(1,c))},\ldots,g_{1..\mu_c, \beta(S(\mu_c,c))}),
\end{equation}
where $S$ runs over all standard tableaux of shape $\la$.

Given a standard tableau $S$ we define a tableau $T$ by replacing the first $D$ entries $1,\ldots,D$ by the number $1$,
the next $D$ entries $D+1,\ldots,2D$ by the number $2$, and so on.
It is easy to check that if $T$ is not regular, then the function corresponding to $S$ describes the zero function, because each summand in \eqref{eq:respectS} has a zero factor that is the determinant of a matrix with a repeating column.
We assume from now on that $T$ is regular.
Since $T$ is regular and $S$ is standard, $T$ is semistandard.
We rewrite \eqref{eq:respectS} as follows:

\begin{equation}\label{eq:sumsofprodofdets}
g \mapsto \sum_{\varphi:\{1,\ldots,\delta\}\to\{1,\ldots,m\}} \prod_{c=1}^{\la_1}\det(g_{1..\mu_c, \varphi(T(1,c))},\ldots,g_{1..\mu_c, \varphi(T(\mu_c,c))})
\end{equation}

Using \eqref{eq:gammadet}, we can write this in terms of $\gamma$ as follows:

\[g \mapsto \sum_{\varphi:\{1,\ldots,\delta\}\to\{1,\ldots,m\}} \gamma(g \varphi T).\]

By definition of $M_{\delta,m}$, this can be rewritten as:

\[g \mapsto \gamma(g M_{\delta,m} T),\]
which finishes the proof of the second part of Theorem~\ref{thm:functionsonorbitclosure}.

For the last claim, we note that in this construction of highest weight functions we did not use that $g$ is a square matrix. A rectangular matrix $A$ works in the same way.
\end{proof}

\section{An equation for Waring rank (Proof of the missing part in Prop.~\protect{\ref{pro:alontarsi}})}\label{sec:threetableaux}
In this section we prove the following proposition that was used in the proof of Proposition~\ref{pro:alontarsi}.
\begin{proposition}\label{pro:threetableaux}
Let $\nu = (2m)+m\times 2m$.
Then $a_\nu(2(m+1),m) > 2$.
\end{proposition}
\begin{proof}
By Proposition~\ref{pro:gctandsymmetries} and Lemma~\ref{lem:cormain} we have 
$\mult_{\nu^*} \IC[\overline{Gp}] = \mult_{\nu^*} \IC[{Gp}] = 2.
$
This means that there are two linearly independent HWVs of weight $\nu^*$ that do not vanish on $Gp$.
To finish the proof it suffices to construct a nonzero HWV of weight $\nu^*$ that vanishes on $Gp$, because then these three HWVs are linearly independent.
Note that in particular we construct an equation that vanishes on all polynomials of Waring rank at most $m$.

We use the last part of Theorem~\ref{thm:functionsonorbitclosure} to construct this third function.
Let $n := 2m+2$.
Let $T_\textsf{left}$ be the $m \times (m+2)$ tableau that is filled in a rowwise fashion from top to bottom and from left to right with $m$ many 1s, then $m$ many 2s, and so on, until $m$ many $(m+2)$s.
For example, if $m=6$, then
\[
T_\textsf{left} =
{\ytableausetup{boxsize=1.1em}
\ytableaushort{
11111122,
22223333,
33444444,
55555566,
66667777,
77888888
}}
\]
We remark that the function corresponding to $T_\textsf{left}$ via Theorem~\ref{thm:functionsonorbitclosure} is a generalization of Aronhold's degree 4 invariant on ternary cubics, see also \cite{BI:17} for other (related) generalizations.

Let $T_\textsf{right}$ be the $(m \times 2)+(m^2-2m)$ tableau whose first two columns are equal and
consist of the entries $m+3,m+4,\ldots,2m+2$ from top to bottom. The remaining singleton columns get filled with $m-2$ many entries $m+3$, $m-2$ many entries $m+4$, $m-2$ many entries $m+5$, and so on, until $m-2$ many entries $2m+2$.
For example, if $m=6$, then
\[
T_\textsf{right} =
{\ytableausetup{boxsize=1.1em}
\ytableaushort{
999999{10}{10}{10}{10}{11}{11}{11}{11}{12}{12}{12}{12}{13}{13}{13}{13}{14}{14}{14}{14},
{10}{10},
{11}{11},
{12}{12},
{13}{13},
{14}{14}
}}
\]
The tableau $T$ is defined as the concatenation $T := T_\textsf{left}+T_\textsf{right}$.
We observe that $T$ is duplex.

By Theorem~\ref{thm:functionsonorbitclosure} the function $f : A \mapsto \gamma(A M_{n,n} T)$
is either zero or a HWV of weight $\nu^*$ in
$\IC[\Sym^m \IC^m]$.
\begin{claim}
$f$ does not vanish on $\Sym^m\IC^m$.
\end{claim}
\begin{proof}
This is due to the fact that $T$ is duplex, in complete analogy to \cite{bci:10}.
Choose $A$ to be an $m \times n$ matrix whose entries are real numbers chosen generically (one can alternatively think of the entries being chosen uniformly at random for example from a Gaussian distribution).
Since $T$ is duplex, each summand in $\gamma(A M_{n,n}T)$ is a product of determinants (see \eqref{eq:sumsofprodofdets}), but each factor in the product appears an even number of times and hence the product is nonnegative.
Since $A$ was chosen generically, for the identity map $\textsf{id} \in \mathcal M_{n,n}$ we have $\gamma(A \,\textsf{id}\, T) > 0$. Any finite sum of nonnegative numbers that contains at least one positive number is nonzero, so $\gamma(A M_{n,n}T)$ is nonzero.
This finishes the proof.
\end{proof}
The preceding claim implies that $f$
is a nonzero HWV of weight $\nu^*$ in $\IC[\Sym^m \IC^m]$.
To finish the proof of Proposition~\ref{pro:threetableaux}
it suffices to prove that $f$ vanishes on $Gp$.
The crucial property is that
no tableau in $\mathcal M_{m+2,m}T_\textsf{left}$ is regular:
Since $T_\textsf{left}$ is rectangular with the maximum number of rows,
a regular tableau in $\mathcal M_{m+2,m}T_\textsf{left}$ has $m+2$ many 1s, $m+2$ many 2s, and so on, but every symbol in $\mathcal M_{m+2,m}T_\textsf{left}$ appears a multiple of $m$ many times.
Since no tableau in $\mathcal M_{m+2,m}T_\textsf{left}$ is regular, no tableau in $\mathcal M_{n,m}T$ is regular. Hence all summands in $\gamma(g M_{n,m}T)$ are zero, see Corollary~\ref{cor:gammazero}.
This finishes the proof of Proposition~\ref{pro:threetableaux}.
\end{proof}

\section{Proof of the Main Technical Theorem~\ref{thm:main} using the Tableau Lifting Theorem}\label{sec:proofmain}
In this section we prove Theorem~\ref{thm:main}, based on 
the following combinatorial Tableau Lifting Theorem~\ref{thm:prolongation} whose long and combinatorial proof we will develop during the remaining sections of this paper.
Much simpler forms of other tableau lifting theorems appeared in \cite{KL:12, BIP:19}.

Fix a shape $\la$ and natural numbers $m$ and $e$.
For a tableau $T$ of shape $(m \times e) + \la$ we define $\leftpart(T)$ to be the $m\times e$ rectangular subtableau consisting of the leftmost $e$ columns, and we define $\rightpart(T)$ to be the shape $\la$ subtableau consisting of the rightmost $\la_1$ columns. In particular, we have $T = \leftpart(T)+\rightpart(T)$.
For $\varrho \in (\IN_{\geq 0})^m$ a tableau $S$ has content $D\varrho$ if each number $i$ appears exactly $D\varrho_i$ many times in $S$.

\begin{theorem}[Tableau Lifting Theorem] \label{thm:prolongation}
Let $D \geq 3$ and $m \geq 2$. Given a regular tableau $S$ of shape $\lambda \vdash_m dD$ and content $D \varrho$ for $\varrho\vdash_m d$.
Let
\begin{itemize}
\item $e_\varrho:=\sum_{i=1}^{m} \lceil \frac{\varrho_i}{D-2} \rceil$ in the case where $D$ is even
\item and $e_\varrho:=\sum_{i=1}^{m} 2\lceil \frac{\varrho_i}{2(D-2)} \rceil$ in the case where $D$ is odd and ${{2(D-1)}\choose{D-1}} \geq 2(m-1)$.
\end{itemize}
Let $\delta := m e_\varrho+d$.
In both cases there exists a 
tableau $T:\lambda +(m \times e_\varrho D) \rightarrow \{1,\dots,\delta\}$ in which every entry appears exactly $D$ many times such that
\begin{enumerate}
\item\label{enum:rightpartinSmS} For each $\varphi \in \mathcal{M}_{\delta,m}$ for which $\varphi(T)$ is regular we have that
$\rightpart(\varphi(T))\in\aS_m S$,
\item\label{enum:duplex} For each $\varphi \in \mathcal{M}_{\delta,m}$ for which $\varphi(T)$ is regular we have that $\leftpart(\varphi(T))$ is duplex,
\item\label{enum:existspreimage} there exists $\varphi\in\mathcal{M}_{\delta,m}$ such that $\varphi(T)$ is regular and $\rightpart(\varphi(T))=S$.
\end{enumerate}
\end{theorem}

\begin{proof}[Proof of the Main Theorem~\ref{thm:main}]
By Theorem~\ref{thm:functionsonorbit} and Proposition~\ref{pro:precisedecomposition} we know that there exists a set of regular tableaux $\{S_{\varrho,i}\mid \varrho\in\Xi, \ 1 \leq i \leq b(\la,\varrho,D,d)\}$ of shape $\la$
such that
each $S_{\varrho,i}$ has content $D\varrho$ and
the set of corresponding functions
\[
f^{S_{\varrho,i}}\in \IC[Gp], \quad gp \mapsto \gamma(g P_m S_{\varrho,i}).
\]
is linearly independent.
Since all these functions $f^{S_{\varrho,i}}$ are homogeneous of the same degree $d$, they are not only linearly independent as functions on $\GL_m p$, but also their restrictions to $\SL_m p$ are linearly independent.
Using the Tableau Lifting Theorem~\ref{thm:prolongation},
for each $S_{\varrho,i}$
we construct a tableau $T_{\varrho,i}$ of shape $\la+(m\times e_\varrho D)$ satisfying the properties listed in Theorem~\ref{thm:prolongation}.
We claim that for all $\varrho,i$ there exists $\alpha\neq 0$ such that
\begin{equation}\label{eq:alpha}\tag{$\ast$}
\begin{minipage}{\textwidth-2cm}
under the map $\psi : \varphi\mapsto\rightpart(\varphi T_{\varrho,i})$ that maps from $\mathcal M_{\delta,m}$ to the set of tableaux of shape $\lambda$,
each tableau in $\aS_m S_{\varrho,i}$
has exactly $\alpha$ many preimages in $\mathcal{M}_{\delta,m}$ for which $\varphi T_{\varrho,i}$ is regular.
\end{minipage}
\end{equation}
Proof of \eqref{eq:alpha}:
Clearly $\varphi T_{\varrho,i}$ is regular iff $\pi \varphi T_{\varrho,i}$ is regular.
Note that $\psi$ is $\aS_m$-equivariant in the following sense:
$\psi(\pi \circ \varphi) = \pi\psi(\varphi)$.
Hence taking the preimage $\psi^{-1}$ is also $\aS_m$-equivariant. Thus for all $\hat S \in \aS_m S_{\varrho,i}$ we have
$\varphi \in \psi^{-1}(\hat S)$ iff $\pi \circ \varphi \in \psi^{-1}(\pi \hat S)$.
Thus the application of $\pi$ gives a bijection between the preimages of $\hat S$ and $\pi \hat S$.
To prove claim~\eqref{eq:alpha} it remains to show that $\alpha \neq 0$.
This follows from Thm.~\ref{thm:prolongation}\eqref{enum:existspreimage}.
$\square$

According to \cite[Prop.~3.25]{BI:17} there exists an $\SL_m$-invariant function $\Phi$ in $\IC[\overline{Gp}]$
with $\Phi(p)=1$, called the fundamental invariant.
Moreover,
\[
\begin{cases}
\text{$\Phi$ has degree $m$ and $\Phi(gp) = \det(g)^D p$ for $g \in \GL_m$} & \text{ if $D$ is even} \\
\text{$\Phi$ has degree $2m$ and $\Phi(gp) = \det(g)^{2D} p$ for $g \in \GL_m$} & \text{ if $D$ is odd and $2m \leq \binom{2D}{D}$}.
\end{cases}
\]
Since orbit closures are irreducible varieties,
given two highest weight vectors $f$ of weight $\la$ and $\tilde f$ of weight $\tilde \la$ in $\IC[\overline{Gp}]$, their product $f \cdot \tilde f$ is nonzero and a highest weight vector of weight $\la+\tilde \la$.
Let $\overline{f}^{S_{\varrho,i}}$ be the product of $\Phi^{e_\Xi-e_\varrho}$ ($\Phi^{(e_\Xi-e_\varrho)/2}$ if $D$ is odd) and the function corresponding to $T_{\varrho,i}$:
\[
 \overline{f}^{S_{\varrho,i}}(gp) = \det(g)^{D(e_\Xi-e_\varrho)}\cdot \gamma(g M_{m,d} T_{\varrho,i}) \quad\quad\text{(holds for $D$ even and odd.)}
\]
We claim that $\overline{f}^{S_{\varrho,i}}$ coincides with $f^S_{\varrho,i}$ when restricted to $\SL_m p$ (up to nonzero a factor), which can be seen as follows. For $g \in \SL_m$ we have
\begin{eqnarray*}
 \overline{f}^{S_{\varrho,i}}(gp) &=& 1 \cdot \gamma(g M_{m,d} T_{\varrho,i}) 
 = \sum_{\varphi \in \mathcal M_{m,d}}\gamma(g \varphi T_{\varrho,i}) \\
 &\stackrel{\text{Cor.}~\ref{cor:gammazero}}{=}& \sum_{\substack{\varphi \in \mathcal M_{m,d}\\\varphi T_{\varrho,i}\text{ regular}}}\gamma(g \varphi T_{\varrho,i}) \\
 &\stackrel{\text{Cor.}~\ref{cor:gammafactorization}\text{ and Cor.}~\ref{cor:pm1}\text{ and Thm.}~\ref{thm:prolongation}\eqref{enum:duplex}}{=}&(\pm 1)^2\sum_{\substack{\varphi \in \mathcal M_{m,d}\\\varphi T_{\varrho,i}\text{ regular}}}\gamma(g \rightpart(\varphi T_{\varrho,i})) \\
 &\stackrel{\text{Thm.~}\ref{thm:prolongation}\eqref{enum:rightpartinSmS}\text{ and }\eqref{eq:alpha}}{=}& \alpha \sum_{\hat S \in \aS_m S_{\varrho,i}}\gamma(g \hat S) 
 = \tfrac{|\aS_m S_{\varrho,i}|\alpha}{m!} \gamma(g P_m S_{\varrho,i}) \\
 &=& \tfrac{|\aS_m S_{\varrho,i}|\alpha}{m!} f^{S_{\varrho,i}}(gp).
\end{eqnarray*}
Since each $\overline{f}^{S_{\varrho,i}}$ coincides with $f^{S_{\varrho,i}}$ when restricted to $\SL_m p$ (up to a nonzero factor),
the $\overline f^{S_{\varrho,i}}$ are linearly independent.
It follows from Theorem~\ref{thm:functionsonorbitclosure} that the $\overline{f}^{S_{\varrho,i}}$ are restrictions of functions in $\HWV_{(\la+(m\times e_\Xi D))^*}\IC[\overline{Gp}]$.
This proves the main Theorem \ref{thm:main}.
\end{proof}

It remains to prove the Tableau Lifting Theorem~\ref{thm:prolongation}, whose purely combinatorial proof will be the focus of the rest of the paper.
The construction for odd $D$ is more complicated than for even $D$, which is why we focus on the case where $D$ is even first.

We will construct $\leftpart(T)$ and $\rightpart(T)$ mainly independently.

For even $D$, the alphabet that we are using for $T$ is not $\{1,\ldots,\delta\}$, but a more descriptive alphabet using the symbols $i_\ell$ and $j_k^i$.
For each box $\square$ in $S$, if $\square$ has the entry $i$, then the box corresponding to $\square$ in $\rightpart(T)$ has the symbol $i_\ell$ for some $\ell$.
The only other constraints for $\rightpart(T)$ are concerned with how often the different symbols $i_\ell$ appear.
The symbols $j_k^i$ do not appear in $\rightpart(T)$, but only in $\leftpart(T)$.
The tableau $\leftpart(T)$ is constructed in several steps, starting with a tableau obtained from a set of hypergraphs $H^{(i)}$, and then reordering entries within the rows.

For odd $D$ the situation is similar.
The alphabet that we are using for $T$ is not $\{1,\ldots,\delta\}$, but a more descriptive alphabet using the symbols $i_\ell$, $j_k^i$, $j_{\overline k}^i$.
For each box $\square$ in $S$, if $\square$ has the entry $i$, then the box corresponding to $\square$ in $\rightpart(T)$ has the symbol $i_\ell$ for some $\ell$.
The only other constraints for $\rightpart(T)$ are concerned with how often the different symbols $i_\ell$ appear.
Symbols $j_k^i$ and $j_{\overline k}^i$ do not appear in $\rightpart(T)$, but only in $\leftpart(T)$.
The tableau $\leftpart(T)$ is constructed in several steps, starting with a tableau obtained from a set of hypergraphs $H^{(i)}$ (similar to those hypergraphs in the case where $D$ is even), and then reordering entries within the rows.

The construction for odd $D$ has many more subtleties than the construction for even $D$ and there are numerous slight differences between the two cases. We think that the readability would suffer greatly
if we did not explain the whole construction again for odd $D$ in a self-contained manner, including the parts that are very similar to the even case.
Therefore in the following sections we first treat the case for even $D$ and then treat the case for odd $D$ in a fairly self-contained manner. The reader will see that much more care and attention to the details is necessary in the case where $D$ is odd.

\section{The hypergraphs $H^{(i)}$ for even $D$}\label{EVENsec:hypergraphsevenD}
Let $I := \{i \mid \varrho_i \neq 0\}$.
In order to construct $\leftpart(T)$ we will first
constuct a so-called $(D,\varrho_i)$-hypergraph for each $i \in I$. We refer to that hypergraph as $H^{(i)}$.
The number of columns in $\leftpart(T)$ will precisely
be the number of vertices in all these hypergraphs together. So we want the hypergraphs to be as small as possible.

We first recall some basic terms.
Let $H = (V,E)$ by a hypergraph and $e \in E$ be an edge. Then we define the \textit{size} of an edge $\size(e)$ as the number of vertices in $e$. 
Let $v,w \in V$. Then a \textit{path} between $v$ and $w$ is a sequence of edges $(e_1,e_2,\ldots,e_l)$ such that $v \in e_1$, $w \in e_l$, and $e_i \cap e_{i+1} \neq \emptyset$.
We say that two vertices are connected in $H$ iff there exists a path between them.
We say that a hypergraph is \textit{connected} iff every pair of vertices is connected.
For a nonempty set $X$ a \textit{set partition} $P$ of $X$ is a set of pairwise disjoint subsets whose union is $X$.
\begin{definition} \label{def:dkhypergraph}
         Let $D,K$ be integers. A \textit{$(D,K)$-hypergraph} is defined to be a hypergraph $H=(V,E)$ that satisfies the following properties:
        \begin{enumerate}[label={(\arabic*)}]
                \item $H$ is connected. \label{eq:def:connected}
                \item $H$ has two different types of hyperedges: the \textit{block edges} and the \textit{name edges}. \label{eq:def:blockandnameedges}
                \item Each block edge has size $D$, and the set of block edges $E_{\text{Block}} \subseteq E$ is a set partition of $V$.
                \label{eq:def:blockedgessizeD}
                \item Each name edge has size strictly less than $D$, but at least size $1$, and
                the set of name edges $E_{\text{Name}} \subseteq E$ is a set partition of $V$.
                \label{eq:def:nameedgessizelessthanD}
                \item $|E_{\text{Name}}| - |E_{\text{Block}}| = K$.\label{eq:def:nameblockdifference}
                \item 
                There exists a name edge $e_{\text{Name}}$ and a block edge $e_{\text{Block}}$ whose intersection contains at least 2 vertices. We choose one of these two vertices and call it the \emph{link vertex}.
                \label{prop:hypergraphshareblockname}
        \end{enumerate}
\end{definition}

Several examples for $(6,K)$-hypergraphs are given in Figure~\ref{EVENfig:6Khypergraphs}.
\begin{figure}[htb]
\scalebox{1.4}{
\begin{tikzpicture}
[
scale=0.52,
n/.style={circle,fill=black,inner sep=0pt, minimum size = 0.2cm},
nlink/.style={diamond,fill=white,draw=black,inner sep=0pt, minimum size = 0.2cm},
be/.style={fill=lightgray, rounded corners, inner sep=0.05cm, dashed},        
se/.style={draw, rounded corners, inner sep=0.1cm},        
]
\node at (-2,0) {\scalebox{0.71}{$K=1:$}};
\node [nlink] (1) at (0,0) {};
\node [n] (2) at (1,0) {};
\node [n] (3) at (2,0) {};
\node [n] (4) at (3,0) {};
\node [n] (5) at (4,0) {};
\node [n] (6) at (5,0) {};
\node [be, fit= (1) (2) (3) (4) (5) (6)] {};
\node [se, fit= (1) (2) (3) (4) (5)] {};
\node [se, fit= (6)] {};
\node [nlink] (1) at (0,0) {};
\node [n] (2) at (1,0) {};
\node [n] (3) at (2,0) {};
\node [n] (4) at (3,0) {};
\node [n] (5) at (4,0) {};
\node [n] (6) at (5,0) {};
\end{tikzpicture}
}

\smallskip

\scalebox{1.4}{
\begin{tikzpicture}
[
scale=0.52,
n/.style={circle,fill=black,inner sep=0pt, minimum size = 0.2cm},
nlink/.style={diamond,fill=white,draw=black,inner sep=0pt, minimum size = 0.2cm},
be/.style={fill=lightgray, rounded corners, inner sep=0.05cm, dashed},        
se/.style={draw, rounded corners, inner sep=0.1cm},        
]
\node at (-2,0) {\scalebox{0.71}{$K=2:$}};
\node [nlink] (1) at (0,0) {};
\node [n] (2) at (1,0) {};
\node [n] (3) at (2,0) {};
\node [n] (4) at (3,0) {};
\node [n] (5) at (4,0) {};
\node [n] (6) at (5,0) {};
\node [be, fit= (1) (2) (3) (4) (5) (6)] {};
\node [se, fit= (1) (2) (3) (4)] {};
\node [se, fit= (5)] {};
\node [se, fit= (6)] {};
\node [nlink] (1) at (0,0) {};
\node [n] (2) at (1,0) {};
\node [n] (3) at (2,0) {};
\node [n] (4) at (3,0) {};
\node [n] (5) at (4,0) {};
\node [n] (6) at (5,0) {};
\end{tikzpicture}
}

\smallskip

\scalebox{1.4}{
\begin{tikzpicture}
[
scale=0.52,
n/.style={circle,fill=black,inner sep=0pt, minimum size = 0.2cm},
nlink/.style={diamond,fill=white,draw=black,inner sep=0pt, minimum size = 0.2cm},
be/.style={fill=lightgray, rounded corners, inner sep=0.05cm, dashed},        
se/.style={draw, rounded corners, inner sep=0.1cm},        
]
\node at (-2,0) {\scalebox{0.71}{$K=3:$}};
\node [nlink] (1) at (0,0) {};
\node [n] (2) at (1,0) {};
\node [n] (3) at (2,0) {};
\node [n] (4) at (3,0) {};
\node [n] (5) at (4,0) {};
\node [n] (6) at (5,0) {};
\node [be, fit= (1) (2) (3) (4) (5) (6)] {};
\node [se, fit= (1) (2) (3)] {};
\node [se, fit= (4)] {};
\node [se, fit= (5)] {};
\node [se, fit= (6)] {};
\node [nlink] (1) at (0,0) {};
\node [n] (2) at (1,0) {};
\node [n] (3) at (2,0) {};
\node [n] (4) at (3,0) {};
\node [n] (5) at (4,0) {};
\node [n] (6) at (5,0) {};
\end{tikzpicture}
}

\smallskip

\scalebox{1.4}{
\begin{tikzpicture}
[
scale=0.52,
n/.style={circle,fill=black,inner sep=0pt, minimum size = 0.2cm},
nlink/.style={diamond,fill=white,draw=black,inner sep=0pt, minimum size = 0.2cm},
be/.style={fill=lightgray, rounded corners, inner sep=0.05cm, dashed},        
se/.style={draw, rounded corners, inner sep=0.1cm},        
]
\node at (-2,0) {\scalebox{0.71}{$K=4:$}};
\node [nlink] (1) at (0,0) {};
\node [n] (2) at (1,0) {};
\node [n] (3) at (2,0) {};
\node [n] (4) at (3,0) {};
\node [n] (5) at (4,0) {};
\node [n] (6) at (5,0) {};
\node [be, fit= (1) (2) (3) (4) (5) (6)] {};
\node [se, fit= (1) (2)] {};
\node [se, fit= (3)] {};
\node [se, fit= (4)] {};
\node [se, fit= (5)] {};
\node [se, fit= (6)] {};
\node [nlink] (1) at (0,0) {};
\node [n] (2) at (1,0) {};
\node [n] (3) at (2,0) {};
\node [n] (4) at (3,0) {};
\node [n] (5) at (4,0) {};
\node [n] (6) at (5,0) {};
\end{tikzpicture}
}

\smallskip

\scalebox{1.4}{
\begin{tikzpicture}
[
scale=0.52,
n/.style={circle,fill=black,inner sep=0pt, minimum size = 0.2cm},
nlink/.style={diamond,fill=white,draw=black,inner sep=0pt, minimum size = 0.2cm},
be/.style={fill=lightgray, rounded corners, inner sep=0.05cm, dashed},        
se/.style={draw, rounded corners, inner sep=0.1cm},        
]
\node at (-2,0) {\scalebox{0.71}{$K=5:$}};
\node [nlink] (1) at (0,0) {};
\node [n] (2) at (1,0) {};
\node [n] (3) at (2,0) {};
\node [n] (4) at (3,0) {};
\node [n] (5) at (4,0) {};
\node [n] (6) at (5,0) {};
\node [n] (7) at (6,0) {};
\node [n] (8) at (7,0) {};
\node [n] (9) at (8,0) {};
\node [n] (10) at (9,0) {};
\node [n] (11) at (10,0) {};
\node [n] (12) at (11,0) {};
\node [be, fit= (1) (2) (3) (4) (5) (6)] {};
\node [be, fit= (7) (8) (9) (10) (11) (12)] {};
\node [se, fit= (1) (2) (3) (4) (5)] {};
\node [se, fit= (6) (7)] {};
\node [se, fit= (8)] {};
\node [se, fit= (9)] {};
\node [se, fit= (10)] {};
\node [se, fit= (11)] {};
\node [se, fit= (12)] {};
\node [nlink] (1) at (0,0) {};
\node [n] (2) at (1,0) {};
\node [n] (3) at (2,0) {};
\node [n] (4) at (3,0) {};
\node [n] (5) at (4,0) {};
\node [n] (6) at (5,0) {};
\node [n] (7) at (6,0) {};
\node [n] (8) at (7,0) {};
\node [n] (9) at (8,0) {};
\node [n] (10) at (9,0) {};
\node [n] (11) at (10,0) {};
\node [n] (12) at (11,0) {};
\end{tikzpicture}
}

\smallskip

\scalebox{1.4}{
\begin{tikzpicture}
[
scale=0.52,
n/.style={circle,fill=black,inner sep=0pt, minimum size = 0.2cm},
nlink/.style={diamond,fill=white,draw=black,inner sep=0pt, minimum size = 0.2cm},
be/.style={fill=lightgray, rounded corners, inner sep=0.05cm, dashed},        
se/.style={draw, rounded corners, inner sep=0.1cm},        
]
\node at (-2,0) {\scalebox{0.71}{$K=6:$}};
\node [nlink] (1) at (0,0) {};
\node [n] (2) at (1,0) {};
\node [n] (3) at (2,0) {};
\node [n] (4) at (3,0) {};
\node [n] (5) at (4,0) {};
\node [n] (6) at (5,0) {};
\node [n] (7) at (6,0) {};
\node [n] (8) at (7,0) {};
\node [n] (9) at (8,0) {};
\node [n] (10) at (9,0) {};
\node [n] (11) at (10,0) {};
\node [n] (12) at (11,0) {};
\node [be, fit= (1) (2) (3) (4) (5) (6)] {};
\node [be, fit= (7) (8) (9) (10) (11) (12)] {};
\node [se, fit= (1) (2) (3) (4)] {};
\node [se, fit= (5)] {};
\node [se, fit= (6) (7)] {};
\node [se, fit= (8)] {};
\node [se, fit= (9)] {};
\node [se, fit= (10)] {};
\node [se, fit= (11)] {};
\node [se, fit= (12)] {};
\node [nlink] (1) at (0,0) {};
\node [n] (2) at (1,0) {};
\node [n] (3) at (2,0) {};
\node [n] (4) at (3,0) {};
\node [n] (5) at (4,0) {};
\node [n] (6) at (5,0) {};
\node [n] (7) at (6,0) {};
\node [n] (8) at (7,0) {};
\node [n] (9) at (8,0) {};
\node [n] (10) at (9,0) {};
\node [n] (11) at (10,0) {};
\node [n] (12) at (11,0) {};
\end{tikzpicture}
}

\smallskip

\scalebox{1.4}{
\begin{tikzpicture}
[
scale=0.52,
n/.style={circle,fill=black,inner sep=0pt, minimum size = 0.2cm},
nlink/.style={diamond,fill=white,draw=black,inner sep=0pt, minimum size = 0.2cm},
be/.style={fill=lightgray, rounded corners, inner sep=0.05cm, dashed},        
se/.style={draw, rounded corners, inner sep=0.1cm},        
]
\node at (-2,0) {\scalebox{0.71}{$K=7:$}};
\node [nlink] (1) at (0,0) {};
\node [n] (2) at (1,0) {};
\node [n] (3) at (2,0) {};
\node [n] (4) at (3,0) {};
\node [n] (5) at (4,0) {};
\node [n] (6) at (5,0) {};
\node [n] (7) at (6,0) {};
\node [n] (8) at (7,0) {};
\node [n] (9) at (8,0) {};
\node [n] (10) at (9,0) {};
\node [n] (11) at (10,0) {};
\node [n] (12) at (11,0) {};
\node [be, fit= (1) (2) (3) (4) (5) (6)] {};
\node [be, fit= (7) (8) (9) (10) (11) (12)] {};
\node [se, fit= (1) (2) (3)] {};
\node [se, fit= (4)] {};
\node [se, fit= (5)] {};
\node [se, fit= (6) (7)] {};
\node [se, fit= (8)] {};
\node [se, fit= (9)] {};
\node [se, fit= (10)] {};
\node [se, fit= (11)] {};
\node [se, fit= (12)] {};
\node [nlink] (1) at (0,0) {};
\node [n] (2) at (1,0) {};
\node [n] (3) at (2,0) {};
\node [n] (4) at (3,0) {};
\node [n] (5) at (4,0) {};
\node [n] (6) at (5,0) {};
\node [n] (7) at (6,0) {};
\node [n] (8) at (7,0) {};
\node [n] (9) at (8,0) {};
\node [n] (10) at (9,0) {};
\node [n] (11) at (10,0) {};
\node [n] (12) at (11,0) {};
\end{tikzpicture}
}

\smallskip

\scalebox{1.4}{
\begin{tikzpicture}
[
scale=0.52,
n/.style={circle,fill=black,inner sep=0pt, minimum size = 0.2cm},
nlink/.style={diamond,fill=white,draw=black,inner sep=0pt, minimum size = 0.2cm},
be/.style={fill=lightgray, rounded corners, inner sep=0.05cm, dashed},        
se/.style={draw, rounded corners, inner sep=0.1cm},        
]
\node at (-2,0) {\scalebox{0.71}{$K=8:$}};
\node [nlink] (1) at (0,0) {};
\node [n] (2) at (1,0) {};
\node [n] (3) at (2,0) {};
\node [n] (4) at (3,0) {};
\node [n] (5) at (4,0) {};
\node [n] (6) at (5,0) {};
\node [n] (7) at (6,0) {};
\node [n] (8) at (7,0) {};
\node [n] (9) at (8,0) {};
\node [n] (10) at (9,0) {};
\node [n] (11) at (10,0) {};
\node [n] (12) at (11,0) {};
\node [be, fit= (1) (2) (3) (4) (5) (6)] {};
\node [be, fit= (7) (8) (9) (10) (11) (12)] {};
\node [se, fit= (1) (2)] {};
\node [se, fit= (3)] {};
\node [se, fit= (4)] {};
\node [se, fit= (5)] {};
\node [se, fit= (6) (7)] {};
\node [se, fit= (8)] {};
\node [se, fit= (9)] {};
\node [se, fit= (10)] {};
\node [se, fit= (11)] {};
\node [se, fit= (12)] {};
\node [nlink] (1) at (0,0) {};
\node [n] (2) at (1,0) {};
\node [n] (3) at (2,0) {};
\node [n] (4) at (3,0) {};
\node [n] (5) at (4,0) {};
\node [n] (6) at (5,0) {};
\node [n] (7) at (6,0) {};
\node [n] (8) at (7,0) {};
\node [n] (9) at (8,0) {};
\node [n] (10) at (9,0) {};
\node [n] (11) at (10,0) {};
\node [n] (12) at (11,0) {};
\end{tikzpicture}
}

\smallskip

\scalebox{1.4}{
\begin{tikzpicture}
[
scale=0.52,
n/.style={circle,fill=black,inner sep=0pt, minimum size = 0.2cm},
nlink/.style={diamond,fill=white,draw=black,inner sep=0pt, minimum size = 0.2cm},
be/.style={fill=lightgray, rounded corners, inner sep=0.05cm, dashed},        
se/.style={draw, rounded corners, inner sep=0.1cm},        
]
\node at (-2,0) {\scalebox{0.71}{$K=9:$}};
\node [nlink] (1) at (0,0) {};
\node [n] (2) at (1,0) {};
\node [n] (3) at (2,0) {};
\node [n] (4) at (3,0) {};
\node [n] (5) at (4,0) {};
\node [n] (6) at (5,0) {};
\node [n] (7) at (6,0) {};
\node [n] (8) at (7,0) {};
\node [n] (9) at (8,0) {};
\node [n] (10) at (9,0) {};
\node [n] (11) at (10,0) {};
\node [n] (12) at (11,0) {};
\node [n] (13) at (12,0) {};
\node [n] (14) at (13,0) {};
\node [n] (15) at (14,0) {};
\node [n] (16) at (15,0) {};
\node [n] (17) at (16,0) {};
\node [n] (18) at (17,0) {};
\node [be, fit= (1) (2) (3) (4) (5) (6)] {};
\node [be, fit= (7) (8) (9) (10) (11) (12)] {};
\node [be, fit= (13) (14) (15) (16) (17) (18)] {};
\node [se, fit= (1) (2) (3) (4) (5)] {};
\node [se, fit= (6) (7)] {};
\node [se, fit= (8)] {};
\node [se, fit= (9)] {};
\node [se, fit= (10)] {};
\node [se, fit= (11)] {};
\node [se, fit= (12) (13)] {};
\node [se, fit= (14)] {};
\node [se, fit= (15)] {};
\node [se, fit= (16)] {};
\node [se, fit= (17)] {};
\node [se, fit= (18)] {};
\node [nlink] (1) at (0,0) {};
\node [n] (2) at (1,0) {};
\node [n] (3) at (2,0) {};
\node [n] (4) at (3,0) {};
\node [n] (5) at (4,0) {};
\node [n] (6) at (5,0) {};
\node [n] (7) at (6,0) {};
\node [n] (8) at (7,0) {};
\node [n] (9) at (8,0) {};
\node [n] (10) at (9,0) {};
\node [n] (11) at (10,0) {};
\node [n] (12) at (11,0) {};
\node [n] (13) at (12,0) {};
\node [n] (14) at (13,0) {};
\node [n] (15) at (14,0) {};
\node [n] (16) at (15,0) {};
\node [n] (17) at (16,0) {};
\node [n] (18) at (17,0) {};
\end{tikzpicture}
}
       \caption{$(6,K)$-hypergraphs for several values of $K$.
       Each block edges is represented as a filled gray blob.
       Each name edges is enclosed by a solid black curve.
       The link vertices are drawn as white diamonds.}
        \label{EVENfig:6Khypergraphs}
\end{figure}
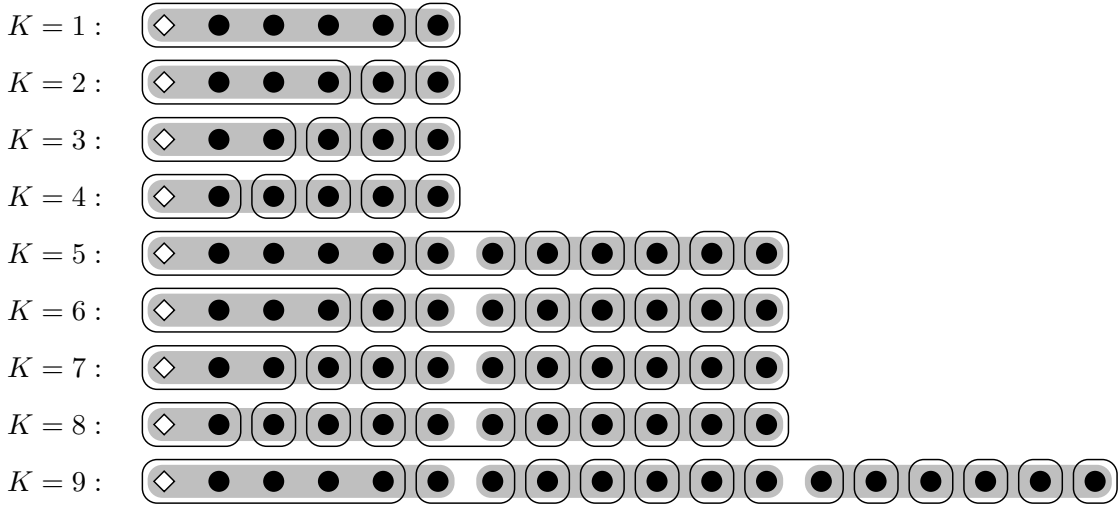

\begin{proposition}\label{EVENprop:hypergraph}
For even $D \geq 4$, $K\neq 0$, there exists a $(D,K)$-hypergraph that has exactly $\lceil \frac{K}{D-2} \rceil$ many block edges.
\end{proposition}
\begin{proof}
Let $n := \lceil \frac{K}{(D-2)} \rceil$.
We start the construction by considering $n$ many disjoint block edges with $D$ many vertices each.
We arrange the vertices in a linear fashion as in Figure~\ref{EVENfig:6Khypergraphs}.
The leftmost vertex is the link vertex.
We now place the vertices in $K+n$ many name edges as follows.
As in Figure~\ref{EVENfig:6Khypergraphs}, the rightmost vertex of every block edge but the last shall be placed in a size 2 name edge with the leftmost vertex of the next block edge.
The resulting hypergraph is connected.
At this point we have $nD-2(n-1)=n(D-2)+2$ vertices that are not in name edges yet; and we have $K+n-(n-1) = K+1$ name edges left to put vertices in. Since
\[
(n(D-2)+2)-(K+1) =
(\lceil \tfrac{K}{(D-2)} \rceil(D-2)+2)-(K+1)
\geq
(K+2)-(K+1) = 1 > 0,
\]
we can position the name edges so that the link vertex is in a name edge of size at least $2$.
Moreover, we position that name edge of size at least 2 in such a way that the link vertex has a vertex that not only lies in the same name edge, but also in the same block edge.
\end{proof}

\section{Construction of $\leftpart(T)$ for even $D$}
\label{EVENsec:rightpartevenD}
For each $i \in I$ let $H^{(i)}$ be a $(D,\varrho_i)$-hypergraph from Proposition~\ref{EVENprop:hypergraph}. We write $E_{\text{Block}}^{(i)}$ to denote its set of block edges and $E_{\text{Name}}^{(i)}$ to denote its set of name edges.

In this section, for every $i \in I$ and every $e \in E_{\text{Block}}^{(i)}$ we construct an $m \times D$ block tableau $\check B_e$ such that
$\leftpart(T)$ is constructed as the concatenation
\begin{equation}\label{EVENeq:leftpartmadefromblocks}
\leftpart(T) := \sum_{i\in I} \sum_{e \in E_{\text{Block}}^{(i)}} \check B_e.
\end{equation}
Notice that since every block edge has size $D$ (see Def.~\ref{def:dkhypergraph}\eqref{eq:def:blockedgessizeD}), this implies that the number of columns in $\leftpart(T)$ is equal to the total number of vertices in the hypergraphs $H^{(i)}$, $i \in I$.

Each $m\times D$ block tableau $\check B_e$ is constructed in three steps: First we construct an $m\times D$ block tableau $B_e$ in which each column corresponds to a vertex in $H^{(i)}$, then we exchange entries between columns that correspond to link vertices.

Let $\zeta^{(i)}$ denote the link vertex in $H^{(i)}$.
We attach some additional data to each $H^{(i)}$ as follows.
We put a linear order on the set of name edges $E_{\text{name}}^{(i)}$
and for each vertex $v$ in $H^{(i)}$ we define $\ell(v)$ to be the index of its corresponding name edge. Here $\ell(v)=1$ if $v$ lies in the first name edge, $\ell(v)=2$ for the next name edge, and so on.
We ensure that
\begin{equation}\label{EVENeq:ellzetaione}
\ell(\zeta^{(i)})=1.
\end{equation}
In the same way, we put a linear order on the set of block edges;
for each block edge $e$ we write $k(e)$ for its index
and for each vertex $v$ in $H^{(i)}$ we define $k(v)$ to be the index of its corresponding block edge.
We ensure that
\begin{equation}\label{EVENeq:kzetaione}
k(\zeta^{(i)})=1.
\end{equation}
Moreover, for every vertex $v$ in any $H^{(i)}$ we define $i(v):=i$.

In the following, for each vertex $v$ we define an $m \times 1$ rectangular tableau (i.e., a column of length $m$) called $B_v$.
Concatenating them results in $B_e := \sum_{v \in e} B_v$.
The order of columns does not matter, but it is convenient to have the vertices of $H^{(i)}$ ordered from left to right in the same way as the columns of $\leftpart(T)$.
Later we define $\check B_e$, from which we can extract $\check B_v$ for $v \in e$ as follows:
If $B_v$ is the $n$-th column of $B_e$, then $\check B_v$ is the $n$-th column of $\check B_e$.

\subsection*{Starting with $B$}
Let $e$ be in the $k$-th block edge in $H^{(i)}$ and let $v \in e$.
The column $B_v$ is defined by the following properties.
\begin{eqnarray}
 \text{ the $i$-th entry of } B_v \text{ is } i_{\ell(v)} \label{EVENeq:niBv}\\
 \text{ the $j$-th entry ($j \neq i$) of } B_v \text{ is } j_{k}^{i}
 \label{EVENeq:njBv}
\end{eqnarray}
An example is given in Figure~\ref{EVENfig:exampleB}.

\begin{figure}
\scalebox{1.3}{
\begin{tikzpicture}
\node at (0,0) {
\begin{tikzpicture}
[
scale=0.52,
n/.style={circle,fill=black,inner sep=0pt, minimum size = 0.2cm},
nlink/.style={diamond,fill=white,draw=black,inner sep=0pt, minimum size = 0.2cm},
be/.style={fill=lightgray, rounded corners, inner sep=0.05cm, dashed},        
se/.style={draw, rounded corners, inner sep=0.1cm},        
]
\node [nlink] (1) at (0,0) {};
\node [n] (2) at (1,0) {};
\node [n] (3) at (2,0) {};
\node [n] (4) at (3,0) {};
\node [n] (5) at (4,0) {};
\node [n] (6) at (5,0) {};
\node [n] (7) at (6,0) {};
\node [n] (8) at (7,0) {};
\node [n] (9) at (8,0) {};
\node [n] (10) at (9,0) {};
\node [n] (11) at (10,0) {};
\node [n] (12) at (11,0) {};
\node [be, fit= (1) (2) (3) (4) (5) (6)] {};
\node [be, fit= (7) (8) (9) (10) (11) (12)] {};
\node [se, fit= (1) (2) (3) (4)] {};
\node [se, fit= (5)] {};
\node [se, fit= (6) (7)] {};
\node [se, fit= (8)] {};
\node [se, fit= (9)] {};
\node [se, fit= (10)] {};
\node [se, fit= (11)] {};
\node [se, fit= (12)] {};
\node [nlink] (1) at (0,0) {};
\node [n] (2) at (1,0) {};
\node [n] (3) at (2,0) {};
\node [n] (4) at (3,0) {};
\node [n] (5) at (4,0) {};
\node [n] (6) at (5,0) {};
\node [n] (7) at (6,0) {};
\node [n] (8) at (7,0) {};
\node [n] (9) at (8,0) {};
\node [n] (10) at (9,0) {};
\node [n] (11) at (10,0) {};
\node [n] (12) at (11,0) {};
\end{tikzpicture}
}; 
\node at (-0.05,-1.4) {
\ytableausetup{boxsize=1.3em}
\ytableaushort{
{1^2_1}{1^2_1}{1^2_1}{1^2_1}{1^2_1}{1^2_1}
{1^2_2}{1^2_2}{1^2_2}{1^2_2}{1^2_2}{1^2_2}
,
{*(lightgray)2_1}{*(lightgray)2_1}{*(lightgray)2_1}{*(lightgray)2_1}{*(lightgray)2_2}{*(lightgray)2_3}{*(lightgray)2_3}{*(lightgray)2_4}{*(lightgray)2_5}{*(lightgray)2_6}{*(lightgray)2_7}{*(lightgray)2_8}
,
{3^2_1}{3^2_1}{3^2_1}{3^2_1}{3^2_1}{3^2_1}
{3^2_2}{3^2_2}{3^2_2}{3^2_2}{3^2_2}{3^2_2}
,
{4^2_1}{4^2_1}{4^2_1}{4^2_1}{4^2_1}{4^2_1}
{4^2_2}{4^2_2}{4^2_2}{4^2_2}{4^2_2}{4^2_2}
}
};
\end{tikzpicture}
}
\caption{
Here $i=2$.
A $(6,6)$-hypergraph $H^{(2)}$ with two block edges $e_1$ and $e_2$ and the corresponding concatenated tableau $B_{e_1} + B_{e_2}$.
Vertices are drawn directly above their corresponding columns.
To make the value of $i$ easy to see, the second row is highlighted.
}
\label{EVENfig:exampleB}
\end{figure}
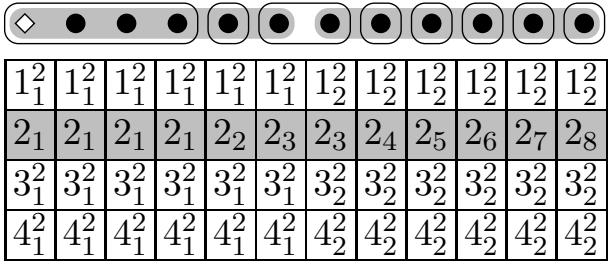

\subsection*{From $B$ to $\check B$}
Most columns $B_v$ and $\check B_v$ coincide, as we define
$\check B_v := B_v$ if $v \notin \{\zeta^{(i)} \mid i \in I\}$.
This means that only the columns corresponding to link vertices are adjusted.

Let $h$ denote the smallest number in $I$.
For $i \in I$, $i \neq h$, the column $\check B_{\zeta^{(i)}}$ arises from $B_{\zeta^{(i)}}$ by switching the $i$-th entry with the $i$-th entry in $B_{\zeta^{(h)}}$. This means that the column $\check B_{\zeta^{(h)}}$ arises from $B_{\zeta^{(i)}}$ by switching the $i$-th entry with the $i$-th entry in $B_{\zeta^{(i)}}$ for all $i \in I$.

The columnwise description of $\check B_v$ thus as follows.
\begin{itemize}
 \item If $v$ is not a link vertex, then
\begin{eqnarray}
 \text{ the $i$-th entry of } \check B_v \text{ is } i_{\ell(v)} \label{EVENeq:iBv}\\
 \text{ the $j$-th entry ($j \neq i$) of } \check B_v \text{ is } j_{k(v)}^{(i)} \label{EVENeq:jBv}
\end{eqnarray}
 \item If $i \neq h$, then
\begin{eqnarray}
 \text{ the $i$-th entry of } \check B_{\zeta^{(i)}} \text{ is } i_{k(v)}^h \label{EVENeq:iBzetaioneh}\\
 \text{ the $j$-th entry ($j \neq i$) of } \check B_{\zeta^{(i)}} \text{ is } j_{k(v)}^i \label{EVENeq:iBzetajoneh}
\end{eqnarray}
\item Moreover,
\begin{eqnarray}
 \text{ the $h$-th entry of } \check B_{\zeta^{(h)}} \text{ is } h_1 \label{EVENeq:hBzetah}\\
 \text{ for $j \neq h$, $j \in I$, the $j$-th entry of } \check B_{\zeta^{(h)}} \text{ is } j_1 \label{EVENeq:jBzetahneq}\\
 \text{ for $j \neq h$, $j \notin I$, the $j$-th entry of } \check B_{\zeta^{(h)}} \text{ is } j_{1}^h\label{EVENeq:jBzetaheq}
\end{eqnarray}
\end{itemize}

An example is provided in Figure~\ref{EVENfig:examplecheckB}.

\begin{figure}
\begin{tikzpicture}
\node at (0.05,0) {
\scalebox{1.2}{
\begin{tikzpicture}
[
scale=0.505,
n/.style={circle,fill=black,inner sep=0pt, minimum size = 0.2cm},
nlink/.style={diamond,fill=white,draw=black,inner sep=0pt, minimum size = 0.2cm},
be/.style={fill=lightgray, rounded corners, inner sep=0.05cm, dashed},        
se/.style={draw, rounded corners, inner sep=0.1cm},        
]
\node [nlink] (1) at (0,0) {};
\node [n] (2) at (1,0) {};
\node [n] (3) at (2,0) {};
\node [n] (4) at (3,0) {};
\node [n] (5) at (4,0) {};
\node [n] (6) at (5,0) {};
\node [n] (7) at (6,0) {};
\node [n] (8) at (7,0) {};
\node [n] (9) at (8,0) {};
\node [n] (10) at (9,0) {};
\node [n] (11) at (10,0) {};
\node [n] (12) at (11,0) {};
\node [nlink] (13) at (12,0) {};
\node [n] (14) at (13,0) {};
\node [n] (15) at (14,0) {};
\node [n] (16) at (15,0) {};
\node [n] (17) at (16,0) {};
\node [n] (18) at (17,0) {};
\node [nlink] (19) at (18,0) {};
\node [n] (20) at (19,0) {};
\node [n] (21) at (20,0) {};
\node [n] (22) at (21,0) {};
\node [n] (23) at (22,0) {};
\node [n] (24) at (23,0) {};
\node [be, fit= (1) (2) (3) (4) (5) (6)] {};
\node [be, fit= (7) (8) (9) (10) (11) (12)] {};
\node [be, fit= (13) (14) (15) (16) (17) (18)] {};
\node [be, fit= (19) (20) (21) (22) (23) (24)] {};
\node [se, fit= (1) (2) (3) (4) (5)] {};
\node [se, fit= (6) (7)] {};
\node [se, fit= (8)] {};
\node [se, fit= (9)] {};
\node [se, fit= (10)] {};
\node [se, fit= (11)] {};
\node [se, fit= (12)] {};
\node [se, fit= (13) (14) (15) (16)] {};
\node [se, fit= (17)] {};
\node [se, fit= (18)] {};
\node [se, fit= (19) (20) (21) (22) (23)] {};
\node [se, fit= (24)] {};
\node [nlink] (1) at (0,0) {};
\node [n] (2) at (1,0) {};
\node [n] (3) at (2,0) {};
\node [n] (4) at (3,0) {};
\node [n] (5) at (4,0) {};
\node [n] (6) at (5,0) {};
\node [n] (7) at (6,0) {};
\node [n] (8) at (7,0) {};
\node [n] (9) at (8,0) {};
\node [n] (10) at (9,0) {};
\node [n] (11) at (10,0) {};
\node [n] (12) at (11,0) {};
\node [nlink] (13) at (12,0) {};
\node [n] (14) at (13,0) {};
\node [n] (15) at (14,0) {};
\node [n] (16) at (15,0) {};
\node [n] (17) at (16,0) {};
\node [n] (18) at (17,0) {};
\node [nlink] (19) at (18,0) {};
\node [n] (20) at (19,0) {};
\node [n] (21) at (20,0) {};
\node [n] (22) at (21,0) {};
\node [n] (23) at (22,0) {};
\node [n] (24) at (23,0) {};
\end{tikzpicture}
}
}; 
\node at (-0.05,-1.6) {
\ytableausetup{boxsize=1.3em}
\scalebox{1.2}{
\ytableaushort{
{*(lightgray)1_1}{*(lightgray)1_1}{*(lightgray)1_1}{*(lightgray)1_1}{*(lightgray)1_1}{*(lightgray)1_2}{*(lightgray)1_2}{*(lightgray)1_3}{*(lightgray)1_4}{*(lightgray)1_5}{*(lightgray)1_6}{*(lightgray)1_7}
{1^2_1}{1^2_1}{1^2_1}{1^2_1}{1^2_1}{1^2_1}
{1^3_1}{1^3_1}{1^3_1}{1^3_1}{1^3_1}{1^3_1}
,
{2^1_1}{2^1_1}{2^1_1}{2^1_1}{2^1_1}{2^1_1}
{2^1_2}{2^1_2}{2^1_2}{2^1_2}{2^1_2}{2^1_2}
{*(lightgray)2_1}{*(lightgray)2_1}{*(lightgray)2_1}{*(lightgray)2_1}{*(lightgray)2_2}{*(lightgray)2_3}
{2^3_1}{2^3_1}{2^3_1}{2^3_1}{2^3_1}{2^3_1}
,
{3^1_1}{3^1_1}{3^1_1}{3^1_1}{3^1_1}{3^1_1}
{3^1_2}{3^1_2}{3^1_2}{3^1_2}{3^1_2}{3^1_2}
{3^2_1}{3^2_1}{3^2_1}{3^2_1}{3^2_1}{3^2_1}
{*(lightgray)3_1}{*(lightgray)3_1}{*(lightgray)3_1}{*(lightgray)3_1}{*(lightgray)3_1}{*(lightgray)3_2}
,
{4^1_1}{4^1_1}{4^1_1}{4^1_1}{4^1_1}{4^1_1}
{4^1_2}{4^1_2}{4^1_2}{4^1_2}{4^1_2}{4^1_2}
{4^2_1}{4^2_1}{4^2_1}{4^2_1}{4^2_1}{4^2_1}
{4^3_1}{4^3_1}{4^3_1}{4^3_1}{4^3_1}{4^3_1}
}
}
};
\node at (-0.05,-4.5) {
\ytableausetup{boxsize=1.3em}
\scalebox{1.2}{
\ytableaushort{
{*(lightgray)1_1}{*(lightgray)1_1}{*(lightgray)1_1}{*(lightgray)1_1}{*(lightgray)1_1}{*(lightgray)1_2}{*(lightgray)1_2}{*(lightgray)1_3}{*(lightgray)1_4}{*(lightgray)1_5}{*(lightgray)1_6}{*(lightgray)1_7}
{1^2_1}{1^2_1}{1^2_1}{1^2_1}{1^2_1}{1^2_1}
{1^3_1}{1^3_1}{1^3_1}{1^3_1}{1^3_1}{1^3_1}
,
{*(black)\textcolor{white}{2_1}}{2^1_1}{2^1_1}{2^1_1}{2^1_1}{2^1_1}
{2^1_2}{2^1_2}{2^1_2}{2^1_2}{2^1_2}{2^1_2}
{*(black)\textcolor{white}{2^1_1}}{*(lightgray)2_1}{*(lightgray)2_1}{*(lightgray)2_1}{*(lightgray)2_2}{*(lightgray)2_3}
{2^3_1}{2^3_1}{2^3_1}{2^3_1}{2^3_1}{2^3_1}
,
{*(black)\textcolor{white}{3_1}}{3^1_1}{3^1_1}{3^1_1}{3^1_1}{3^1_1}
{3^1_2}{3^1_2}{3^1_2}{3^1_2}{3^1_2}{3^1_2}
{3^2_1}{3^2_1}{3^2_1}{3^2_1}{3^2_1}{3^2_1}
{*(black)\textcolor{white}{3^1_1}}{*(lightgray)3_1}{*(lightgray)3_1}{*(lightgray)3_1}{*(lightgray)3_1}{*(lightgray)3_2}
,
{4^1_1}{4^1_1}{4^1_1}{4^1_1}{4^1_1}{4^1_1}
{4^1_2}{4^1_2}{4^1_2}{4^1_2}{4^1_2}{4^1_2}
{4^2_1}{4^2_1}{4^2_1}{4^2_1}{4^2_1}{4^2_1}
{4^3_1}{4^3_1}{4^3_1}{4^3_1}{4^3_1}{4^3_1}
}
}
};
\end{tikzpicture}
\caption{$D=6$, $\varrho=(5,2,1)$.
On top: Three $(D,\varrho_i)$-hypergraphs $H^{(1)}$, $H^{(2)}$, $H^{(3)}$. In the middle: The corresponding tableaux $B_e$. On the bottom: The tableaux $\check B_e$. The only differences between the middle and the bottom are highlighted in black and happen in columns that correspond to link vertices.
}
\label{EVENfig:examplecheckB}
\end{figure}

We quickly observe the following.
\begin{claim}\label{EVENcla:acutecheck}
For $i \in I$ and a block edge $e$ in $H^{(i)}$ we have that
\begin{itemize}
 \item if $\zeta^{(i)} \notin e$, then $B_e = \check B_e$,
 \item if $\zeta^{(i)} \in e$, $i \neq h$, then $B_e$ and $\check B_e$ differ only in a single entry:
 The $i$-th entry of the column $\check B_{\zeta^{(i)}}$ is $i^{h}_{1}$ instead of $i_{1}$.
\end{itemize}
\end{claim}
\begin{proof}
This follows from \eqref{EVENeq:iBzetaioneh} and using \eqref{EVENeq:ellzetaione}and \eqref{EVENeq:kzetaione}.
\end{proof}

\begin{claim}\label{EVENcla:leftpart}
In each row $j$ of $\leftpart(T)$ there are only entries $j_\ell$ for some $\ell$, or $j_k^i$ for some $k$, $i$.
\end{claim}
\begin{proof}
Recall \eqref{EVENeq:leftpartmadefromblocks}. The claim now follows from combining \eqref{EVENeq:iBv}, \eqref{EVENeq:jBv}, \eqref{EVENeq:iBzetaioneh}, \eqref{EVENeq:iBzetajoneh}, \eqref{EVENeq:hBzetah}, \eqref{EVENeq:jBzetahneq}, and \eqref{EVENeq:jBzetaheq}.
\end{proof}

\begin{claim}\label{EVENcla:symbolsleft}
If $i \notin I$, then no symbol $i_\ell$ appears in $\leftpart(T)$ for any $\ell$.
For a fixed $i \in I$, the symbol $i_\ell$ appears in $\leftpart(T)$ iff there is a vertex $v$ in $H^{(i)}$ with $\ell(v)=\ell$. Moreover, $i_\ell$ appears exactly as many times as there are vertices $v$ in $H^{(i)}$ with $\ell(v)=\ell$.
\end{claim}
\begin{proof}
Consider \eqref{EVENeq:leftpartmadefromblocks}
and observe that $\leftpart(T)$ is obtained by a permutation of the entries of the tableau
\[
\sum_{i\in I} \sum_{e \in E_{\text{Block}}^{(i)}} B_e.
\]
Now use \eqref{EVENeq:niBv} and \eqref{EVENeq:njBv}.
\end{proof}

\section{Construction of $\rightpart(T)$ for even $D$} \label{EVENsec:constructC}
The tableau $\rightpart(T)$ is constructed in any way (for example in a greedy fashion) such that the following constraints are satisfied:
\begin{flalign}\label{EVENeq:sameshape}
\text{$\rightpart(T)$ has the same shape as $S$,}
&&\end{flalign}
\begin{equation}
\label{EVENeq:directreplacement}
\begin{minipage}{15.1cm}
a box in $S$ has entry $i$ iff there is some $\ell$ for which the corresponding box in $\rightpart(T)$ has entry $i_\ell$,
\end{minipage}
\end{equation}
\begin{flalign}\label{EVENeq:symbolsappearDtimes}
\text{The symbol $i_\ell$ appears in $\rightpart(T)$ and $\leftpart(T)$ together exactly $D$ many times,}
&&\end{flalign}
\begin{flalign}\label{EVENeq:allsymbolsappearleftiffright}
\text{the symbol $i_\ell$ appears in $\rightpart(T)$ iff $i_\ell$ appears in $\leftpart(T)$}.
&&\end{flalign}
Such a tableau might not be unique, but we only care about its existence.
The existence can be shown as follows.

Let $n(i_\ell)$ denote the number of times the symbol $i_\ell$ appears in $\leftpart(T)$.
If $n(i_\ell)>0$, then Claim~\ref{EVENcla:symbolsleft} implies that there are $n(i_\ell)>0$ many vertices $v$ in $H^{(i)}$ with $\ell(v)=i$. By Def.~\ref{def:dkhypergraph}\eqref{eq:def:nameedgessizelessthanD} we know that $n(i_\ell)<D$.
We construct $\rightpart(T)$ by arbitrarily replacing $D-n(i_\ell)$ many entries $i$ in $S$ by the symbol $i_\ell$ for each $i$, $\ell$ for which $n(i_\ell)>0$.
Claim~\ref{EVENcla:divisibility} below shows that this procedure replaces exactly all entries of $S$ (recall that $i$ appears in $S$ exactly $D\varrho_i$ many times).
It is clear that this construction satisfies \eqref{EVENeq:sameshape}, \eqref{EVENeq:directreplacement} and \eqref{EVENeq:symbolsappearDtimes}.
Since $0 < n(i_\ell) < D$ iff $0 < D-n(i_\ell) < D$, we conclude \eqref{EVENeq:allsymbolsappearleftiffright}.

\begin{claim}\label{EVENcla:divisibility}
        \[
        \forall i\in I: \quad \sum_{\ell \textup{ with } n(i_\ell)>0} (D-n(i_\ell)) = D\varrho_i.
        \]
\end{claim}
\begin{proof}
Since $H^{(i)}$ satisfies Def.~\ref{def:dkhypergraph}\ref{eq:def:nameblockdifference} we have that
\[
|E^{(i)}_{\text{Name}}|-|E^{(i)}_{\text{Block}}| = \varrho_i
\]
and hence
\begin{equation} \label{EVENproof_C_*}
D|E^{(i)}_{\text{Name}}|-D|E^{(i)}_{\text{Block}}| = D\varrho_i \tag{*}.
\end{equation}
Moreover Def.~\ref{def:dkhypergraph}\ref{eq:def:blockedgessizeD} states that block edges form a set partition of $V$ and each block edge has size $D$. Together with the fact that the name edges form a set partition of $V$ (Def.~\ref{def:dkhypergraph}\ref{eq:def:nameedgessizelessthanD}) we see that
$
D|E^{(i)}_{\text{Block}}| = \sum_{e \in E_{\text{Name}}^{(i)}} \size(e).
$
Together with \eqref{EVENproof_C_*} we obtain
$
D|E^{(i)}_{\text{Name}}| - \sum_{e \in E_{\text{Name}}^{(i)}} \size(e) = D\varrho_i
$
and hence
\[
\sum_{e \in E_{\text{Name}}^{(i)}} (D-\size(e)) = D\varrho_i.
\]

Since for each vertex $v$ in a name edge $e$ the value $\ell(v)$ is the same, we write $\ell(e):=\ell(v)$.
From Claim~\ref{EVENcla:symbolsleft} we know that for all $e\in E_{\text{Name}}^{(i)}$ we have $n(i_{\ell(e)})=\size(e)$.
Therefore
        \[
        \sum_{e \in E_{\text{Name}}^{(i)}} (D-n(i_{\ell(e)})) = D\varrho_i
        \]

All numbers $\ell(e)$, $e \in E_{\text{Name}}^{(i)}$, are distinct by definition.
Hence all symbols $i_{\ell(e)}$ are distinct.
All $i_{\ell(e)}$ satisfy $n(i_{\ell(e)})>0$ by Claim~\ref{EVENcla:symbolsleft}.
Moreover, for each $\ell$ with $n(i_\ell)>0$ there exists some $e$ with $\ell(e)=\ell$ also by Claim~\ref{EVENcla:symbolsleft}.
Therefore we can rewrite the sum as
        \[
        \sum_{\ell \textup{ with } n(i_\ell)>0} (D-n(i_\ell)) = D\varrho_i,
        \]
        which concludes the proof.
\end{proof}

An example of the whole construction can be seen in Figure~\ref{EVENfig:fullexample}.

\begin{figure}
\begin{tikzpicture}
\node at (-3,0) {
\rotatebox{53}{
\ytableaushort{
11111111111111111111111111111123,
2222222222233,
333
}
}
};
\node at (0,0) {
\scalebox{0.9}{
\rotatebox{53}{
\ytableaushort{
{*(lightgray)1_1}{*(lightgray)1_1}{*(lightgray)1_1}{*(lightgray)1_1}{*(lightgray)1_1}{*(lightgray)1_2}{*(lightgray)1_2}{*(lightgray)1_3}{*(lightgray)1_4}{*(lightgray)1_5}{*(lightgray)1_6}{*(lightgray)1_7}
{1^2_1}{1^2_1}{1^2_1}{1^2_1}{1^2_1}{1^2_1}
{1^3_1}{1^3_1}{1^3_1}{1^3_1}{1^3_1}{1^3_1}
{1_1}{1_2}{1_2}{1_2}{1_2}{1_3}{1_3}{1_3}{1_3}{1_3}{1_4}{1_4}{1_4}{1_4}{1_4}{1_5}{1_5}{1_5}{1_5}{1_5}{1_6}{1_6}{1_6}{1_6}{1_6}{1_7}{1_7}{1_7}{1_7}{1_7}{2_3}{3_2}
,
{*(black)\textcolor{white}{2_1}}{2^1_1}{2^1_1}{2^1_1}{2^1_1}{2^1_1}
{2^1_2}{2^1_2}{2^1_2}{2^1_2}{2^1_2}{2^1_2}
{*(black)\textcolor{white}{2^1_1}}{*(lightgray)2_1}{*(lightgray)2_1}{*(lightgray)2_1}{*(lightgray)2_2}{*(lightgray)2_3}
{2^3_1}{2^3_1}{2^3_1}{2^3_1}{2^3_1}{2^3_1}
{2_1}{2_1}{2_2}{2_2}{2_2}{2_2}{2_2}{2_3}{2_3}{2_3}{2_3}{3_2}{3_2}
,
{*(black)\textcolor{white}{3_1}}{3^1_1}{3^1_1}{3^1_1}{3^1_1}{3^1_1}
{3^1_2}{3^1_2}{3^1_2}{3^1_2}{3^1_2}{3^1_2}
{3^2_1}{3^2_1}{3^2_1}{3^2_1}{3^2_1}{3^2_1}
{*(black)\textcolor{white}{3^1_1}}{*(lightgray)3_1}{*(lightgray)3_1}{*(lightgray)3_1}{*(lightgray)3_1}{*(lightgray)3_2}
{3_1}{3_2}{3_2}
,
{4^1_1}{4^1_1}{4^1_1}{4^1_1}{4^1_1}{4^1_1}
{4^1_2}{4^1_2}{4^1_2}{4^1_2}{4^1_2}{4^1_2}
{4^2_1}{4^2_1}{4^2_1}{4^2_1}{4^2_1}{4^2_1}
{4^3_1}{4^3_1}{4^3_1}{4^3_1}{4^3_1}{4^3_1}
}
}
}
};
\end{tikzpicture}
\caption{A full example of a tableau $S$ (on top) and the corresponding tableau $T$ (below). Here $D = 6$, $m=4$.}
\label{EVENfig:fullexample}
\end{figure}

We draw some quick corollaries.
\begin{claim}\label{EVENcla:symbols}
If $i \notin I$, then the symbol $i_\ell$ does not appear in $T$ for any $\ell$.
For a fixed $i \in I$, the symbol $i_\ell$ appears in $T$ iff $i_\ell$ appears in $\leftpart(T)$ iff $i_\ell$ appears in $\rightpart(T)$ iff there is a vertex $v$ in $H^{(i)}$ with $\ell(v)=\ell$.
\end{claim}
\begin{proof}
We combine Claim~\ref{EVENcla:symbolsleft} and \eqref{EVENeq:allsymbolsappearleftiffright}.
\end{proof}

\begin{claim}\label{EVENcla:jki}
If a symbol $j_k^i$ appears in $T$, then it appears exactly $D$ many times in $T$.
\end{claim}
\begin{proof}
By \eqref{EVENeq:directreplacement} the symbols  $j_k^i$ only appear in $\leftpart(T)$.
Consider \eqref{EVENeq:leftpartmadefromblocks}
and observe that $\leftpart(T)$ is obtained by a permutation of the entries of the tableau
\[
\sum_{i\in I} \sum_{e \in E_{\text{Block}}^{(i)}} B_e.
\]
Now use Def.~\ref{def:dkhypergraph}\eqref{eq:def:blockedgessizeD}, \eqref{EVENeq:niBv}, and \eqref{EVENeq:njBv}.
\end{proof}

\section{Proof of the Tableau Lifting Theorem~\ref{thm:prolongation} for even $D$}\label{EVENsec:psipropertiesDodd}
In this section we prove the Tableau Lifting Theorem~\ref{thm:prolongation} for even $D$.

First we observe that the shape of $T$ is indeed the required shape: This follows from Proposition~\ref{EVENprop:hypergraph}, \eqref{EVENeq:leftpartmadefromblocks}, and the fact that $B_e$ and $\check B_e$ have the same rectangular shape $m \times D$.

We remark that every symbol in $T$ appears exactly $D$ many times: For the symbols $j_k^i$ this follows from Claim~\ref{EVENcla:jki}. For the symbols $i_\ell$ this follows from \eqref{EVENeq:symbolsappearDtimes}.

It remains to prove the parts \eqref{enum:rightpartinSmS}, \eqref{enum:duplex}, and \eqref{enum:existspreimage} of Theorem~\ref{thm:prolongation}.
We start with part~\eqref{enum:existspreimage},
then build up insights that then eventually lead to the proof of parts~\eqref{enum:rightpartinSmS} and~\eqref{enum:duplex}.

\subsection*{Proof of part~\eqref{enum:existspreimage} of Theorem~\ref{thm:prolongation}}
Part~\eqref{enum:existspreimage} of Theorem~\ref{thm:prolongation} is proved as follows.
We choose $\varphi(i_\ell) := i$ and $\varphi(j_k^i) := j$.
We observe that $\rightpart(\varphi(T))=S$, see~\eqref{EVENeq:directreplacement}.
Since $S$ is regular, $\rightpart(\varphi(T))$ is regular.
It remains to show that $\leftpart(\varphi(T))$ is also regular.
From Claim~\ref{EVENcla:leftpart} we see that every column of $\leftpart(\varphi(T))$ contains all entries $1,\ldots,m$, sorted from top to bottom. Thus $\leftpart(\varphi(T))$ is regular.
Since $\leftpart(\varphi(T))$ and $\rightpart(\varphi(T))$ are both regular, we conclude that $\varphi(T)$ is regular, which finishes the proof of part \eqref{enum:existspreimage} of Theorem~\ref{thm:prolongation}.

\subsection*{Parts~\eqref{enum:rightpartinSmS} and~\eqref{enum:duplex} of Theorem~\ref{thm:prolongation}: Preliminaries}

In order to prove parts \eqref{enum:rightpartinSmS} and \eqref{enum:duplex} of Theorem~\ref{thm:prolongation}, we start with some preliminary observations.

\begin{claim}\label{EVENcla:varphii}
If $\varphi(T)$ is regular, then for each $i \in I$ we have:
For every $i_\ell$ that appears in $T$,
$\varphi(i_\ell)$ only depends on $i$ and does not depend on $\ell$.
\end{claim}
\begin{proof}
By definition, for every name edge $e$ in $H^{(i)}$ the values $\ell(v)$ coincide for all $v \in e$. This trivially implies that
\begin{equation}\label{EVENeq:nameedgevarphicoincideNEW}
\text{for every name edge $e$ in $H^{(i)}$:
the values $\varphi(i_{\ell(v)})$ coincide for all $v \in e$.}
\end{equation}

We claim that
\begin{equation}\label{EVENeq:blockedgevarphicoincideNEW}
\text{for every block edge $e$ in $H^{(i)}$:
the values $\varphi(i_{\ell(v)})$ coincide for all $v \in e$.}
\end{equation}
Proof of \eqref{EVENeq:blockedgevarphicoincideNEW}:
Let $k:=k(e)$.
According to Def.~\ref{def:dkhypergraph}\eqref{eq:def:nameblockdifference} there exists a vertex $\xi^{(i)}\neq\zeta^{(i)}$ that has the same name edge and block edge as $\zeta^{(i)}$, i.e., $\ell(\zeta^{(i)}) = \ell(\xi^{(i)})$ and $k=k(\zeta^{(i)}) = k(\xi^{(i)})$.
For each $v \in e$, $v \neq \zeta^{(i)}$ we have that
the $j$-th entry ($j \neq i$) of $\check B^{(i)}_v$ is $j_k^i$, see \eqref{EVENeq:jBv}.
Moreover, the symbol that appears as the $i$-th entry of $\check B^{(i)}_v$
is $i_{\ell(v)}$, see \eqref{EVENeq:iBv}.
By construction of $T$, we have that $\check B^{(i)}_v$ is a column in $T$.
Since by assumption $\varphi(T)$ is regular,
it follows that $\varphi(\check B^{(i)}_v)$ is regular.
Hence if $v \neq \zeta^{(i)}$,
the $\varphi(j_k^i)$ are pairwise distinct. Thus $\varphi(i_{\ell(v)})$ equals the one element in $\{1,\ldots,m\} \setminus \{\varphi(j_k^i) \mid 1 \leq j \leq m, \ j\neq i\}$. This is independent of~$\ell$.
Hence the values $\varphi(i_{\ell(v)})$ coincide for all $v \in e$, $v \neq \zeta^{(i)}$. This proves \eqref{EVENeq:blockedgevarphicoincideNEW} for all $v\in e$, $v \neq \zeta^{(i)}$. Now, if $\zeta^{(i)} \in e$, then $\xi^{(i)} \in e$, for which we have $\ell(\zeta^{(i)}) = \ell(\xi^{(i)})$, and thus clearly $\varphi(i_{\ell(\zeta^{(i)})})=\varphi(i_{\ell(\xi^{(i)})})$.
This proves the claim \eqref{EVENeq:blockedgevarphicoincideNEW}.

Since $H^{(i)}$ is connected (Def.~\ref{def:dkhypergraph}\ref{eq:def:connected}), we conclude with \eqref{EVENeq:nameedgevarphicoincideNEW} and  \eqref{EVENeq:blockedgevarphicoincideNEW}:
The values $\varphi(i_{\ell(v)})$ coincide for all $v$ in $H^{(i)}$.
Since the symbol $i_\ell$ appears in $T$ iff there is some vertex $v$ in $H^{(i)}$ with $\ell(v)=\ell$
(see Claim~\ref{EVENcla:symbols}),
Claim~\ref{EVENcla:varphii} follows.
\end{proof}
For $i\in I$ we define
\begin{equation}\label{EVENeq:defvarphii}
\varphi^{\circ}(i):=\varphi(i_1).
\end{equation}
This definition is natural, because we saw in Claim~\ref{EVENcla:varphii} that if $\varphi(T)$ is regular, then
\[
\varphi^{\circ}(i)=\varphi(i_1)=\varphi(i_2)=\ldots
\]

\begin{claim}\label{EVENcla:differentphiNEW}
Let $\varphi(T)$ be regular. Let $i, j \in I$, $i\neq j$.
Then $\varphi^{\circ}(i)\neq \varphi^{\circ}(j)$.
\end{claim}
\begin{proof}
The column $\check B_{\zeta^{(h)}}$ contains the symbol $i_1$ in row $i$ and the symbol $j_1$ in row $j$, see \eqref{EVENeq:hBzetah} and \eqref{EVENeq:jBzetahneq}.
The fact that $\varphi(T)$ is regular implies that $\varphi(i_1)\neq\varphi(j_1)$.
By \eqref{EVENeq:defvarphii} this concludes the proof.
\end{proof}

\begin{claim}\label{EVENcla:droph}
Let $\varphi(T)$ be regular. Let $i \in I$, $i \neq h$.
Then $\varphi(i_1^h) = \varphi(i_1) = \varphi^{\circ}(i)$.
\end{claim}
\begin{proof}
The last equality is \eqref{EVENeq:defvarphii}.
We now prove the first equality.
Let $e$ be the block edge in $H^{(i)}$ that contains the link vertex $\zeta^{(i)}$.
Then $\check B_e^{(i)}$ is an $m \times D$ subtableau of $\leftpart(T)$, which differs from $B_e^{(i)}$ only in a single entry in the length $m$ column corresponding to $\zeta^{(i)}$: The $i$-th entry of the column $\check B_{\zeta^{(i)}}$ is $i^{h}_{1}$ instead of $i_{1}$, see Claim~\ref{EVENcla:acutecheck}.
Hence $\varphi(B_{\zeta^{(i)}})$ and $\varphi(\check B_{\zeta^{(i)}})$ are columns that coincide in all but at most this single box.
Since $\varphi(T)$ is regular and the $\varphi(T)$ only contains entries from $\{1,\ldots,m\}$ and the columns $\varphi(B_{\zeta^{(i)}})$ and $\varphi(\check B_{\zeta^{(i)}})$ are of length $m$, we conclude
that $\varphi(i^{h}_{1}) = \varphi(i_{1})$.
\end{proof}

\subsection*{Proof of part~\eqref{enum:rightpartinSmS} of Theorem~\ref{thm:prolongation}}

We now prove part \eqref{enum:rightpartinSmS} of Theorem~\ref{thm:prolongation}.
The tableau $\rightpart(T)$ only contains entries $i_\ell$ and no entries $j_k^i$, see \eqref{EVENeq:directreplacement}.
As also seen in \eqref{EVENeq:directreplacement},
if $\rightpart(T)$ contains an entry $i_\ell$, then the corresponding entry of $S$ is $i$.
Therefore $\varphi^{\circ}(S) = \varphi(\rightpart(T))$,
where we lifted the map $\varphi^{\circ} : I \to \{1,\ldots,m\}$ to a map with the same name that is defined on tableaux with entries from $I$.
Claim~\ref{EVENcla:differentphiNEW} proves property \eqref{enum:rightpartinSmS} of Theorem~\ref{thm:prolongation}.

\subsection*{Proof of part~\eqref{enum:duplex} of Theorem~\ref{thm:prolongation}}

The rest of this section is devoted to proving part~\ref{enum:duplex} of Theorem~\ref{thm:prolongation}.
A rectangular tableau whose columns all coincide is called \emph{uniform}.
In the following proof we will crucially use that a uniform tableau with an even number of columns is duplex.
Indeed, we prove part~\ref{enum:duplex} of Theorem~\ref{thm:prolongation}
by showing that if $\varphi(T)$ is regular, then for every block edge $e$:
\begin{itemize}
 \item[(I)] $\varphi(\check B_{e})$ is uniform if $e$ does not contain any link vertex $\zeta^{(i)}$,
 \item[(II)] $\varphi(\check B_{e})$ is uniform if $\zeta^{(i)} \in e$ for $i \neq h$, and
 \item[(III)] $\varphi(\check B_{e})$ is uniform if $\zeta^{(h)} \in e$.
\end{itemize}
It is clear that these three properties cover all cases and hence $\varphi(T)$ is uniform by construction \eqref{EVENeq:leftpartmadefromblocks}.
This implies part~\ref{enum:duplex} of Theorem~\ref{thm:prolongation}.

We start with proving (I).
\begin{claim}\label{EVENcla:ithentry}
Let $\varphi(T)$ be regular.
Given a block edge $e$ in $H^{(i)}$.
For all $v \in e$, $v \neq \zeta^{(i)}$, we have that
the $i$-th entry of $\varphi(\check B_v)$ is $\varphi^{\circ}(i)$.
\end{claim}
\begin{proof}
Combine \eqref{EVENeq:iBv} and Claim~\ref{EVENcla:varphii}.
\end{proof}
\begin{claim}\label{EVENcla:jBvneqzeta}
Let $\varphi(T)$ be regular.
Given a block edge $e$ in $H^{(i)}$.
For all $j \neq i$ we have that the set
\[
\{
\text{$j$-th entry of $\varphi(\check B_v)$} \mid v \in e, v \neq \zeta^{(i)}
\}
\]
consists of the single element $\varphi(j_{k(e)}^{i})$.
\end{claim}
\begin{proof}
This follows from \eqref{EVENeq:jBv}.
\end{proof}
Combining Claim~\ref{EVENcla:ithentry} and Claim~\ref{EVENcla:jBvneqzeta} we see that
(I) is true.

We now prove (II). Let $i \neq h$ and let $e$ be the block edge in $H^{(i)}$ that contains $\zeta^{(i)}$. Note that $k(e)=1$.
\begin{claim}\label{EVENcla:almostexceptionalcolumn}
Let $\varphi(T)$ be regular.
Then  $\varphi(\check B_{\zeta^{(i)}})$ coincides with $\varphi(\check B_v)$, $v \in e$, $i \neq h$.
\end{claim}
\begin{proof}
We compare the columns entrywise.
Note that $k(v)=k(\zeta^{(i)})=1$.
We make a case distinction.

Case 1: Let $j \neq i$.
The $j$-th entry of $\check B_v$ is $j_{1}^i$, see \eqref{EVENeq:jBv}.
The $j$-th entry of $\check B_{\zeta^{(i)}}$ is $j_{1}^i$, see \eqref{EVENeq:jBzetaheq}.
Hence the $j$-th entry of $\varphi(\check B_v)$ equals the $j$-th entry of $\varphi(\check B_{\zeta^{(i)}})$.

Case 2:
The $i$-th entry of $\check B_v$ is $i_1$, see \eqref{EVENeq:iBv}.
The $i$-th entry of $\check B_{\zeta^{(i)}}$ is $i_1^h$, see \eqref{EVENeq:iBzetaioneh}. Hence Claim~\ref{EVENcla:droph} implies that
the $i$-th entry of $\varphi(\check B_v)$ equals the $i$-th entry of $\varphi(\check B_{\zeta^{(i)}})$.
\end{proof}
It follows from Claim~\ref{EVENcla:almostexceptionalcolumn} that all columns in $\varphi(\check B_e)$ coincide, i.e., $\varphi(\check B_e)$ is uniform.
Thus (II) is proved.

It remains to show (III), i.e., that $\varphi(\check B_{e})$ is uniform if $\zeta^{(h)} \in e$.

\begin{claim}\label{EVENcla:hthentry}
Let $\varphi(T)$ be regular and $\zeta^{(h)}$ the link vertex in the block edge $e$.
For all $v \in e$, $v \neq \zeta^{(h)}$, we have that
the $h$-th entry of $\varphi(\check B_v)$ is $\varphi^{\circ}(h)$.
\end{claim}
\begin{proof}
This is a direct implication of Claim~\ref{EVENcla:ithentry}.
\end{proof}

\begin{claim}\label{EVENcla:exceptionalcolumn}
Let $\varphi(T)$ be regular and $\zeta^{(h)} \in e$.
Then $\varphi(\check B_{\zeta^{(h)}})$ coincides with $\varphi(\check B_{v})$,
$v \in e$.
\end{claim}
\begin{proof}
We compare the columns entrywise, considering three cases.

Case 1: We compare the $h$-th entry:
According to Claim~\ref{EVENcla:hthentry}, the $h$-th entry of $\varphi(\check B_{v})$ is $\varphi^{\circ}(h)$. According to \eqref{EVENeq:hBzetah}
the $h$-th entry of $\check B_{\zeta^{(h)}}$ is $h_1$, so
the $h$-th entry of $\varphi(\check B_{\zeta^{(h)}})$ is $\varphi(h_1)=\varphi^{\circ}(h)$, see \eqref{EVENeq:defvarphii}.

Case 2: We compare the $j$-th entry, $j \neq h$, in the case $j \notin I$:
According to \eqref{EVENeq:jBv}, the $j$-th entry of $\check B_{v}$ is $j_1^{h}$.
The $j$-th entry of $\check B_{\zeta^{(h)}}$ is also $j_1^{h}$, see \eqref{EVENeq:jBzetaheq}. Therefore
the $j$-th entry of $\varphi(\check B_{v})$ equals the $j$-th entry of $\varphi(\check B_{\zeta^{(h)}})$.

Case 3: We compare the $j$-th entry, $j \neq h$, in the case $j\in I$:
According to \eqref{EVENeq:jBv}, the $j$-th entry of $\check B_{v}$ is $j_1^{h}$.
The $j$-th entry of $\check B_{\zeta^{(h)}}$ is $j_1$, see \eqref{EVENeq:jBzetahneq}.
Claim~\ref{cla:droph} shows that
the $j$-th entry of $\varphi(\check B_{v})$ equals the $j$-th entry of $\varphi(\check B_{\zeta^{(h)}})$.
\end{proof}
It follows from Claim~\ref{EVENcla:exceptionalcolumn} that all columns in $\varphi(\check B_e)$ coincide, i.e., $\varphi(\check B_e)$ is uniform.
Thus (III) is proved.
This finishes the proof of part~\ref{enum:duplex} of Theorem~\ref{thm:prolongation}.

Theorem~\ref{thm:prolongation} is now completely proved for even $D$.

\section{The hypergraphs $H^{(i)}$ for odd $D$}\label{sec:hypergraphsoddD}
Let $I := \{i \mid \varrho_i \neq 0\}$.
In order to construct $\leftpart(T)$ we will first
construct a so-called $(D,\varrho_i)$-paired-hypergraph for each $i \in I$.
The number of columns in $\leftpart(T)$ will be precisely
the number of vertices in all these hypergraphs together. So we want the hypergraphs to be as small as possible.

We will need the basic terms from section~\ref{EVENsec:hypergraphsevenD}.
Moreover, we will need the definition of a $(D,K)$-hypergraph (Def.~\ref{def:dkhypergraph}).

\begin{definition} \label{def:dkpairedhypergraph}
        Let $D,K$ be integers. A \textit{$(D,K)$-paired-hypergraph} is defined to be a $(D,K)$-hypergraph $H=(V,E)$ that satisfies the following additional property:
        \begin{equation}
        \begin{minipage}{\textwidth-2cm}
Each block edge $e \in E_{\text{Block}}$ is paired with another block edge $\overline{e} \in E_{\text{Block}}$ such that they are connected by a name edge, i.e., there are vertices $v \in e$ and $\overline{v} \in \overline{e}$ called \emph{bridge vertices} and a name edge $e_{\text{Name}} \in E_{\text{Name}}$ such that $v,\overline{v} \in e_{\text{Name}}$. In other words, the set of block edges can be written as a disjoint union of sets of cardinality two such that the elements of each of the sets are connected by a name edge.
        \end{minipage}
         \label{eq:hypergraphoddpairs}
        \end{equation}
\end{definition}

Several examples for $(5,K)$-paired-hypergraphs are given in Figure~\ref{fig:5Kpairedhypergraphs}.
For a block edge $e$ we write $\overline e$ to denote the other block edge in its pair, and $\overline{\overline e} = e$.
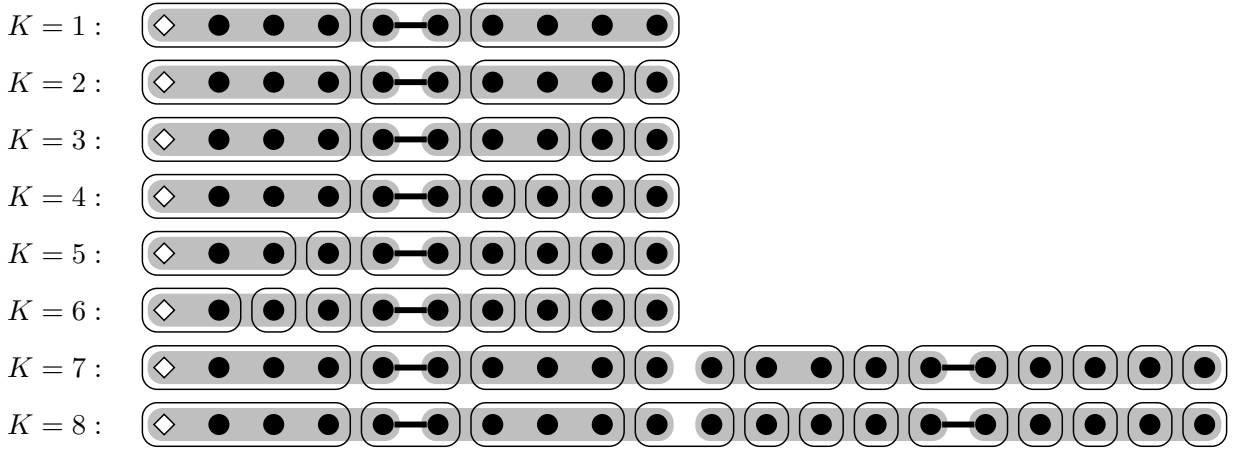
\begin{figure}[htb]
\scalebox{1.4}{
\begin{tikzpicture}
[
scale=0.52,
n/.style={circle,fill=black,inner sep=0pt, minimum size = 0.2cm},
nlink/.style={diamond,fill=white,draw=black,inner sep=0pt, minimum size = 0.2cm},
be/.style={fill=lightgray, rounded corners, inner sep=0.05cm, dashed},        
se/.style={draw, rounded corners, inner sep=0.1cm},        
]
\node at (-2,0) {\scalebox{0.71}{$K=1:$}};
\node [nlink] (1) at (0,0) {};
\node [n] (2) at (1,0) {};
\node [n] (3) at (2,0) {};
\node [n] (4) at (3,0) {};
\node [n] (5) at (4,0) {};
\node [n] (6) at (5,0) {};
\node [n] (7) at (6,0) {};
\node [n] (8) at (7,0) {};
\node [n] (9) at (8,0) {};
\node [n] (10) at (9,0) {};
\node [be, fit= (1) (2) (3) (4) (5)] {};
\node [be, fit= (6) (7) (8) (9) (10)] {};
\node [se, fit= (1) (2) (3) (4)] {};
\node [se, fit= (7) (8) (9) (10)] {};
\node [se, fit= (5) (6)] {};
\draw[color=black, ultra thick] (5) -- (6);
\node [nlink] (1) at (0,0) {};
\node [n] (2) at (1,0) {};
\node [n] (3) at (2,0) {};
\node [n] (4) at (3,0) {};
\node [n] (5) at (4,0) {};
\node [n] (6) at (5,0) {};
\node [n] (7) at (6,0) {};
\node [n] (8) at (7,0) {};
\node [n] (9) at (8,0) {};
\node [n] (10) at (9,0) {};
\end{tikzpicture}
}

\smallskip

\scalebox{1.4}{
\begin{tikzpicture}
[
scale=0.52,
n/.style={circle,fill=black,inner sep=0pt, minimum size = 0.2cm},
nlink/.style={diamond,fill=white,draw=black,inner sep=0pt, minimum size = 0.2cm},
be/.style={fill=lightgray, rounded corners, inner sep=0.05cm, dashed},        
se/.style={draw, rounded corners, inner sep=0.1cm},        
]
\node at (-2,0) {\scalebox{0.71}{$K=2:$}};
\node [nlink] (1) at (0,0) {};
\node [n] (2) at (1,0) {};
\node [n] (3) at (2,0) {};
\node [n] (4) at (3,0) {};
\node [n] (5) at (4,0) {};
\node [n] (6) at (5,0) {};
\node [n] (7) at (6,0) {};
\node [n] (8) at (7,0) {};
\node [n] (9) at (8,0) {};
\node [n] (10) at (9,0) {};
\node [be, fit= (1) (2) (3) (4) (5)] {};
\node [be, fit= (6) (7) (8) (9) (10)] {};
\node [se, fit= (1) (2) (3) (4)] {};
\node [se, fit= (7) (8) (9)] {};
\node [se, fit= (10)] {};
\node [se, fit= (5) (6)] {};
\draw[color=black, ultra thick] (5) -- (6);
\node [nlink] (1) at (0,0) {};
\node [n] (2) at (1,0) {};
\node [n] (3) at (2,0) {};
\node [n] (4) at (3,0) {};
\node [n] (5) at (4,0) {};
\node [n] (6) at (5,0) {};
\node [n] (7) at (6,0) {};
\node [n] (8) at (7,0) {};
\node [n] (9) at (8,0) {};
\node [n] (10) at (9,0) {};
\end{tikzpicture}
}

\smallskip

\scalebox{1.4}{
\begin{tikzpicture}
[
scale=0.52,
n/.style={circle,fill=black,inner sep=0pt, minimum size = 0.2cm},
nlink/.style={diamond,fill=white,draw=black,inner sep=0pt, minimum size = 0.2cm},
be/.style={fill=lightgray, rounded corners, inner sep=0.05cm, dashed},        
se/.style={draw, rounded corners, inner sep=0.1cm},        
]
\node at (-2,0) {\scalebox{0.71}{$K=3:$}};
\node [nlink] (1) at (0,0) {};
\node [n] (2) at (1,0) {};
\node [n] (3) at (2,0) {};
\node [n] (4) at (3,0) {};
\node [n] (5) at (4,0) {};
\node [n] (6) at (5,0) {};
\node [n] (7) at (6,0) {};
\node [n] (8) at (7,0) {};
\node [n] (9) at (8,0) {};
\node [n] (10) at (9,0) {};
\node [be, fit= (1) (2) (3) (4) (5)] {};
\node [be, fit= (6) (7) (8) (9) (10)] {};
\node [se, fit= (1) (2) (3) (4)] {};
\node [se, fit= (7) (8)] {};
\node [se, fit= (9)] {};
\node [se, fit= (10)] {};
\node [se, fit= (5) (6)] {};
\draw[color=black, ultra thick] (5) -- (6);
\node [nlink] (1) at (0,0) {};
\node [n] (2) at (1,0) {};
\node [n] (3) at (2,0) {};
\node [n] (4) at (3,0) {};
\node [n] (5) at (4,0) {};
\node [n] (6) at (5,0) {};
\node [n] (7) at (6,0) {};
\node [n] (8) at (7,0) {};
\node [n] (9) at (8,0) {};
\node [n] (10) at (9,0) {};
\end{tikzpicture}
}

\smallskip

\scalebox{1.4}{
\begin{tikzpicture}
[
scale=0.52,
n/.style={circle,fill=black,inner sep=0pt, minimum size = 0.2cm},
nlink/.style={diamond,fill=white,draw=black,inner sep=0pt, minimum size = 0.2cm},
be/.style={fill=lightgray, rounded corners, inner sep=0.05cm, dashed},        
se/.style={draw, rounded corners, inner sep=0.1cm},        
]
\node at (-2,0) {\scalebox{0.71}{$K=4:$}};
\node [nlink] (1) at (0,0) {};
\node [n] (2) at (1,0) {};
\node [n] (3) at (2,0) {};
\node [n] (4) at (3,0) {};
\node [n] (5) at (4,0) {};
\node [n] (6) at (5,0) {};
\node [n] (7) at (6,0) {};
\node [n] (8) at (7,0) {};
\node [n] (9) at (8,0) {};
\node [n] (10) at (9,0) {};
\node [be, fit= (1) (2) (3) (4) (5)] {};
\node [be, fit= (6) (7) (8) (9) (10)] {};
\node [se, fit= (1) (2) (3) (4)] {};
\node [se, fit= (7)] {};
\node [se, fit= (8)] {};
\node [se, fit= (9)] {};
\node [se, fit= (10)] {};
\node [se, fit= (5) (6)] {};
\draw[color=black, ultra thick] (5) -- (6);
\node [nlink] (1) at (0,0) {};
\node [n] (2) at (1,0) {};
\node [n] (3) at (2,0) {};
\node [n] (4) at (3,0) {};
\node [n] (5) at (4,0) {};
\node [n] (6) at (5,0) {};
\node [n] (7) at (6,0) {};
\node [n] (8) at (7,0) {};
\node [n] (9) at (8,0) {};
\node [n] (10) at (9,0) {};
\end{tikzpicture}
}

\smallskip

\scalebox{1.4}{
\begin{tikzpicture}
[
scale=0.52,
n/.style={circle,fill=black,inner sep=0pt, minimum size = 0.2cm},
nlink/.style={diamond,fill=white,draw=black,inner sep=0pt, minimum size = 0.2cm},
be/.style={fill=lightgray, rounded corners, inner sep=0.05cm, dashed},        
se/.style={draw, rounded corners, inner sep=0.1cm},        
]
\node at (-2,0) {\scalebox{0.71}{$K=5:$}};
\node [nlink] (1) at (0,0) {};
\node [n] (2) at (1,0) {};
\node [n] (3) at (2,0) {};
\node [n] (4) at (3,0) {};
\node [n] (5) at (4,0) {};
\node [n] (6) at (5,0) {};
\node [n] (7) at (6,0) {};
\node [n] (8) at (7,0) {};
\node [n] (9) at (8,0) {};
\node [n] (10) at (9,0) {};
\node [be, fit= (1) (2) (3) (4) (5)] {};
\node [be, fit= (6) (7) (8) (9) (10)] {};
\node [se, fit= (1) (2) (3)] {};
\node [se, fit= (4)] {};
\node [se, fit= (7)] {};
\node [se, fit= (8)] {};
\node [se, fit= (9)] {};
\node [se, fit= (10)] {};
\node [se, fit= (5) (6)] {};
\draw[color=black, ultra thick] (5) -- (6);
\node [nlink] (1) at (0,0) {};
\node [n] (2) at (1,0) {};
\node [n] (3) at (2,0) {};
\node [n] (4) at (3,0) {};
\node [n] (5) at (4,0) {};
\node [n] (6) at (5,0) {};
\node [n] (7) at (6,0) {};
\node [n] (8) at (7,0) {};
\node [n] (9) at (8,0) {};
\node [n] (10) at (9,0) {};
\end{tikzpicture}
}

\smallskip

\scalebox{1.4}{
\begin{tikzpicture}
[
scale=0.52,
n/.style={circle,fill=black,inner sep=0pt, minimum size = 0.2cm},
nlink/.style={diamond,fill=white,draw=black,inner sep=0pt, minimum size = 0.2cm},
be/.style={fill=lightgray, rounded corners, inner sep=0.05cm, dashed},        
se/.style={draw, rounded corners, inner sep=0.1cm},        
]
\node at (-2,0) {\scalebox{0.71}{$K=6:$}};
\node [nlink] (1) at (0,0) {};
\node [n] (2) at (1,0) {};
\node [n] (3) at (2,0) {};
\node [n] (4) at (3,0) {};
\node [n] (5) at (4,0) {};
\node [n] (6) at (5,0) {};
\node [n] (7) at (6,0) {};
\node [n] (8) at (7,0) {};
\node [n] (9) at (8,0) {};
\node [n] (10) at (9,0) {};
\node [be, fit= (1) (2) (3) (4) (5)] {};
\node [be, fit= (6) (7) (8) (9) (10)] {};
\node [se, fit= (1) (2)] {};
\node [se, fit= (3)] {};
\node [se, fit= (4)] {};
\node [se, fit= (7)] {};
\node [se, fit= (8)] {};
\node [se, fit= (9)] {};
\node [se, fit= (10)] {};
\node [se, fit= (5) (6)] {};
\draw[color=black, ultra thick] (5) -- (6);
\node [nlink] (1) at (0,0) {};
\node [n] (2) at (1,0) {};
\node [n] (3) at (2,0) {};
\node [n] (4) at (3,0) {};
\node [n] (5) at (4,0) {};
\node [n] (6) at (5,0) {};
\node [n] (7) at (6,0) {};
\node [n] (8) at (7,0) {};
\node [n] (9) at (8,0) {};
\node [n] (10) at (9,0) {};
\end{tikzpicture}
}

\smallskip

\scalebox{1.4}{
\begin{tikzpicture}
[
scale=0.52,
n/.style={circle,fill=black,inner sep=0pt, minimum size = 0.2cm},
nlink/.style={diamond,fill=white,draw=black,inner sep=0pt, minimum size = 0.2cm},
be/.style={fill=lightgray, rounded corners, inner sep=0.05cm, dashed},        
se/.style={draw, rounded corners, inner sep=0.1cm},        
]
\node at (-2,0) {\scalebox{0.71}{$K=7:$}};
\node [nlink] (1) at (0,0) {};
\node [n] (2) at (1,0) {};
\node [n] (3) at (2,0) {};
\node [n] (4) at (3,0) {};
\node [n] (5) at (4,0) {};
\node [n] (6) at (5,0) {};
\node [n] (7) at (6,0) {};
\node [n] (8) at (7,0) {};
\node [n] (9) at (8,0) {};
\node [n] (10) at (9,0) {};
\node [n] (11) at (10,0) {};
\node [n] (12) at (11,0) {};
\node [n] (13) at (12,0) {};
\node [n] (14) at (13,0) {};
\node [n] (15) at (14,0) {};
\node [n] (16) at (15,0) {};
\node [n] (17) at (16,0) {};
\node [n] (18) at (17,0) {};
\node [n] (19) at (18,0) {};
\node [n] (20) at (19,0) {};
\node [be, fit= (1) (2) (3) (4) (5)] {};
\node [be, fit= (6) (7) (8) (9) (10)] {};
\node [be, fit= (11) (12) (13) (14) (15)] {};
\node [be, fit= (16) (17) (18) (19) (20)] {};
\node [se, fit= (1) (2) (3) (4)] {};
\node [se, fit= (5) (6)] {};
\node [se, fit= (7) (8) (9)] {};
\node [se, fit= (10) (11)] {};
\node [se, fit= (12) (13)] {};
\node [se, fit= (14)] {};
\node [se, fit= (15) (16)] {};
\node [se, fit= (17)] {};
\node [se, fit= (18)] {};
\node [se, fit= (19)] {};
\node [se, fit= (20)] {};
\draw[color=black, ultra thick] (5) -- (6);
\draw[color=black, ultra thick] (15) -- (16);
\node [nlink] (1) at (0,0) {};
\node [n] (2) at (1,0) {};
\node [n] (3) at (2,0) {};
\node [n] (4) at (3,0) {};
\node [n] (5) at (4,0) {};
\node [n] (6) at (5,0) {};
\node [n] (7) at (6,0) {};
\node [n] (8) at (7,0) {};
\node [n] (9) at (8,0) {};
\node [n] (10) at (9,0) {};
\node [n] (11) at (10,0) {};
\node [n] (12) at (11,0) {};
\node [n] (13) at (12,0) {};
\node [n] (14) at (13,0) {};
\node [n] (15) at (14,0) {};
\node [n] (16) at (15,0) {};
\node [n] (17) at (16,0) {};
\node [n] (18) at (17,0) {};
\node [n] (19) at (18,0) {};
\node [n] (20) at (19,0) {};
\end{tikzpicture}
}

\smallskip

\scalebox{1.4}{
\begin{tikzpicture}
[
scale=0.52,
n/.style={circle,fill=black,inner sep=0pt, minimum size = 0.2cm},
nlink/.style={diamond,fill=white,draw=black,inner sep=0pt, minimum size = 0.2cm},
be/.style={fill=lightgray, rounded corners, inner sep=0.05cm, dashed},        
se/.style={draw, rounded corners, inner sep=0.1cm},        
]
\node at (-2,0) {\scalebox{0.71}{$K=8:$}};
\node [nlink] (1) at (0,0) {};
\node [n] (2) at (1,0) {};
\node [n] (3) at (2,0) {};
\node [n] (4) at (3,0) {};
\node [n] (5) at (4,0) {};
\node [n] (6) at (5,0) {};
\node [n] (7) at (6,0) {};
\node [n] (8) at (7,0) {};
\node [n] (9) at (8,0) {};
\node [n] (10) at (9,0) {};
\node [n] (11) at (10,0) {};
\node [n] (12) at (11,0) {};
\node [n] (13) at (12,0) {};
\node [n] (14) at (13,0) {};
\node [n] (15) at (14,0) {};
\node [n] (16) at (15,0) {};
\node [n] (17) at (16,0) {};
\node [n] (18) at (17,0) {};
\node [n] (19) at (18,0) {};
\node [n] (20) at (19,0) {};
\node [be, fit= (1) (2) (3) (4) (5)] {};
\node [be, fit= (6) (7) (8) (9) (10)] {};
\node [be, fit= (11) (12) (13) (14) (15)] {};
\node [be, fit= (16) (17) (18) (19) (20)] {};
\node [se, fit= (1) (2) (3) (4)] {};
\node [se, fit= (5) (6)] {};
\node [se, fit= (7) (8) (9)] {};
\node [se, fit= (10) (11)] {};
\node [se, fit= (12)] {};
\node [se, fit= (13)] {};
\node [se, fit= (14)] {};
\node [se, fit= (15) (16)] {};
\node [se, fit= (17)] {};
\node [se, fit= (18)] {};
\node [se, fit= (19)] {};
\node [se, fit= (20)] {};
\draw[color=black, ultra thick] (5) -- (6);
\draw[color=black, ultra thick] (15) -- (16);
\node [nlink] (1) at (0,0) {};
\node [n] (2) at (1,0) {};
\node [n] (3) at (2,0) {};
\node [n] (4) at (3,0) {};
\node [n] (5) at (4,0) {};
\node [n] (6) at (5,0) {};
\node [n] (7) at (6,0) {};
\node [n] (8) at (7,0) {};
\node [n] (9) at (8,0) {};
\node [n] (10) at (9,0) {};
\node [n] (11) at (10,0) {};
\node [n] (12) at (11,0) {};
\node [n] (13) at (12,0) {};
\node [n] (14) at (13,0) {};
\node [n] (15) at (14,0) {};
\node [n] (16) at (15,0) {};
\node [n] (17) at (16,0) {};
\node [n] (18) at (17,0) {};
\node [n] (19) at (18,0) {};
\node [n] (20) at (19,0) {};
\end{tikzpicture}
}
       \caption{$(5,K)$-paired-hypergraphs for several values of $K$.
       Each block edges is represented as a filled gray blob.
       Each name edges is enclosed by a solid black curve.
       The link vertices are drawn as white diamonds. The bridge vertex pairs $\{v,\overline{v}\}$ for property \protect\eqref{eq:hypergraphoddpairs} are illustrated via thick line segments.}
        \label{fig:5Kpairedhypergraphs}
\end{figure}

\begin{proposition} \label{prop:hypergraphodd}
        For odd $D \geq 3$, $K \neq 0$, there exists a $(D,K)$-paired-hypergraph that has exactly $2\lceil \frac{K}{2(D-2)} \rceil$ many block edges.
\end{proposition}
\begin{proof}
 Let $n := 2\lceil \frac{K}{2(D-2)} \rceil$.
We start the construction by considering $n$ many disjoint block edges with $D$ many vertices each.
We arrange the vertices in a linear fashion as in Figure~\ref{fig:5Kpairedhypergraphs}.
The leftmost vertex is the link vertex.
We now place these vertices in $K+n$ many name edges as follows.
As in Figure~\ref{fig:5Kpairedhypergraphs}, the rightmost vertex of every block edge but the last shall be placed in a size 2 name edge with the leftmost vertex of the next block edge.
The resulting hypergraph is connected.
The rightmost vertex of every odd block edge and the leftmost vertex of every even block edge are bridge vertices.
At this point we have $nD-2(n-1)=n(D-2)+2$ vertices that are not in name edges yet; and we have $K+n-(n-1) = K+1$ name edges left to put vertices in. Since
\[
(n(D-2)+2)-(K+1) =
(2\lceil \tfrac{K}{2(D-2)} \rceil(D-2)+2)-(K+1)
\geq
(K+2)-(K+1) = 1 > 0,
\]
we can position the name edges so that the link vertex is in a name edge of size at least $2$.
Moreover, we position that name edge of size at least 2 in such a way that the link vertex has a vertex that not only lies in the same name edge, but also in the same block edge.
\end{proof}

\section{Construction of $\leftpart(T)$ for odd $D$}
\label{sec:rightpartoddD}

For each $i \in I$ let $H^{(i)}$ be the $(D,\varrho_i)$-paired-hypergraph from Proposition~\ref{prop:hypergraphodd}. We write $E_{\text{Block}}^{(i)}$ to denote its set of block edges and $E_{\text{Name}}^{(i)}$ to denote its set of name edges.

In this section, for every $i \in I$ and every $e \in E_{\text{Block}}^{(i)}$ we construct an $m \times D$ block tableau $\check B_e$ such that
$\leftpart(T)$ is constructed as the concatenation
\begin{equation}\label{eq:leftpartmadefromblocks}
\leftpart(T) := \sum_{i\in I} \sum_{e \in E_{\text{Block}}^{(i)}} \check B_e.
\end{equation}
Notice that since every block edge has size $D$ (see Def.~\ref{def:dkhypergraph}\eqref{eq:def:blockedgessizeD}), this implies that the number of columns in $\leftpart(T)$ is equal to the sum of numbers of vertices in the hypergraphs $H^{(i)}$, $i \in I$.

Each $m\times D$ block tableau $\check B_e$ is constructed in three steps: First we construct an $m\times D$ block tableau $B_e$, then we modify its entries to $\acute B_e$, and the we make final adjustments to the entries to obtain $\check B_e$.

Let $\zeta^{(i)}$ denote the link vertex in $H^{(i)}$.
We attach some additional data to each $H^{(i)}$ as follows.
We put a linear order on the set of name edges $E_{\text{name}}^{(i)}$
and for each vertex $v$ in $H^{(i)}$ we define $\ell(v)$ to be the index of its corresponding name edge. Here $\ell(v)=1$ if $v$ lies in the first name edge, $\ell(v)=2$ for the next name edge, and so on.
We ensure that
\begin{equation}\label{eq:ellzetaione}
\ell(\zeta^{(i)})=1.
\end{equation}
In the same way, we put a linear order on the set of block edge pairs;
for each block edge $e$ we write $k(e)$ for the index of its corresponding block edge pair
and for each vertex $v$ in $H^{(i)}$ we define $k(v)$ to be the index of its corresponding block edge pair.
We ensure that
\begin{equation}\label{eq:kzetaione}
k(\zeta^{(i)})=1.
\end{equation}
Moreover, for every vertex $v$ in any $H^{(i)}$ we define $i(v):=i$.

In the following, for each vertex $v$ we define an $m \times 1$ rectangular tableau (i.e., a column of length $m$) called $B_v$.
Concatenating them results in $B_e$: $B_e := \sum_{v \in e} B_v$.
The order of columns does not matter.
Analogously, later we define $\acute B_e := \sum_{v \in e} \acute B_v$ and
$\check B_e := \sum_{v \in e} \check B_v$.

\subsection*{Starting with $B$}
For each block edge pair we choose one block edge to be the \emph{barred} block edge, and the other one to be the \emph{unbarred} block edge.

Let $e$ be in the $k$-th pair of block edges in $H^{(i)}$ and let $v \in e$.
The column $B_v$ is defined by the following properties.
\begin{eqnarray}
 \text{ the $i$-th entry of } B_v \text{ is } i_{\ell(v)} \label{eq:niBv}\\
 \text{ the $j$-th entry ($j \neq i$) of } B_v \text{ is } \begin{cases}
                                                            j_{k}^{i} & \text{ if } e \text{ is unbarred} \\
                                                            j_{\overline k}^{i} & \text{ if } e \text{ is barred} \\
                                                           \end{cases}
 \label{eq:njBv}
\end{eqnarray}
An example is given in Figure~\ref{fig:exampleB}.

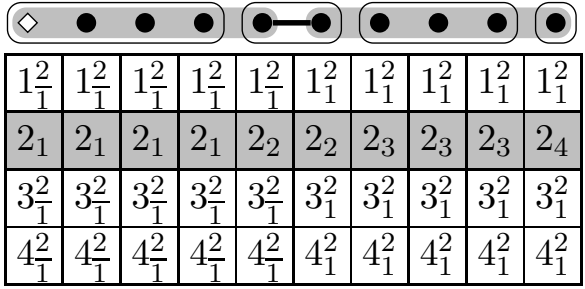
\begin{figure}
\scalebox{1.3}{
\begin{tikzpicture}
\node at (0,0.1) {
\begin{tikzpicture}
[
scale=0.6,
n/.style={circle,fill=black,inner sep=0pt, minimum size = 0.2cm},
nlink/.style={diamond,fill=white,draw=black,inner sep=0pt, minimum size = 0.2cm},
be/.style={fill=lightgray, rounded corners, inner sep=0.05cm, dashed},        
se/.style={draw, rounded corners, inner sep=0.1cm},        
]
\node [nlink] (1) at (0,0) {};
\node [n] (2) at (1,0) {};
\node [n] (3) at (2,0) {};
\node [n] (4) at (3,0) {};
\node [n] (5) at (4,0) {};
\node [n] (6) at (5,0) {};
\node [n] (7) at (6,0) {};
\node [n] (8) at (7,0) {};
\node [n] (9) at (8,0) {};
\node [n] (10) at (9,0) {};
\node [be, fit= (1) (2) (3) (4) (5)] {};
\node [be, fit= (6) (7) (8) (9) (10)] {};
\node [se, fit= (1) (2) (3) (4)] {};
\node [se, fit= (7) (8) (9)] {};
\node [se, fit= (10)] {};
\node [se, fit= (5) (6)] {};
\draw[color=black, ultra thick] (5) -- (6);
\node [nlink] (1) at (0,0) {};
\node [n] (2) at (1,0) {};
\node [n] (3) at (2,0) {};
\node [n] (4) at (3,0) {};
\node [n] (5) at (4,0) {};
\node [n] (6) at (5,0) {};
\node [n] (7) at (6,0) {};
\node [n] (8) at (7,0) {};
\node [n] (9) at (8,0) {};
\node [n] (10) at (9,0) {};
\end{tikzpicture}
}; 
\node at (-0.05,-1.4) {
\ytableausetup{boxsize=1.5em}
\ytableaushort{
{1^2_{\A}}{1^2_{\A}}{1^2_{\A}}{1^2_{\A}}{1^2_{\A}}{1^2_1}{1^2_1}{1^2_1}{1^2_1}{1^2_1}
,
{*(lightgray)2_1}{*(lightgray)2_1}{*(lightgray)2_1}{*(lightgray)2_1}{*(lightgray)2_2}{*(lightgray)2_2}{*(lightgray)2_3}{*(lightgray)2_3}{*(lightgray)2_3}{*(lightgray)2_4}
,
{3^2_{\A}}{3^2_{\A}}{3^2_{\A}}{3^2_{\A}}{3^2_{\A}}{3^2_1}{3^2_1}{3^2_1}{3^2_1}{3^2_1}
,
{4^2_{\A}}{4^2_{\A}}{4^2_{\A}}{4^2_{\A}}{4^2_{\A}}{4^2_1}{4^2_1}{4^2_1}{4^2_1}{4^2_1}
}
};
\end{tikzpicture}
}
\caption{
Here $i=2$ and $k=1$.
A $(5,2)$-paired-hypergraph $H^{(2)}$ and the corresponding concatenated tableau $B_e + B_{\overline e}$.
Vertices are drawn directly above their corresponding columns.
The left block edge is considered barred, the right block edge is considered unbarred.
To make the value of $i$ easy to see, the second row is highlighted.
}
\label{fig:exampleB}
\end{figure}

\subsection*{From $B$ to $\acute B$}
Fix $i \in I$.
To go from $B_e$ to $\acute B_e$ we switch some entries $j_k^{i}$ to $j_{\overline k}^{i}$ and vice versa.
We do this by considering the concatenation $B_e+B_{\overline{e}}$
and permuting some entries within the rows of this $m \times (2D)$ block tableau to obtain $\acute B_e+\acute B_{\overline{e}}$.
For each $1 \leq k \leq |E_{\text{Block}}^{(i)}|$,
let $\{e,\overline e\}$ denote the $k$-th pair of block edges in $H^{(i)}$
and 
choose a set of $m-1$ many distinct cardinality $D$ subsets $\barred(i,j,k)$ of the vertex set $e \cup \overline e$ such that
\begin{equation}\label{eq:bridgebar}
\begin{minipage}{16cm}
one of the two bridge vertices is contained in \emph{all} the $m-1$ many sets $\barred(i,j,k)$, $1 \leq j \leq m$, $j \neq i$, and the other bridge vertex is contained in \emph{none} of those sets.
\end{minipage}
\end{equation}

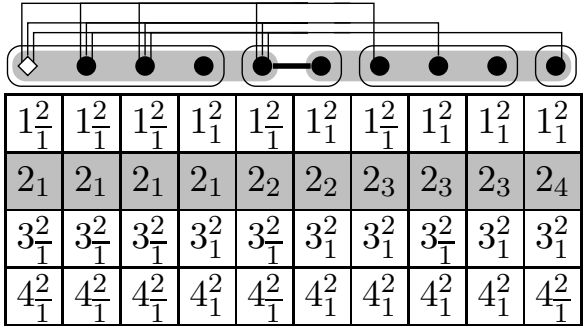
\begin{figure}
\scalebox{1.3}{
\begin{tikzpicture}
\node at (0,0.15) {
\begin{tikzpicture}
[
scale=0.6,
n/.style={circle,fill=black,inner sep=0pt, minimum size = 0.2cm},
nlink/.style={diamond,fill=white,draw=black,inner sep=0pt, minimum size = 0.2cm},
be/.style={fill=lightgray, rounded corners, inner sep=0.05cm, dashed},        
se/.style={draw, rounded corners, inner sep=0.1cm},        
]
\node [nlink] (1) at (0,0) {};
\node [n] (2) at (1,0) {};
\node [n] (3) at (2,0) {};
\node [n] (4) at (3,0) {};
\node [n] (5) at (4,0) {};
\node [n] (6) at (5,0) {};
\node [n] (7) at (6,0) {};
\node [n] (8) at (7,0) {};
\node [n] (9) at (8,0) {};
\node [n] (10) at (9,0) {};
\node [be, fit= (1) (2) (3) (4) (5)] {};
\node [be, fit= (6) (7) (8) (9) (10)] {};
\node [se, fit= (1) (2) (3) (4)] {};
\node [se, fit= (7) (8) (9)] {};
\node [se, fit= (10)] {};
\node [se, fit= (5) (6)] {};
\draw[color=black, ultra thick] (5) -- (6);
\node[above = 0.4cm of 6] (2barred) {};
\draw (2barred)+(-0.5,0) -- +(0.5,0);
\draw (1)+(-0.1,0) |- (2barred);
\draw (2)+(-0.1,0) |- (2barred);
\draw (3)+(-0.1,0) |- (2barred);
\draw (5)+(-0.1,0) |- (2barred);
\draw (7)+(-0.1,0) |- (2barred);
\node[above = 0.2cm of 6] (3barred) {};
\draw (3barred)+(-0.5,0) -- +(0.5,0);
\draw (1) |- (3barred);
\draw (2) |- (3barred);
\draw (3) |- (3barred);
\draw (5) |- (3barred);
\draw (8) |- (3barred);
\node[above = 0.1cm of 6] (4barred) {};
\draw (4barred)+(-0.5,0) -- +(0.5,0);
\draw (1)+(0.1,0) |- (4barred);
\draw (2)+(0.1,0) |- (4barred);
\draw (3)+(0.1,0) |- (4barred);
\draw (5)+(0.1,0) |- (4barred);
\draw (10)+(0.1,0) |- (4barred);
\node [nlink] (1) at (0,0) {};
\node [n] (2) at (1,0) {};
\node [n] (3) at (2,0) {};
\node [n] (4) at (3,0) {};
\node [n] (5) at (4,0) {};
\node [n] (6) at (5,0) {};
\node [n] (7) at (6,0) {};
\node [n] (8) at (7,0) {};
\node [n] (9) at (8,0) {};
\node [n] (10) at (9,0) {};
\end{tikzpicture}
}; 
\node at (-0.05,-1.6) {
\ytableausetup{boxsize=1.5em}
\ytableaushort{
{1^2_{\A}}{1^2_{\A}}{1^2_{\A}}{1^2_1}{1^2_{\A}}{1^2_1}{1^2_{\A}}{1^2_1}{1^2_1}{1^2_1}
,
{*(lightgray)2_1}{*(lightgray)2_1}{*(lightgray)2_1}{*(lightgray)2_1}{*(lightgray)2_2}{*(lightgray)2_2}{*(lightgray)2_3}{*(lightgray)2_3}{*(lightgray)2_3}{*(lightgray)2_4}
,
{3^2_{\A}}{3^2_{\A}}{3^2_{\A}}{3^2_1}{3^2_{\A}}{3^2_1}{3^2_1}{3^2_{\A}}{3^2_1}{3^2_1}
,
{4^2_{\A}}{4^2_{\A}}{4^2_{\A}}{4^2_1}{4^2_{\A}}{4^2_1}{4^2_1}{4^2_1}{4^2_1}{4^2_{\A}}
}
};
\end{tikzpicture}
}
\caption{
Here $i=2$ and $k=1$.
On top:
A $(5,1)$-paired-hypergraph $H^{(2)}$ where the three $\barred$ sets of cardinality 5 are indicated by rectangular lines. The top rectangular line represents $\barred(2,1,1)$, the next line for $\barred(2,2,1)$ is not present (because $i=2$), the next rectangular line represents $\barred(2,3,1)$, and the last one $\barred(2,4,1)$.
It can be seen that the left bridge vertex is contained in all sets $\barred(2,.,1)$, while the right bridge vertex is contained in none of those.
Below:
The corresponding tableau $\acute B_e + \acute B_{\overline e}$.
}
\label{fig:exampleacuteB}
\end{figure}

We define
\[
\kbar(i,j,v) := \begin{cases}
                 \overline{k(v)} & \text{ if } v \in \barred(i,j,k(v)) \\
                 k(v) & \text{otherwise}
                \end{cases}.
\]
Now we define
\begin{eqnarray}
 \text{ the $i$-th entry of } \acute B_v \text{ is } i_{\ell(v)}  \label{eq:acuteiBv}\\
 \text{ the $j$-th entry ($j \neq i$) of } \acute B_v \text{ is } j_{\kbar(i,j,v)}^{i}\label{eq:acutejBv}
\end{eqnarray}

An example is provided in Figure~\ref{fig:exampleacuteB}.

Since the sets $\barred(i,j,k)$ have cardinality $D$, it follows:
For $i \in I$ and a block edge $e\in H^{(i)}$ we have that
\begin{equation}\label{eq:shuffleacute}
 B_e+B_{\overline e} \text{ and } \acute B_e+\acute B_{\overline e} \text{ differ only by permutations of boxes within rows, while row $i$ stays the same.}
\end{equation}
For $1 \leq j \leq m$, $j \neq i$, it is now straightforward to see that
\begin{equation}\label{eq:Dmanysymbols}
\text{${B}_{e} + {B}_{\overline e}$ contains exactly the symbols $j_k^i$ and $j_{\overline k}^i$, both exactly $D$ many times.}
\end{equation}

We state another insight at this point:
In $H^{(i)}$,
for the pair of bridge vertices $v$, $w$ of the $k$-th block edge pair pair we have that
\begin{equation}\label{eq:bridge}
\text{$\acute B_v + \acute B_w$ contains all symbols from $S$},
\end{equation}
where $S := \{j_k^i \mid 1 \leq j \leq m, j \neq i\} \cup \{j_{\overline k}^i \mid 1 \leq j \leq m, j \neq i\}$.
\begin{proof}
This follows from combining \eqref{eq:bridgebar} and \eqref{eq:acutejBv}:
For all $j \neq i$ we have that $j_k^i$ is the $j$-th entry of $\acute B_v$ and $j_{\overline k}^i$ is the $j$-th entry of $\acute B_w$ (or vice versa).
\end{proof}

\subsection*{From $\acute B$ to $\check B$}

Most columns $\acute B_v$ and $\check B_v$ coincide, as we define
$\check B_v := \acute B_v$ if $v \notin \{\zeta^{(i)} \mid i \in I\}$.
This means that only the columns corresponding to link vertices are adjusted.

Let $h$ denote the smallest number in $I$.
For $i \in I$, $i \neq h$, the column $\check B_{\zeta^{(i)}}$ arises from $\acute B_{\zeta^{(i)}}$ by switching the $i$-th entry with the $i$-th entry in $\acute B_{\zeta^{(h)}}$. This means that the column $\check B_{\zeta^{(h)}}$ arises from $\acute B_{\zeta^{(i)}}$ by switching the $i$-th entry with the $i$-th entry in $\acute B_{\zeta^{(i)}}$ for all $i \in I$.

The columnwise description of $\check B_v$ thus as follows.
\begin{itemize}
 \item If $v$ is not a link vertex, then
\begin{eqnarray}
 \text{ the $i$-th entry of } \check B_v \text{ is } i_{\ell(v)} \label{eq:iBv}\\
 \text{ the $j$-th entry ($j \neq i$) of } \check B_v \text{ is } j_{\kbar(i,j,v)}^{(i)} \label{eq:jBv}
\end{eqnarray}
 \item If $i \neq h$, then
\begin{eqnarray}
 \text{ the $i$-th entry of } \check B_{\zeta^{(i)}} \text{ is } i_{\kbar(h,i,\zeta^{(i)})}^h \label{eq:iBzetaioneh}\\
 \text{ the $j$-th entry ($j \neq i$) of } \check B_{\zeta^{(i)}} \text{ is } j_{\kbar(i,j,\zeta^{(i)})}^i \label{eq:iBzetajoneh}
\end{eqnarray}
\item Moreover,
\begin{eqnarray}
 \text{ the $h$-th entry of } \check B_{\zeta^{(h)}} \text{ is } h_1 \label{eq:hBzetah}\\
 \text{ for $j \neq h$, $j \in I$, the $j$-th entry of } \check B_{\zeta^{(h)}} \text{ is } j_1 \label{eq:jBzetahneq}\\
 \text{ for $j \neq h$, $j \notin I$, the $j$-th entry of } \check B_{\zeta^{(h)}} \text{ is } j_{\kbar(h,j,\zeta^{(h)})}^h\label{eq:jBzetaheq}
\end{eqnarray}
\end{itemize}

An example is provided in Figure~\ref{fig:examplecheckB}.

\begin{figure}
\scalebox{0.9}{
\begin{tikzpicture}
\node at (0.05,0.1) {
\scalebox{0.8}{
\begin{tikzpicture}
[
scale=0.55,
n/.style={circle,fill=black,inner sep=0pt, minimum size = 0.2cm},
nlink/.style={diamond,fill=white,draw=black,inner sep=0pt, minimum size = 0.2cm},
be/.style={fill=lightgray, rounded corners, inner sep=0.05cm, dashed},        
se/.style={draw, rounded corners, inner sep=0.1cm},        
]
\node [nlink] (1) at (0,0) {};
\node [n] (2) at (1,0) {};
\node [n] (3) at (2,0) {};
\node [n] (4) at (3,0) {};
\node [n] (5) at (4,0) {};
\node [n] (6) at (5,0) {};
\node [n] (7) at (6,0) {};
\node [n] (8) at (7,0) {};
\node [n] (9) at (8,0) {};
\node [n] (10) at (9,0) {};
\node [n] (11) at (10,0) {};
\node [n] (12) at (11,0) {};
\node [n] (13) at (12,0) {};
\node [n] (14) at (13,0) {};
\node [n] (15) at (14,0) {};
\node [n] (16) at (15,0) {};
\node [n] (17) at (16,0) {};
\node [n] (18) at (17,0) {};
\node [n] (19) at (18,0) {};
\node [n] (20) at (19,0) {};
\node [nlink] (21) at (20,0) {};
\node [n] (22) at (21,0) {};
\node [n] (23) at (22,0) {};
\node [n] (24) at (23,0) {};
\node [n] (25) at (24,0) {};
\node [n] (26) at (25,0) {};
\node [n] (27) at (26,0) {};
\node [n] (28) at (27,0) {};
\node [n] (29) at (28,0) {};
\node [n] (30) at (29,0) {};
\node [nlink] (31) at (30,0) {};
\node [n] (32) at (31,0) {};
\node [n] (33) at (32,0) {};
\node [n] (34) at (33,0) {};
\node [n] (35) at (34,0) {};
\node [n] (36) at (35,0) {};
\node [n] (37) at (36,0) {};
\node [n] (38) at (37,0) {};
\node [n] (39) at (38,0) {};
\node [n] (40) at (39,0) {};
\node [be, fit= (1) (2) (3) (4) (5)] {};
\node [be, fit= (6) (7) (8) (9) (10)] {};
\node [be, fit= (11) (12) (13) (14) (15)] {};
\node [be, fit= (16) (17) (18) (19) (20)] {};
\node [be, fit= (21) (22) (23) (24) (25)] {};
\node [be, fit= (26) (27) (28) (29) (30)] {};
\node [be, fit= (31) (32) (33) (34) (35)] {};
\node [be, fit= (36) (37) (38) (39) (40)] {};
\node [se, fit= (1) (2) (3) (4)] {};
\node [se, fit= (5) (6)] {};
\draw[color=black, ultra thick] (5) -- (6);
\draw[color=black, ultra thick] (15) -- (16);
\draw[color=black, ultra thick] (25) -- (26);
\draw[color=black, ultra thick] (35) -- (36);
\node [se, fit= (7) (8) (9)] {};
\node [se, fit= (10) (11)] {};
\node [se, fit= (12) (13)] {};
\node [se, fit= (14)] {};
\node [se, fit= (15) (16)] {};
\node [se, fit= (17)] {};
\node [se, fit= (18)] {};
\node [se, fit= (19)] {};
\node [se, fit= (20)] {};
\node [se, fit= (21) (22) (23) (24)] {};
\node [se, fit= (25) (26)] {};
\node [se, fit= (27) (28) (29)] {};
\node [se, fit= (30)] {};
\node [se, fit= (31) (32) (33) (34)] {};
\node [se, fit= (35) (36)] {};
\node [se, fit= (37) (38) (39) (40)] {};
\node[above = 0.3cm of 6] (2barred) {};
\draw (2barred)+(-0.5,0) -- +(0.5,0);
\draw (1)+(-0.1,0) |- (2barred);
\draw (2)+(-0.1,0) |- (2barred);
\draw (3)+(-0.1,0) |- (2barred);
\draw (5)+(-0.1,0) |- (2barred);
\draw (7)+(-0.1,0) |- (2barred);
\node[above = 0.2cm of 6] (3barred) {};
\draw (3barred)+(-0.5,0) -- +(0.5,0);
\draw (1) |- (3barred);
\draw (2) |- (3barred);
\draw (3) |- (3barred);
\draw (5) |- (3barred);
\draw (8) |- (3barred);
\node[above = 0.1cm of 6] (4barred) {};
\draw (4barred)+(-0.5,0) -- +(0.5,0);
\draw (1)+(0.1,0) |- (4barred);
\draw (2)+(0.1,0) |- (4barred);
\draw (3)+(0.1,0) |- (4barred);
\draw (5)+(0.1,0) |- (4barred);
\draw (10)+(0.1,0) |- (4barred);
\node[above = 0.3cm of 16] (2barredB) {};
\draw (2barredB)+(-0.5,0) -- +(0.5,0);
\draw (11)+(-0.1,0) |- (2barredB);
\draw (12)+(-0.1,0) |- (2barredB);
\draw (13)+(-0.1,0) |- (2barredB);
\draw (15)+(-0.1,0) |- (2barredB);
\draw (17)+(-0.1,0) |- (2barredB);
\node[above = 0.2cm of 16] (3barredB) {};
\draw (3barredB)+(-0.5,0) -- +(0.5,0);
\draw (11) |- (3barredB);
\draw (12) |- (3barredB);
\draw (13) |- (3barredB);
\draw (15) |- (3barredB);
\draw (18) |- (3barredB);
\node[above = 0.1cm of 16] (4barredB) {};
\draw (4barredB)+(-0.5,0) -- +(0.5,0);
\draw (11)+(0.1,0) |- (4barredB);
\draw (12)+(0.1,0) |- (4barredB);
\draw (13)+(0.1,0) |- (4barredB);
\draw (15)+(0.1,0) |- (4barredB);
\draw (20)+(0.1,0) |- (4barredB);
\node[above = 0.4cm of 26] (2barredC) {};
\draw (2barredC)+(-0.5,0) -- +(0.5,0);
\draw (21)+(-0.1,0) |- (2barredC);
\draw (22)+(-0.1,0) |- (2barredC);
\draw (23)+(-0.1,0) |- (2barredC);
\draw (25)+(-0.1,0) |- (2barredC);
\draw (27)+(-0.1,0) |- (2barredC);
\node[above = 0.2cm of 26] (3barredC) {};
\draw (3barredC)+(-0.5,0) -- +(0.5,0);
\draw (21) |- (3barredC);
\draw (22) |- (3barredC);
\draw (23) |- (3barredC);
\draw (25) |- (3barredC);
\draw (28) |- (3barredC);
\node[above = 0.1cm of 26] (4barredC) {};
\draw (4barredC)+(-0.5,0) -- +(0.5,0);
\draw (21)+(0.1,0) |- (4barredC);
\draw (22)+(0.1,0) |- (4barredC);
\draw (23)+(0.1,0) |- (4barredC);
\draw (25)+(0.1,0) |- (4barredC);
\draw (30)+(0.1,0) |- (4barredC);
\node[above = 0.4cm of 36] (2barredD) {};
\draw (2barredD)+(-0.5,0) -- +(0.5,0);
\draw (31)+(-0.1,0) |- (2barredD);
\draw (32)+(-0.1,0) |- (2barredD);
\draw (33)+(-0.1,0) |- (2barredD);
\draw (35)+(-0.1,0) |- (2barredD);
\draw (37)+(-0.1,0) |- (2barredD);
\node[above = 0.3cm of 36] (3barredD) {};
\draw (3barredD)+(-0.5,0) -- +(0.5,0);
\draw (31) |- (3barredD);
\draw (32) |- (3barredD);
\draw (33) |- (3barredD);
\draw (35) |- (3barredD);
\draw (38) |- (3barredD);
\node[above = 0.1cm of 36] (4barredD) {};
\draw (4barredD)+(-0.5,0) -- +(0.5,0);
\draw (31)+(0.1,0) |- (4barredD);
\draw (32)+(0.1,0) |- (4barredD);
\draw (33)+(0.1,0) |- (4barredD);
\draw (35)+(0.1,0) |- (4barredD);
\draw (40)+(0.1,0) |- (4barredD);
%
\node [nlink] (1) at (0,0) {};
\node [n] (2) at (1,0) {};
\node [n] (3) at (2,0) {};
\node [n] (4) at (3,0) {};
\node [n] (5) at (4,0) {};
\node [n] (6) at (5,0) {};
\node [n] (7) at (6,0) {};
\node [n] (8) at (7,0) {};
\node [n] (9) at (8,0) {};
\node [n] (10) at (9,0) {};
\node [n] (11) at (10,0) {};
\node [n] (12) at (11,0) {};
\node [n] (13) at (12,0) {};
\node [n] (14) at (13,0) {};
\node [n] (15) at (14,0) {};
\node [n] (16) at (15,0) {};
\node [n] (17) at (16,0) {};
\node [n] (18) at (17,0) {};
\node [n] (19) at (18,0) {};
\node [n] (20) at (19,0) {};
\node [nlink] (21) at (20,0) {};
\node [n] (22) at (21,0) {};
\node [n] (23) at (22,0) {};
\node [n] (24) at (23,0) {};
\node [n] (25) at (24,0) {};
\node [n] (26) at (25,0) {};
\node [n] (27) at (26,0) {};
\node [n] (28) at (27,0) {};
\node [n] (29) at (28,0) {};
\node [n] (30) at (29,0) {};
\node [nlink] (31) at (30,0) {};
\node [n] (32) at (31,0) {};
\node [n] (33) at (32,0) {};
\node [n] (34) at (33,0) {};
\node [n] (35) at (34,0) {};
\node [n] (36) at (35,0) {};
\node [n] (37) at (36,0) {};
\node [n] (38) at (37,0) {};
\node [n] (39) at (38,0) {};
\node [n] (40) at (39,0) {};
\end{tikzpicture}
}
}; 
\node at (-0.05,-1.25) {
\ytableausetup{boxsize=1.5em}
\scalebox{0.75}{
\ytableaushort{
{*(lightgray)1_1}{*(lightgray)1_1}{*(lightgray)1_1}{*(lightgray)1_1}{*(lightgray)1_2}{*(lightgray)1_2}{*(lightgray)1_3}{*(lightgray)1_3}{*(lightgray)1_3}{*(lightgray)1_4}{*(lightgray)1_4}{*(lightgray)1_5}{*(lightgray)1_5}{*(lightgray)1_6}{*(lightgray)1_7}{*(lightgray)1_7}{*(lightgray)1_8}{*(lightgray)1_9}{*(lightgray)\text{$1_{\!1\!0}$}}{*(lightgray)\text{$1_{\!1\!1}$}}
{1^2_{\A}}{1^2_{\A}}{1^2_{\A}}{1^2_1}{1^2_{\A}}{1^2_1}{1^2_{\A}}{1^2_1}{1^2_1}{1^2_1}
{1^3_{\A}}{1^3_{\A}}{1^3_{\A}}{1^3_1}{1^3_{\A}}{1^3_1}{1^3_{\A}}{1^3_1}{1^3_1}{1^3_1}
,
{2^1_{\A}}{2^1_{\A}}{2^1_{\A}}{2^1_1}{2^1_{\A}}{2^1_1}{2^1_{\A}}{2^1_1}{2^1_1}{2^1_1}
{2^1_{\B}}{2^1_{\B}}{2^1_{\B}}{2^1_2}{2^1_{\B}}{2^1_2}{2^1_{\B}}{2^1_2}{2^1_2}{2^1_2}
{*(lightgray)2_1}{*(lightgray)2_1}{*(lightgray)2_1}{*(lightgray)2_1}{*(lightgray)2_2}{*(lightgray)2_2}{*(lightgray)2_3}{*(lightgray)2_3}{*(lightgray)2_3}{*(lightgray)2_4}
{2^3_{\A}}{2^3_{\A}}{2^3_{\A}}{2^3_1}{2^3_{\A}}{2^3_1}{2^3_1}{2^3_{\A}}{2^3_1}{2^3_1}
,
{3^1_{\A}}{3^1_{\A}}{3^1_{\A}}{3^1_1}{3^1_{\A}}{3^1_1}{3^1_1}{3^1_{\A}}{3^1_1}{3^1_1}
{3^1_{\B}}{3^1_{\B}}{3^1_{\B}}{3^1_2}{3^1_{\B}}{3^1_2}{3^1_2}{3^1_{\B}}{3^1_2}{3^1_2}
{3^2_{\A}}{3^2_{\A}}{3^2_{\A}}{3^2_1}{3^2_{\A}}
{3^2_1}{3^2_1}{3^2_{\A}}{3^2_1}{3^2_1}
{*(lightgray)3_1}{*(lightgray)3_1}{*(lightgray)3_1}{*(lightgray)3_1}{*(lightgray)3_2}{*(lightgray)3_2}{*(lightgray)3_3}{*(lightgray)3_3}{*(lightgray)3_3}{*(lightgray)3_3}
,
{4^1_{\A}}{4^1_{\A}}{4^1_{\A}}{4^1_1}{4^1_{\A}}{4^1_1}{4^1_1}{4^1_1}{4^1_1}{4^1_{\A}}
{4^1_{\B}}{4^1_{\B}}{4^1_{\B}}{4^1_2}{4^1_{\B}}{4^1_2}{4^1_2}{4^1_2}{4^1_2}{4^1_{\B}}
{4^2_{\A}}{4^2_{\A}}{4^2_{\A}}{4^2_1}{4^2_{\A}}{4^2_1}{4^2_1}{4^2_1}{4^2_1}{4^2_{\A}}
{4^3_{\A}}{4^3_{\A}}{4^3_{\A}}{4^3_1}{4^3_{\A}}{4^3_1}{4^3_1}{4^3_1}{4^3_1}{4^3_{\A}}
}
}
};
\node at (-0.05,-3.25) {
\ytableausetup{boxsize=1.5em}
\scalebox{0.75}{
\ytableaushort{
{*(lightgray)1_1}{*(lightgray)1_1}{*(lightgray)1_1}{*(lightgray)1_1}{*(lightgray)1_2}{*(lightgray)1_2}{*(lightgray)1_3}{*(lightgray)1_3}{*(lightgray)1_3}{*(lightgray)1_4}{*(lightgray)1_4}{*(lightgray)1_5}{*(lightgray)1_5}{*(lightgray)1_6}{*(lightgray)1_7}{*(lightgray)1_7}{*(lightgray)1_8}{*(lightgray)1_9}{*(lightgray)\text{$1_{\!1\!0}$}}{*(lightgray)\text{$1_{\!1\!1}$}}
{1^2_{\A}}{1^2_{\A}}{1^2_{\A}}{1^2_1}{1^2_{\A}}{1^2_1}{1^2_{\A}}{1^2_1}{1^2_1}{1^2_1}
{1^3_{\A}}{1^3_{\A}}{1^3_{\A}}{1^3_1}{1^3_{\A}}{1^3_1}{1^3_{\A}}{1^3_1}{1^3_1}{1^3_1}
,
{*(black)\textcolor{white}{2_1}}{2^1_{\A}}{2^1_{\A}}{2^1_1}{2^1_{\A}}{2^1_1}{2^1_{\A}}{2^1_1}{2^1_1}{2^1_1}
{2^1_{\B}}{2^1_{\B}}{2^1_{\B}}{2^1_2}{2^1_{\B}}{2^1_2}{2^1_{\B}}{2^1_2}{2^1_2}{2^1_2}
{*(black)\textcolor{white}{2^1_{\A}}}{*(lightgray)2_1}{*(lightgray)2_1}{*(lightgray)2_1}{*(lightgray)2_2}{*(lightgray)2_2}{*(lightgray)2_3}{*(lightgray)2_3}{*(lightgray)2_3}{*(lightgray)2_4}
{2^3_{\A}}{2^3_{\A}}{2^3_{\A}}{2^3_1}{2^3_{\A}}{2^3_1}{2^3_1}{2^3_{\A}}{2^3_1}{2^3_1}
,
{*(black)\textcolor{white}{3_1}}{3^1_{\A}}{3^1_{\A}}{3^1_1}{3^1_{\A}}{3^1_1}{3^1_1}{3^1_{\A}}{3^1_1}{3^1_1}
{3^1_{\B}}{3^1_{\B}}{3^1_{\B}}{3^1_2}{3^1_{\B}}{3^1_2}{3^1_2}{3^1_{\B}}{3^1_2}{3^1_2}
{3^2_{\A}}{3^2_{\A}}{3^2_{\A}}{3^2_1}{3^2_{\A}}
{3^2_1}{3^2_1}{3^2_{\A}}{3^2_1}{3^2_1}
{*(black)\textcolor{white}{3^1_{\A}}}{*(lightgray)3_1}{*(lightgray)3_1}{*(lightgray)3_1}{*(lightgray)3_2}{*(lightgray)3_2}{*(lightgray)3_3}{*(lightgray)3_3}{*(lightgray)3_3}{*(lightgray)3_3}
,
{4^1_{\A}}{4^1_{\A}}{4^1_{\A}}{4^1_1}{4^1_{\A}}{4^1_1}{4^1_1}{4^1_1}{4^1_1}{4^1_{\A}}
{4^1_{\B}}{4^1_{\B}}{4^1_{\B}}{4^1_2}{4^1_{\B}}{4^1_2}{4^1_2}{4^1_2}{4^1_2}{4^1_{\B}}
{4^2_{\A}}{4^2_{\A}}{4^2_{\A}}{4^2_1}{4^2_{\A}}{4^2_1}{4^2_1}{4^2_1}{4^2_1}{4^2_{\A}}
{4^3_{\A}}{4^3_{\A}}{4^3_{\A}}{4^3_1}{4^3_{\A}}{4^3_1}{4^3_1}{4^3_1}{4^3_1}{4^3_{\A}}
}
}
};
\end{tikzpicture}
}
\caption{
On top: Three hypergraphs $H^{(1)}$, $H^{(2)}$, $H^{(3)}$. In the middle: The corresponding tableaux $\acute B_e$. On the bottom: The tableaux $\check B_e$. The only differences between the middle and the bottom are highlighted in black and happen in columns that correspond to link vertices.
}
\label{fig:examplecheckB}
\end{figure}

We quickly observe the following.
\begin{claim}\label{cla:acutecheck}
For $i \in I$ and a block edge $e\in H^{(i)}$ we have that
\begin{itemize}
 \item if $\zeta^{(i)} \notin e$, then $\acute B_e = \check B_e$,
 \item if $\zeta^{(i)} \in e$, $i \neq h$, then $\acute B_e$ and $\check B_e$ differ only in a single entry:
 The $i$-th entry of the column $\check B_{\zeta^{(i)}}$ is $i^{h}_{1}$ or $i^{h}_{\overline 1}$ instead of $i_{1}$.
\end{itemize}
\end{claim}
\begin{proof}
This follows from \eqref{eq:iBzetaioneh} and using \eqref{eq:ellzetaione}and \eqref{eq:kzetaione}.
\end{proof}

\begin{claim}\label{cla:leftpart}
In each row $j$ of $\leftpart(T)$ there are only entries $j_\ell$ for some $\ell$, or $j_k^i$ or $j_{\overline k}^i$ for some $k$, $i$.
\end{claim}
\begin{proof}
Recall \eqref{eq:leftpartmadefromblocks}. The claim now follows from combining \eqref{eq:iBv}, \eqref{eq:jBv}, \eqref{eq:iBzetaioneh}, \eqref{eq:iBzetajoneh}, \eqref{eq:hBzetah}, \eqref{eq:jBzetahneq}, and \eqref{eq:jBzetaheq}.
\end{proof}

\begin{claim}\label{cla:symbolsleft}
If $i \notin I$, then no symbol $i_\ell$ appears in $\leftpart(T)$ for any $\ell$.
For a fixed $i \in I$, the symbol $i_\ell$ appears in $\leftpart(T)$ iff there is a vertex $v$ in $H^{(i)}$ with $\ell(v)=\ell$. Moreover, $i_\ell$ appears exactly as many times as there are vertices $v$ in $H^{(i)}$ with $\ell(v)=\ell$.
\end{claim}
\begin{proof}
Consider \eqref{eq:leftpartmadefromblocks}
and observe that $\leftpart(T)$ is obtained by a permutation of the box entries of the tableau
\[
\sum_{i\in I} \sum_{e \in E_{\text{Block}}^{(i)}} B_e.
\]
Now use \eqref{eq:niBv} and \eqref{eq:njBv}.
\end{proof}

\section{Construction of $\rightpart(T)$ for odd $D$} \label{sec:constructC}
The tableau $\rightpart(T)$ is constructed in any way (for example in a greedy fashion) such that the following constraints are satisfied:
\begin{flalign}\label{eq:sameshape}
\text{$\rightpart(T)$ has the same shape as $S$,}
&&\end{flalign}
\begin{equation}
\label{eq:directreplacement}
\begin{minipage}{15.1cm}
a box in $S$ has entry $i$ iff there is some $\ell$ for which the corresponding box in $\rightpart(T)$ has entry $i_\ell$,
\end{minipage}
\end{equation}
\begin{flalign}\label{eq:symbolsappearDtimes}
\text{The symbol $i_\ell$ appears in $\rightpart(T)$ and $\leftpart(T)$ together exactly $D$ many times,}
&&\end{flalign}
\begin{flalign}\label{eq:allsymbolsappearleftiffright}
\text{the symbol $i_\ell$ appears in $\rightpart(T)$ iff $i_\ell$ appears in $\leftpart(T)$}.
&&\end{flalign}
Such a tableau might not be unique, but we only care about its existence.
The existence can be shown as follows.

Let $n(i_\ell)$ denote the number of times the symbol $i_\ell$ appears in $\leftpart(T)$.
If $n(i_\ell)>0$, then Claim~\ref{cla:symbolsleft} implies that there are $n(i_\ell)>0$ many vertices $v$ in $H^{(i)}$ with $\ell(v)=i$. Using Def.~\ref{def:dkhypergraph}\eqref{eq:def:nameedgessizelessthanD} we see that $n(i_\ell)<D$.
We construct $\rightpart(T)$ by arbitrarily replacing $D-n(i_\ell)$ many entries $i$ in $S$ by the symbol $i_\ell$ for each $i$, $\ell$ for which $n(i_\ell)>0$.
This replaces exactly all entries of $S$, as Claim~\ref{cla:divisibility} below shows (recall that $i$ appears in $S$ exactly $D\varrho_i$ many times).
It is clear that this construction satisfies \eqref{eq:sameshape}, \eqref{eq:directreplacement} and \eqref{eq:symbolsappearDtimes}.
Since $0 < n(i_\ell) < D$ iff $0 < D-n(i_\ell) < D$, we conclude \eqref{eq:allsymbolsappearleftiffright}.

\begin{claim}\label{cla:divisibility}
        \[
        \forall i\in I: \quad \sum_{\ell \textup{ with } n(i_\ell)>0} (D-n(i_\ell)) = D\varrho_i.
        \]
\end{claim}
\begin{proof}
Since $H^{(i)}$ satisfies Def.~\ref{def:dkhypergraph}\ref{eq:def:nameblockdifference} we have that
\[
|E^{(i)}_{\text{Name}}|-|E^{(i)}_{\text{Block}}| = \varrho_i
\]
and hence
\begin{equation} \label{proof_C_*}
D|E^{(i)}_{\text{Name}}|-D|E^{(i)}_{\text{Block}}| = D\varrho_i \tag{*}.
\end{equation}
Moreover Def.~\ref{def:dkhypergraph}\ref{eq:def:blockedgessizeD} states that block edges form a set partition of $V$ and each block edge has size $D$. Together with the fact that the name edges form a set partition of $V$ (Def.~\ref{def:dkhypergraph}\ref{eq:def:nameedgessizelessthanD}) we see that
$
D|E^{(i)}_{\text{Block}}| = \sum_{e \in E_{\text{Name}}^{(i)}} \size(e).
$
Together with \eqref{proof_C_*} we obtain
$
D|E^{(i)}_{\text{Name}}| - \sum_{e \in E_{\text{Name}}^{(i)}} \size(e) = D\varrho_i
$
and hence
\[
\sum_{e \in E_{\text{Name}}^{(i)}} (D-\size(e)) = D\varrho_i.
\]

Since for each vertex $v$ in a name edge $e$ the value $\ell(v)$ is the same, we write $\ell(e):=\ell(v)$.
From Claim~\ref{cla:symbolsleft} we know that for all $e\in E_{\text{Name}}^{(i)}$ we have $n(i_{\ell(e)})=\size(e)$.
Therefore
        \[
        \sum_{e \in E_{\text{Name}}^{(i)}} (D-n(i_{\ell(e)})) = D\varrho_i
        \]

All numbers $\ell(e)$, $e \in E_{\text{Name}}^{(i)}$, are distinct by definition.
Hence all symbols $i_{\ell(e)}$ are distinct.
All $i_{\ell(e)}$ satisfy $n(i_{\ell(e)})>0$ by Claim~\ref{cla:symbolsleft}.
Moreover, for each $\ell$ with $n(i_\ell)>0$ there exists some $e$ with $\ell(e)=\ell$ also by Claim~\ref{cla:symbolsleft}.
Therefore we can rewrite the sum as
        \[
        \sum_{\ell \textup{ with } n(i_\ell)>0} (D-n(i_\ell)) = D\varrho_i,
        \]
        which concludes the proof.
\end{proof}

An example of the whole construction can be seen in Figure~\ref{fig:fullexample}.

We draw some quick corollaries.
\begin{claim}\label{cla:symbols}
If $i \notin I$, then the symbol $i_\ell$ does not appear in $T$ for any $\ell$.
For a fixed $i \in I$, the symbol $i_\ell$ appears in $T$ iff $i_\ell$ appears in $\leftpart(T)$ iff $i_\ell$ appears in $\rightpart(T)$ iff there is a vertex $v$ in $H^{(i)}$ with $\ell(v)=\ell$.
\end{claim}
\begin{proof}
We combine Claim~\ref{cla:symbolsleft} and \eqref{eq:allsymbolsappearleftiffright}.
\end{proof}

\begin{claim}\label{cla:jki}
If a symbol $j_k^i$ or $j_{\overline{k}}^i$ appears in $T$, then it appears exactly $D$ many times in $T$.
\end{claim}
\begin{proof}
By \eqref{eq:directreplacement} the symbols  $j_k^i$ and $j_{\overline{k}}^i$ only appear in $\leftpart(T)$.
Consider \eqref{eq:leftpartmadefromblocks}
and observe that $\leftpart(T)$ is obtained by a permutation of the box entries of the tableau
\[
\sum_{i\in I} \sum_{e \in E_{\text{Block}}^{(i)}} B_e.
\]
Now use Def.~\ref{def:dkpairedhypergraph}\eqref{eq:def:blockedgessizeD}, \eqref{eq:niBv}, and \eqref{eq:njBv}.
\end{proof}

\section{Proof of the Tableau Lifting Theorem~\ref{thm:prolongation} for odd $D$}\label{sec:psipropertiesDodd}
In this section we prove the Tableau Lifting Theorem~\ref{thm:prolongation} for odd $D$.

First we observe that the shape of $T$ is indeed the required shape: This follows from Proposition~\ref{prop:hypergraphodd}, \eqref{eq:leftpartmadefromblocks}, and the fact that $B_e$ and $\check B_e$ have the same rectangular shape $m \times D$.

We remark that every symbol in $T$ appears exactly $D$ many times: For the symbols $j_k^i$ and $j_{\overline{k}}^i$ this follows from Claim~\ref{cla:jki}. For the symbols $i_\ell$ this follows from \eqref{eq:symbolsappearDtimes}.

It remains to prove the parts \eqref{enum:rightpartinSmS}, \eqref{enum:duplex}, and \eqref{enum:existspreimage} of Theorem~\ref{thm:prolongation}.
We start with part~\eqref{enum:existspreimage},
then build up insights that then eventually lead to the proof of part~\eqref{enum:rightpartinSmS} and~\eqref{enum:duplex}.

\subsection*{Proof of part~\eqref{enum:existspreimage} of Theorem~\ref{thm:prolongation}}
Part~\eqref{enum:existspreimage} of Theorem~\ref{thm:prolongation} is proved as follows.
We choose $\varphi(i_\ell) := i$ and $\varphi(j_k^i) := j$ and $\varphi(j_{\overline k}^i) := j$.
We observe that $\rightpart(\varphi(T))=S$, see~\eqref{eq:directreplacement}.
Since $S$ is regular, $\rightpart(\varphi(T))$ is regular.
It remains to show that $\leftpart(\varphi(T))$ is also regular.
From Claim~\ref{cla:leftpart} we see that every column of $\leftpart(\varphi(T))$ contains all entries $1,\ldots,m$, sorted from top to bottom. Thus $\leftpart(\varphi(T))$ is regular.
Since $\leftpart(\varphi(T))$ and $\rightpart(\varphi(T))$ are both regular, we conclude that $\varphi(T)$ is regular, which finishes the proof of part \eqref{enum:existspreimage} of Theorem~\ref{thm:prolongation}.

\subsection*{Parts~\eqref{enum:rightpartinSmS} and~\eqref{enum:duplex} of Theorem~\ref{thm:prolongation}: Preliminaries}

In order to prove parts \eqref{enum:rightpartinSmS} and \eqref{enum:duplex} of Theorem~\ref{thm:prolongation}, we start with some preliminary observations.

\begin{claim}\label{cla:barcoincide}
If $\varphi(T)$ is regular, then for each $i \in I$, $j \neq i$, $k$ we have:
$\varphi(j_k^i) = \varphi(j_{\overline k}^i)$.
\end{claim}
\begin{proof}
Fix $i,k$, but do not fix $j$.
For notational convenience we define $\acute B := \acute B_{e} + \acute B_{\overline e}$.
Let
\[
S := \{j_k^i \mid j \neq i\} \cup \{j_{\overline k}^i \mid j \neq i\}
\]
denote the set of symbols in $\acute B$ in all rows $j \neq i$ (which is the same as the set of symbols in $B_e + B_{\overline e}$ in all rows $j \neq i$, see \eqref{eq:shuffleacute}).
Note that $|S| = 2(m-1) = 2m-2$.
Since ${B}_{e} + {B}_{\overline e}$ contains exactly the symbols $j_k^i$ and $j_{\overline k}^i$, both exactly $D$ many times (see \eqref{eq:Dmanysymbols}), the same is true for row $j$ of $\acute{B}$ (again, by \eqref{eq:shuffleacute}).

The tableau $\check B_{e} + \check B_{\overline e}$ differs from $\acute B$ iff $\zeta^{(i)} \in e \cup \overline e$. In this case all differences are in the column that corresponds to the link vertex $\zeta^{(i)}$ (see Claim~\ref{cla:acutecheck}).

If $\zeta^{(i)} \in e \cup \overline e$, then define $\acute B'$ as the $m \times (2D-1)$ tableau that is obtained from $\acute B$ by removing the column corresponding to $\zeta^{(i)}$. Note that $\zeta^{(i)}$ is not a bridge vertex.
If $\zeta^{(i)} \notin e \cup \overline e$, then define $\acute B'$ as the $m \times (2D-1)$ tableau that is obtained from $\acute B$ by removing a single arbitrary column that does not correspond to a bridge vertex.

A symbol $s \in S$ is called \emph{abound} if it appears $D$ many times in $\acute B'$.
If $s \in S$ appears $D-1$ many times in $\acute B'$, then $s$ is called \emph{scarce}.
Note that each $s\in S$ is either abound or scarce.

Note that $\acute B'$ is a subtableau of $T$ and hence $\varphi(\acute B')$ is regular.
Let $v,w$ denote the two bridge vertices. Since by definition they lie in the same name edge, we have $\ell(v)=\ell(w)$. Let $\ell := \ell(v)$.
Note that $\acute B_v$ and $\acute B_w$ both have the symbol $i_\ell$ in row $i$, see \eqref{eq:acuteiBv}.
Since $\varphi(\acute B_v + \acute B_w)$ is regular and since the rows $j\neq i$ of $\acute B_v + \acute B_w$ contains all symbols from $S$ (see \eqref{eq:bridge}), it follows that
\begin{equation}\label{eq:P}
\forall s \in S: \ \varphi(s) \in \underbrace{\{1,\ldots,m\}\setminus\{\varphi(i_\ell)\}}_{=: P}.
\end{equation}
Note that $|P| = m-1$.

Two symbols $s,t \in S$ are called \emph{partners} if $s$ appears exactly in those columns of $\acute B$ in which $t$ does not appear. In this case (by construction) we have
$s=j_k^i$ and $t=j_{\overline k}^i$, or $t=j_k^i$ and $s=j_{\overline k}^i$.
If two symbols $s,t$ are not partners, then there exists a column in $\acute B$ that contains both $s$ and $t$,
because both $s$ and $t$ appear in exactly $D$ many columns in the shape $m\times (2D)$ tableau $\acute B$.

Clearly, for each partnership one partner is abound and the other is scarce.

Since $\varphi(B')$ has shape $m \times (2D-1)$,
every symbol from $\{1,\ldots,m\}$ appears exactly $2D-1$ times in $\varphi(B')$.
Therefore for every symbol $p \in P$ there is at most one abound symbol $s \in S$
with $\varphi(s)=p$. Since $S$ contains $m-1$ many abound symbols,
from \eqref{eq:P} it follows that for each $p \in P$ there is \emph{exactly} one abound symbol $s \in S$ with $\varphi(s)=p$.
Having seen this, it follows that for each $p\in P$ there is at most one scarce symbol $t \in S$ with $\varphi(s)=p$. Again, since $S$ contains $m-1$ many scarce symbols and because of \eqref{eq:P},
there is \emph{exactly} one scarce symbol $t \in S$ with the property that $\varphi(t)=p$.

Given an abound symbol $s \in S$ and a symbol $t\in S$ that is not the partner of $s$, then in $\acute B'$ there is a column that contains both $s$ and $t$.
Therefore $\varphi(s)\neq \varphi(t)$ if $s$ is abound and $t$ is not the partner of $s$.
We conclude that for every $p \in P$ there is exactly one abound $s \in S$ and its scarce partner $t\in S$ that have $\varphi(s)=\varphi(t) = p$.
This implies the claim.
\end{proof}

\begin{claim}\label{cla:varphii}
If $\varphi(T)$ is regular, then for each $i \in I$ we have:
For every $i_\ell$ that appears in $T$,
$\varphi(i_\ell)$ only depends on $i$ and does not depend on $\ell$.
\end{claim}
\begin{proof}
By definition, for every name edge $e$ in $H^{(i)}$ the values $\ell(v)$ coincide for all $v \in e$. This trivially implies that
\begin{equation}\label{eq:nameedgevarphicoincideNEW}
\text{for every name edge $e$ in $H^{(i)}$:
the values $\varphi(i_{\ell(v)})$ coincide for all $v \in e$.}
\end{equation}

We claim that
\begin{equation}\label{eq:blockedgevarphicoincideNEW}
\text{for every block edge $e$ in $H^{(i)}$:
the values $\varphi(i_{\ell(v)})$ coincide for all $v \in e$.}
\end{equation}
Proof:
Let $k:=k(e)$.
According to Def.~\ref{def:dkhypergraph}\eqref{eq:def:nameblockdifference} there exists a vertex $\xi^{(i)}\neq\zeta^{(i)}$ that has the same name edge and block edge as $\zeta^{(i)}$, i.e., $\ell(\zeta^{(i)}) = \ell(\xi^{(i)})$ and $k=k(\zeta^{(i)}) = k(\xi^{(i)})$.
For each $v \in e$, $v \neq \zeta^{(i)}$ we have that from each of the $m-1$ sets
$\{j_k^i,j_{\overline k}^i\}$, $1\leq j \leq m$, $j \neq i$, there is one symbol in the column $\check B^{(i)}_v$, see \eqref{eq:jBv}.
Moreover, the symbol that appears as the $i$-th entry of $\check B^{(i)}_v$
is $i_{\ell(v)}$, see \eqref{eq:iBv}.
Since $\varphi(T)$ is regular, Claim~\ref{cla:barcoincide} implies that $\varphi(j_k^i)=\varphi(j_{\overline k}^i)$, which we will use implicitly in the upcoming argument.
For each $v \neq \zeta^{(i)}$ we have that $\check B^{(i)}_v$ is a column in $T$.
In this case, since by assumption $\varphi(T)$ is regular,
it follows that $\varphi(\check B^{(i)}_v)$ is regular and hence the $\varphi(j_k^i)$ are pairwise distinct. Thus $\varphi(i_{\ell(v)})$ equals the one element in $\{1,\ldots,m\} \setminus \{\varphi(j_k^i) \mid 1 \leq j \leq m, \ j\neq i\}$, which is independent of~$\ell$.
This implies that the values $\varphi(i_{\ell(v)})$ coincide for all $v \in e$, $v \neq \zeta^{(i)}$. This proves \eqref{eq:blockedgevarphicoincideNEW} for all $v\in e$, $v \neq \zeta^{(i)}$. Now, if $\zeta^{(i)} \in e$, then $\xi^{(i)} \in e$, for which we have $\ell(\zeta^{(i)}) = \ell(\xi^{(i)})$, and thus clearly $\varphi(i_{\ell(\zeta^{(i)})})=\varphi(i_{\ell(\xi^{(i)})})$.
This proves the claim \eqref{eq:blockedgevarphicoincideNEW}.

Since $H^{(i)}$ is connected (Def.~\ref{def:dkhypergraph}\ref{eq:def:connected}), we conclude with \eqref{eq:nameedgevarphicoincideNEW} and  \eqref{eq:blockedgevarphicoincideNEW}:
The values $\varphi(i_{\ell(v)})$ coincide for all $v$ in $H^{(i)}$.
Since the symbol $i_\ell$ appears in $T$ iff there is some vertex $v$ in $H^{(i)}$ with $\ell(v)=\ell$
(see Claim~\ref{cla:symbols}),
Claim~\ref{cla:varphii} follows.
\end{proof}
For $i\in I$ we define
\begin{equation}\label{eq:defvarphii}
\varphi^{\circ}(i):=\varphi(i_1).
\end{equation}
This definition is natural, because we saw in Claim~\ref{cla:varphii} that if $\varphi(T)$ is regular, then
\[
\varphi^{\circ}(i)=\varphi(i_1)=\varphi(i_2)=\ldots
\]

\begin{claim}\label{cla:differentphiNEW}
Let $\varphi(T)$ be regular. Let $i, j \in I$, $i\neq j$.
Then $\varphi^{\circ}(i)\neq \varphi^{\circ}(j)$.
\end{claim}
\begin{proof}
The column $\check B_{\zeta^{(h)}}$ contains the symbol $i_1$ in row $i$ and the symbol $j_1$ in row $j$, see \eqref{eq:hBzetah} and \eqref{eq:jBzetahneq}.
The fact that $\varphi(T)$ is regular implies that $\varphi(i_1)\neq\varphi(j_1)$.
By definition \eqref{eq:defvarphii}, this concludes the proof.
\end{proof}

\begin{claim}\label{cla:droph}
Let $\varphi(T)$ be regular. Let $i \in I$, $i \neq h$.
Then $\varphi(i_{\overline 1}^h) = \varphi(i_1^h) = \varphi(i_1) = \varphi^{\circ}(i)$.
\end{claim}
\begin{proof}
The first equality follows from Claim~\ref{cla:barcoincide}.
The last equality is \eqref{eq:defvarphii}.
We now prove the second equality.
Let $e$ be the block edge in $H^{(i)}$ that contains the link vertex $\zeta^{(i)}$.
Then $\check B_e^{(i)}$ is an $m \times D$ subtableau of $\leftpart(T)$, which differs from $\acute B_e^{(i)}$ only in a single entry in the length $m$ column corresponding to $\zeta^{(i)}$: The $i$-th entry of the column $\check B_{\zeta^{(i)}}$ is $i^{h}_{1}$ or $i^{h}_{\overline 1}$ instead of $i_{1}$, see Claim~\ref{cla:acutecheck}.
Hence $\varphi(\acute B_{\zeta^{(i)}})$ and $\varphi(\check B_{\zeta^{(i)}})$ are columns that coincide in all but at most this single box.
Since $\varphi(T)$ is regular and the $\varphi(T)$ only contains entries from $\{1,\ldots,m\}$ and the columns $\varphi(\acute B_{\zeta^{(i)}})$ and $\varphi(\check B_{\zeta^{(i)}})$ are of length $m$, we conclude with Claim~\ref{cla:barcoincide} (i.e., $\varphi(i^{h}_{1})=\varphi(i^{h}_{\overline 1})$)
that $\varphi(i^{h}_{1}) = \varphi(i_{1})$.
\end{proof}

\subsection*{Proof of part~\eqref{enum:rightpartinSmS} of Theorem~\ref{thm:prolongation}}

We now prove part \eqref{enum:rightpartinSmS} of Theorem~\ref{thm:prolongation}.
The tableau $\rightpart(T)$ only contains entries $i_\ell$ and no entries $j_k^i$ or $j_{\overline k}^i$, see \eqref{eq:directreplacement}.
As also seen in \eqref{eq:directreplacement},
if $\rightpart(T)$ contains an entry $i_\ell$, then the corresponding entry of $S$ is $i$.
Therefore $\varphi^{\circ}(S) = \varphi(\rightpart(T))$,
where we lifted the map $\varphi^{\circ} : I \to \{1,\ldots,m\}$ to a map with the same name that is defined on tableaux with entries from $I$.
Claim~\ref{cla:differentphiNEW} proves property \eqref{enum:rightpartinSmS} of Theorem~\ref{thm:prolongation}.

\subsection*{Proof of part~\eqref{enum:duplex} of Theorem~\ref{thm:prolongation}}

The rest of this section is devoted to proving part~\ref{enum:duplex} of Theorem~\ref{thm:prolongation}.
A rectangular tableau whose columns all coincide is called \emph{uniform}.
In the following proof we will crucially use that a uniform tableau with an even number of columns is duplex.
Indeed, we prove part~\ref{enum:duplex} of Theorem~\ref{thm:prolongation}
by showing that if $\varphi(T)$ is regular, then:
\begin{itemize}
 \item[(I)] $\varphi(\check B_{e})$ is uniform if $e$ does not contain any link vertex $\zeta^{(i)}$,
 \item[(II)] $\varphi(\check B_{e})$ is uniform if $\zeta^{(i)} \in e$ for $i \neq h$, and
 \item[(III)] $\varphi(\check B_{e})$ is uniform if $\zeta^{(h)} \in e$.
\end{itemize}
It is clear that these three properties cover all cases and hence $\varphi(T)$ is uniform by construction \eqref{eq:leftpartmadefromblocks}.
This implies part~\ref{enum:duplex} of Theorem~\ref{thm:prolongation}.

We start with proving (I).
\begin{claim}\label{cla:ithentry}
Let $\varphi(T)$ be regular.
Given a block edge $e$ in $H^{(i)}$.
For all $v \in e$, $v \neq \zeta^{(i)}$, we have that
the $i$-th entry of $\varphi(\check B_v)$ is $\varphi^{\circ}(i)$.
\end{claim}
\begin{proof}
Combine \eqref{eq:iBv} and Claim~\ref{cla:varphii}.
\end{proof}
\begin{claim}\label{cla:jBvneqzeta}
Let $\varphi(T)$ be regular.
Given a block edge $e$ in $H^{(i)}$.
For all $j \neq i$ we have that the set
\[
\{
\text{$j$-th entry of $\varphi(\check B_v)$} \mid v \in e, v \neq \zeta^{(i)}
\}
\]
consists of the single element $\varphi(j_{k(e)}^{i})$.
\end{claim}
\begin{proof}
Combine \eqref{eq:jBv} and Claim~\ref{cla:barcoincide}.
\end{proof}
Combining Claim~\ref{cla:ithentry} and Claim~\ref{cla:jBvneqzeta} we see that
(I) is true.

We now prove (II). Let $i \neq h$ and let $e$ be the block edge in $H^{(i)}$ that contains $\zeta^{(i)}$. Note that $k(e)=1$.
\begin{claim}\label{cla:almostexceptionalcolumn}
Let $\varphi(T)$ be regular.
Then  $\varphi(\check B_{\zeta^{(i)}})$ coincides with $\varphi(\check B_v)$, $v \in e$, $i \neq h$.
\end{claim}
\begin{proof}
We compare the columns entrywise.
Note that $k(v)=k(\zeta^{(i)})=1$.
We make a case distinction.

Case 1: Let $j \neq i$.
The $j$-th entry of $\check B_v$ is either $j_{1}^i$ or $j_{\overline 1}^i$, see \eqref{eq:jBv}.
The $j$-th entry of $\check B_{\zeta^{(i)}}$ is either $j_{1}^i$ or $j_{\overline 1}^i$, see \eqref{eq:jBzetaheq}.
Using Claim~\ref{cla:barcoincide} we see that the $j$-th entry of $\varphi(\check B_v)$ equals the $j$-th entry of $\varphi(\check B_{\zeta^{(i)}})$.

Case 2:
The $i$-th entry of $\check B_v$ is $i_1$, see \eqref{eq:iBv}.
The $i$-th entry of $\check B_{\zeta^{(i)}}$ is $i_1^h$, see \eqref{eq:iBzetaioneh}. Hence Claim~\ref{cla:droph} implies that
the $i$-th entry of $\varphi(\check B_v)$ equals the $i$-th entry of $\varphi(\check B_{\zeta^{(i)}})$.
\end{proof}
It follows from Claim~\ref{cla:almostexceptionalcolumn} that all columns in $\varphi(\check B_e)$ coincide, i.e., $\varphi(\check B_e)$ is uniform.
Thus (II) is proved.

It remains to show (III), i.e., that $\varphi(\check B_{e})$ is uniform if $\zeta^{(h)} \in e$.

\begin{claim}\label{cla:hthentry}
Let $\varphi(T)$ be regular and $\zeta^{(h)}$ the link vertex in the block edge $e$.
For all $v \in e$, $v \neq \zeta^{(h)}$, we have that
the $h$-th entry of $\varphi(\check B_v)$ is $\varphi^{\circ}(h)$.
\end{claim}
\begin{proof}
This is a direct implication of Claim~\ref{cla:ithentry}.
\end{proof}

\begin{claim}\label{cla:exceptionalcolumn}
Let $\varphi(T)$ be regular and $\zeta^{(h)} \in e$.
Then $\varphi(\check B_{\zeta^{(h)}})$ coincides with $\varphi(\check B_{v})$,
$v \in e$.
\end{claim}
\begin{proof}
We compare the columns entrywise, considering three cases.

Case 1: We compare the $h$-th entry:
According to Claim~\ref{cla:hthentry}, the $h$-th entry of $\varphi(\check B_{v})$ is $\varphi^{\circ}(h)$. According to \eqref{eq:hBzetah}
the $h$-th entry of $\check B_{\zeta^{(h)}}$ is $h_1$, so
the $h$-th entry of $\varphi(\check B_{\zeta^{(h)}})$ is $\varphi(h_1)=\varphi^{\circ}(h)$, see \eqref{eq:defvarphii}.

Case 2: We compare the $j$-th entry, $j \neq h$, in the case $j \notin I$:
According to \eqref{eq:jBv}, the $j$-th entry of $\check B_{v}$ is either $j_1^{h}$ or $j_{\overline 1}^{h}$.
The $j$-th entry of $\check B_{\zeta^{(h)}}$ is $j_1^{h}$, see \eqref{eq:jBzetaheq}. Therefore Claim~\ref{cla:barcoincide} shows that
the $j$-th entry of $\varphi(\check B_{v})$ equals the $j$-th entry of $\varphi(\check B_{\zeta^{(h)}})$.

Case 3: We compare the $j$-th entry, $j \neq h$, in the case $j\in I$:
According to \eqref{eq:jBv}, the $j$-th entry of $\check B_{v}$ is either $j_1^{h}$ or $j_{\overline 1}^{h}$.
The $j$-th entry of $\check B_{\zeta^{(h)}}$ is $j_1$, see \eqref{eq:jBzetahneq}.
Combining Claim~\ref{cla:barcoincide} and Claim~\ref{cla:droph} shows that
the $j$-th entry of $\varphi(\check B_{v})$ equals the $j$-th entry of $\varphi(\check B_{\zeta^{(h)}})$.
\end{proof}
It follows from Claim~\ref{cla:exceptionalcolumn} that all columns in $\varphi(\check B_e)$ coincide, i.e., $\varphi(\check B_e)$ is uniform.
Thus (III) is proved.
This finishes the proof of part~\ref{enum:duplex} of Theorem~\ref{thm:prolongation}.

Theorem~\ref{thm:prolongation} is now completely proved for odd $D$.

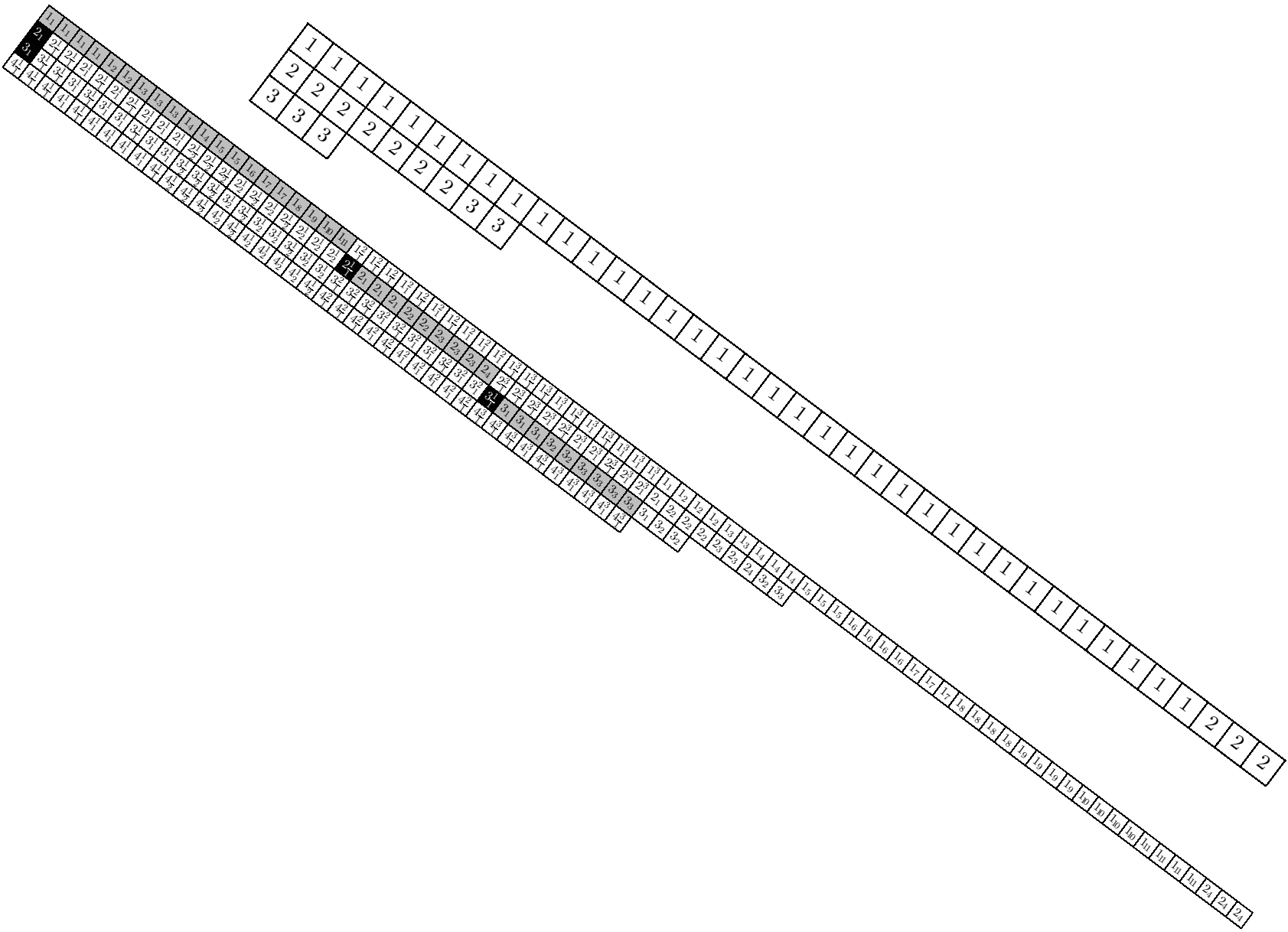
\begin{figure}
\scalebox{0.95}{
\begin{tikzpicture}
 \node at (0,4) {
\rotatebox{53}{
\ytableaushort{
11111111111111111111111111111111111222,
222222233,
333
}
}
};
 \node at (0,0) {
\scalebox{0.8}{\rotatebox{53}{
\begin{minipage}{30cm}
\ytableausetup{boxsize=1.5em}
\scalebox{0.75}{
\ytableaushort{
{*(lightgray)1_1}{*(lightgray)1_1}{*(lightgray)1_1}{*(lightgray)1_1}{*(lightgray)1_2}{*(lightgray)1_2}{*(lightgray)1_3}{*(lightgray)1_3}{*(lightgray)1_3}{*(lightgray)1_4}{*(lightgray)1_4}{*(lightgray)1_5}{*(lightgray)1_5}{*(lightgray)1_6}{*(lightgray)1_7}{*(lightgray)1_7}{*(lightgray)1_8}{*(lightgray)1_9}{*(lightgray)\text{$1_{\!1\!0}$}}{*(lightgray)\text{$1_{\!1\!1}$}}
{1^2_{\A}}{1^2_{\A}}{1^2_{\A}}{1^2_1}{1^2_{\A}}{1^2_1}{1^2_{\A}}{1^2_1}{1^2_1}{1^2_1}
{1^3_{\A}}{1^3_{\A}}{1^3_{\A}}{1^3_1}{1^3_{\A}}{1^3_1}{1^3_{\A}}{1^3_1}{1^3_1}{1^3_1}
{1_1}{1_2}{1_2}{1_2}{1_3}{1_3}{1_4}{1_4}{1_4}{1_5}{1_5}{1_5}{1_6}{1_6}{1_6}{1_6}{1_7}{1_7}{1_7}{1_8}{1_8}{1_8}{1_8}{1_9}{1_9}{1_9}{1_9}{\text{$1_{\!1\!0}$}}{\text{$1_{\!1\!0}$}}{\text{$1_{\!1\!0}$}}{\text{$1_{\!1\!0}$}}{\text{$1_{\!1\!1}$}}{\text{$1_{\!1\!1}$}}{\text{$1_{\!1\!1}$}}{\text{$1_{\!1\!1}$}}
{2_4}{2_4}{2_4}
,
{*(black)\textcolor{white}{2_1}}{2^1_{\A}}{2^1_{\A}}{2^1_1}{2^1_{\A}}{2^1_1}{2^1_{\A}}{2^1_1}{2^1_1}{2^1_1}
{2^1_{\B}}{2^1_{\B}}{2^1_{\B}}{2^1_2}{2^1_{\B}}{2^1_2}{2^1_{\B}}{2^1_2}{2^1_2}{2^1_2}
{*(black)\textcolor{white}{2^1_{\A}}}{*(lightgray)2_1}{*(lightgray)2_1}{*(lightgray)2_1}{*(lightgray)2_2}{*(lightgray)2_2}{*(lightgray)2_3}{*(lightgray)2_3}{*(lightgray)2_3}{*(lightgray)2_4}
{2^3_{\A}}{2^3_{\A}}{2^3_{\A}}{2^3_1}{2^3_{\A}}{2^3_1}{2^3_1}{2^3_{\A}}{2^3_1}{2^3_1}
{2_1}{2_2}{2_2}{2_2}{2_3}{2_3}{2_4}{3_2}{3_3}
,
{*(black)\textcolor{white}{3_1}}{3^1_{\A}}{3^1_{\A}}{3^1_1}{3^1_{\A}}{3^1_1}{3^1_1}{3^1_{\A}}{3^1_1}{3^1_1}
{3^1_{\B}}{3^1_{\B}}{3^1_{\B}}{3^1_2}{3^1_{\B}}{3^1_2}{3^1_2}{3^1_{\B}}{3^1_2}{3^1_2}
{3^2_{\A}}{3^2_{\A}}{3^2_{\A}}{3^2_1}{3^2_{\A}}
{3^2_1}{3^2_1}{3^2_{\A}}{3^2_1}{3^2_1}
{*(black)\textcolor{white}{3^1_{\A}}}{*(lightgray)3_1}{*(lightgray)3_1}{*(lightgray)3_1}{*(lightgray)3_2}{*(lightgray)3_2}{*(lightgray)3_3}{*(lightgray)3_3}{*(lightgray)3_3}{*(lightgray)3_3}
{3_1}{3_2}{3_2}
,
{4^1_{\A}}{4^1_{\A}}{4^1_{\A}}{4^1_1}{4^1_{\A}}{4^1_1}{4^1_1}{4^1_1}{4^1_1}{4^1_{\A}}
{4^1_{\B}}{4^1_{\B}}{4^1_{\B}}{4^1_2}{4^1_{\B}}{4^1_2}{4^1_2}{4^1_2}{4^1_2}{4^1_{\B}}
{4^2_{\A}}{4^2_{\A}}{4^2_{\A}}{4^2_1}{4^2_{\A}}{4^2_1}{4^2_1}{4^2_1}{4^2_1}{4^2_{\A}}
{4^3_{\A}}{4^3_{\A}}{4^3_{\A}}{4^3_1}{4^3_{\A}}{4^3_1}{4^3_1}{4^3_1}{4^3_1}{4^3_{\A}}
}
}
\end{minipage}
}
}
};
\end{tikzpicture}
}
\caption{A full example of a tableau $S$ (on top) and the corresponding tableau $T$ (below). Here $D = 5$, $m=4$.}
\label{fig:fullexample}
\end{figure}


\end{document}